\documentclass[sigconf]{acmart}

\usepackage[utf8]{inputenc} 
\usepackage[T1]{fontenc}    
\usepackage{hyperref}       
\usepackage{url}            
\usepackage{booktabs}       
\usepackage{amsfonts}       
\usepackage{nicefrac}       
\usepackage{microtype}      
\usepackage{xcolor}         

\usepackage{xspace}
\usepackage{graphicx}
\usepackage{subfigure}
\usepackage{wrapfig}
\usepackage{graphicx}
\usepackage{subfigure}
\usepackage{url}
\usepackage{color}
\usepackage{multirow}
\usepackage{gensymb}
\usepackage{longtable}
\usepackage{tikz}
\usepackage{comment}
\usepackage{dsfont}
\usepackage{framed}
\usepackage{enumitem}
\usepackage{xspace}
\usepackage{bm}
\usepackage{algorithm}
\usepackage{algorithmic}

\usepackage{amsmath}
\usepackage{mathtools}
\usepackage{amsthm}
\usepackage{comment}
\usepackage{enumitem}
\usepackage{nicefrac}
\usepackage{tikz}
\usetikzlibrary{patterns}
\usepackage{pifont}

\newcommand{\ouralg}{\texttt{LAOC}\xspace}

\newcommand{\opt}{\texttt{OPT}\xspace}
\newcommand{\robd}{\texttt{ROBD}\xspace}
\newcommand{\mpc}{\texttt{MPC}\xspace}
\newcommand{\mpclstm}{\texttt{MPC-LSTM}\xspace}
\newcommand{\tmpc}{\texttt{TMPC}\xspace}
\newcommand{\rl}{\texttt{ML}\xspace}
\newcommand{\crl}{\texttt{CRL}\xspace}
\newcommand{\lin}{\texttt{Lin}\xspace}
\newcommand{\linplus}{\texttt{Lin+}\xspace}
\newcommand{\ogd}{\texttt{OGD}\xspace}

\theoremstyle{plain}
\newtheorem{theorem}{Theorem}[section]
\newtheorem{proposition}[theorem]{Proposition}
\newtheorem{lemma}[theorem]{Lemma}

\newtheorem{assumption}[theorem]{Assumption}
\theoremstyle{remark}

\AtBeginDocument{%
  }

\setcopyright{acmlicensed}
\copyrightyear{2018}
\acmYear{2018}
\acmDOI{XXXXXXX.XXXXXXX}

\acmConference[Conference acronym 'XX]{the 16th ACM International Conference
on Future and Sustainable Energy Systems (e-Energy 2025)}{June 17 - 20,
  2025}{Rotterdam, Netherlands}
\acmISBN{978-1-4503-XXXX-X/18/06}

\begin{document}

\title{Learning-Augmented Online Control for \\ Decarbonizing Water Infrastructures}

\author{Jianyi Yang}
\affiliation{%
	\institution{University of Houston}
 \country{}
 \city{}
}
\email{jyang66@uh.edu}

\author{Pengfei Li}
\affiliation{%
  \institution{University of California, Riverside}
  \country{}
  \city{}
}
\email{pli081@ucr.edu}

\author{Tongxin Li}
\affiliation{%
  \institution{Chinese University of Hongkong, Shenzhen}
  \country{}
  \city{}
}
\email{litongxin@cuhk.edu.cn}

\author{Adam Wierman}
\affiliation{%
\institution{California Institute of Technology}
 \country{}
 \city{}
}
\email{adamw@caltech.edu}

\author{Shaolei Ren}
\affiliation{%
	\institution{University of California, Riverside}
 \country{}
 \city{}
}
\email{shaolei@ucr.edu}

\renewcommand{\shortauthors}{Yang et al.}

\begin{abstract}
 Water infrastructures are essential for drinking water supply, irrigation, fire protection, and other critical applications. However, water pumping systems, which are key to transporting water to the point of use, consume significant amounts of energy and emit millions of tons of greenhouse gases annually. With the wide deployment of digital water meters and sensors in these infrastructures, Machine Learning (ML) has the potential to optimize water supply control and reduce greenhouse gas emissions. Nevertheless, the inherent vulnerability of ML methods in terms of worst-case performance raises safety concerns when deployed in critical water infrastructures. To address this challenge, we propose a learning-augmented online control algorithm, termed \ouralg, designed to dynamically schedule the activation and/or speed of water pumps. To ensure safety, we introduce a novel design of safe action sets for online control problems. By leveraging these safe action sets, \ouralg can provably guarantee safety constraints while utilizing ML predictions to reduce energy and environmental costs. Our analysis reveals the tradeoff between safety requirements and average energy/environmental cost performance.
Additionally, we conduct an experimental study on a building water supply system to demonstrate the empirical performance of \ouralg. The results indicate that \ouralg can effectively reduce environmental and energy costs while guaranteeing safety constraints.
\end{abstract}

\maketitle

\section{Introduction}

Water supply is a critical utility for numerous infrastructures, including residential and commercial buildings, manufacturing facilities, and data centers. Globally, water systems consume about $4\%$ of the total electricity use \cite{IEA_energy_and_water}. In municipalities, energy consumption of water systems typically accounts for approximately $30\%$ to $40\%$ of the total electricity use \cite{energy_efficiency_water_supply}. In the United States alone, the energy costs associated with water infrastructure amount to around 4 billion annually and contribute over 45 million tons of greenhouse gases \cite{energy_efficiency_water_supply}. Pumping is usually the most energy-intensive part of water infrastructures, representing up to $80\%$ of the energy consumed by municipal water systems \cite{saving_energy_public_water}. This significant energy consumption has spurred widespread interest in optimizing water pump systems to reduce both greenhouse gas emissions and monetary costs \cite{IEA_energy_and_water,energy_efficiency_water_supply,saving_energy_public_water}.

In most critical infrastructures, water supply systems use storage tanks to ensure a reliable water provision. Pumps are employed to maintain adequate water levels of these tanks to meet water demand. Beyond providing a reliable water supply, these tanks can serve as buffers that can be exploited to manage pumping systems more efficiently, thereby reducing greenhouse gas emissions and monetary costs. With the integration of renewable energy, both carbon intensity and electricity prices fluctuate over time \cite{predicting_carbon_intensity_sun2022predictions,electricity_markets}. This time-varying property, combined with the widespread deployment of sensors, allows water supply systems to dynamically schedule the activation and/or the speed of pumps with the goal of optimizing carbon/energy efficiency \cite{Water_Pump_Operation_eEnergy24,pump_control_optimization}.  Importantly, the scheduling policy should ensure safe water levels of the tanks to address any emergencies.

Water supply management is an online control problem characterized by time-varying dynamics and cost functions that are revealed sequentially to the pump controller. Such problems are challenging due to the uncertainty of future contexts including demand, carbon intensity, and/or energy prices \cite{L2O_OnlineBipartiteMatching_Toronto_ArXiv_2021_DBLP:journals/corr/abs-2109-10380,L2O_OnlineResource_PriceCloud_ChuanWu_AAAI_2019_10.1609/aaai.v33i01.33017570,learning_OBM_ICML23,LAAU_yang2022learning}. Without precise knowledge of the future contexts, the controllers of pumping systems are difficult to achieve high energy efficiency. Nevertheless, 
exploiting the data of water usage, carbon intensity and energy price, machine learning (ML) can be applied to overcome the uncertainties inherent in online control, often surpassing the performance of manually designed policies \cite{renewable_aggregation_li2021learning,online_optimal_control_li2019online,SOCO_Prediction_Error_RHIG_NaLi_Harvard_NIPS_2020_NEURIPS2020_a6e4f250,Shaolei_L2O_ExpertCalibrated_SOCO_SIGMETRICS_2022}. Recently, ML predictions have been utilized in water supply systems to enhance cost savings and carbon efficiency \cite{Water_Pump_Operation_eEnergy24,MPC_real_time_water_operation_wang2021minimizing,MPC_energy_water_management_wanjiru2016model}.

However, ML can sometimes provide inaccurate predictions or low-quality advice, which can lead to arbitrarily poor performance and raise safety concerns for critical water infrastructures. For instance, a water tank in a conference center is crucial for ensuring a reliable water supply and fire protection. If the controller fails to maintain a safe water level, serious accidents can occur in the event of a municipal distribution system fault or a fire emergency. Naive deployments of ML-based controllers could result in such failures, leading to significant safety risks. Despite significant efforts to improve ML models for water supply systems \cite{Water_Pump_Operation_eEnergy24,uncertainty_aware_water_management_sopasakis2018uncertainty,chance_constrained_pumpung_stuhlmacher2020chance}, ML-based controllers fundamentally lack performance guarantees, especially for adversarial or out-of-distribution problem instances. Such lack of performance guarantees hinders the deployment of ML in real-world critical infrastructures.

To solve the fundamental challenges of ensuring worst-case performance guarantees for ML-based controllers, we propose a method that leverages control priors. Control priors are human-crafted online algorithms with provable worst-case performance guarantees \cite{competitive_control_memory_shi2020online,CompetitiveControl_GuanyaShi_CISS_2021_9400281,CompetitiveControl_Hassibi_arXiv_2021_https://doi.org/10.48550/arxiv.2107.13657,SOCO_OBD_LQR_Abstract_Goel_Adam_Caltech_2019_10.1145/3374888.3374892} or trusted rule-based heuristics that have been reliably used in real systems for a long time \cite{Cooling_control_luo2022controlling,RL_cooling_chervonyi2022semi}. These control priors are highly reliable in terms of safety metrics. By integrating these priors into ML-based controllers, we aim to develop an algorithm that ensures the safety performance of the ML-based controller is no worse than a the safety performance benchmark. Drawing on the concept of learning-augmented algorithms that incorporate ML advice into algorithm design, we call our proposed algorithm Learning-Augmented Online Control (\ouralg). 

While initially developed for water systems, the proposed algorithm (\ouralg) is versatile and can be applied to various practical online control and resource management problems, such as battery management for electric vehicle (EV) charging station\cite{EV_charging_sun2021data}, workload scheduling for sustainable data centers \cite{RL_resource_provisioning_salahuddin2016reinforcement}, and control of cooling systems \cite{Cooling_control_luo2022controlling}. Adaptation of \ouralg to these applications can improve the average performance while providing a worst-case performance guarantee.

\textbf{Contributions.} The contributions of the paper are summarized as follows. First, it presents an online control framework designed to sustainably and safely manage water supply for critical infrastructures. The framework addresses the urgent need for a worst-case safety risk guarantee in decarbonizing critical infrastructures. Notably, this framework extends to various online control and resource management problems across different critical infrastructures.
Central to the paper's contribution is the development of a novel learning-augmented algorithm named \ouralg, which integrates a control prior into the ML-based controller to ensure worst-case safety risk constraints while optimizing decarbonization performance. Our analysis demonstrates that the proposed method reliably satisfies safety performance constraints for any problem instance while effectively leveraging ML predictions for decarbonization and cost saving. Furthermore, our analysis illuminates the tradeoff between the decarbonization and cost saving performance and the worst-case safety guarantee.
Lastly, the paper evaluates the proposed algorithm for the water supply system of critical buildings. Results indicate that \ouralg achieves significant carbon reduction and cost savings compared to traditional controllers used in water supply systems focusing on maintaining water levels. Moreover, it showcases the advantage of \ouralg in guaranteeing worst-case safety performance compared to pure ML-based algorithms.

\section{Related Work}\label{sec:related}

\textbf{Optimization of water supply systems.}
The considered problem stems from the tradition field of water supply management. In this area, a lot of works consider the scheduling for water distribution systems \cite{optimal_scheduling_water_singh2019optimal,chance_constrained_water_distribution_stuhlmacher2020water,water_network_optimization_d2015mathematical,optimal_operation_oikonomou2018optimal}. Some works have developed the pump control methods to maintain a water level for demand satisfaction and save energy, which has been studied in \cite{MPC_energy_water_management_wanjiru2016model,MPC_real_time_water_operation_wang2021minimizing,pump_control_optimization,pump_scheduling_water_networks,luna2019improving,Water_Pump_Operation_eEnergy24,uncertainty_aware_water_management_sopasakis2018uncertainty,chance_constrained_pumpung_stuhlmacher2020chance}.  Most of these works only consider the energy price, but do not explicitly consider the dynamical carbon intensity. The carbon emission of water infrastructures has recently become a crucial social concern \cite{energy_efficiency_water_supply,sustainable_water_infrastructure}, so we include the carbon emissions in the optimization objective to ensure sustainable operation. 

Much of the literature, e.g.,  \cite{MPC_energy_water_management_wanjiru2016model,MPC_real_time_water_operation_wang2021minimizing,Water_Pump_Operation_eEnergy24}, utilizes ML predictions of the future demand and/or energy price to improve the control performance. To fight against the future uncertainty, some works have developed robust control algorithms  or constrained control algorithms for water supply systems \cite{chance_constrained_water_distribution_stuhlmacher2020water,optimal_scheduling_water_singh2019optimal, robust_optimization_water_ghelichi2018novel,robust_energy_optimization_WDS_goryashko2014robust}. They either satisfy the safety constraints with a large probability or provide no guarantee on safety constraints. However,  it is critically needed for water infrastructures to guarantee the worst-case safety performance of water supply given any problem instance. In this paper, we solve this challenge by designing a novel learning-augmented control algorithm utilizing the trusted control prior.

\textbf{Online control.}
Our problem formulation is relevant to the literature of online competitive control.  
In our problem setting, the target is to minimize the cumulative cost in the nonlinear dynamics, which is
different from the traditional control literature that uses  measures for stabilization purposes \cite{polycarpou1993robust,franco2006robust,freeman2008robust,khalil1996robust}. 
Like the recent works on competitive control~\cite{competitive_control_goel2022competitive,goel2022best,competitive_control_yu2022competitive,SOCO_OBD_LQR_Abstract_Goel_Adam_Caltech_2019_10.1145/3374888.3374892,competitive_control_memory_shi2020online,SOCO_OBD_R-OBD_Goel_Adam_NIPS_2019_NEURIPS2019_9f36407e, SOCO_Memory_FeedbackDelay_Nonlinear_Adam_Sigmetrics_2022_10.1145/3508037}, our work considers guarantees on the worst-case competitiveness, but our main focus is different --- we leverage ML to explore policies with low average cost while enforcing competitiveness guarantees for any step in any episode. This enables the use of the existing competitive control policies as priors. Achieving our objective requires novel design of safe action sets and new analysis techniques to find the trade-off between the average performance and worst-case competitiveness.

\textbf{Learning-based online control}.
Our algorithm is relevant to the broad area of learning-based control
\cite{learning_MPC_hewing2020learning,safe_learning_in_robotics_brunke2022safe, data_driven_control_tang2022data, learning_tube_MPC_fan2020deep,learning_uncertainty_set_lilearning,L1_GP_gahlawat2020l1,Control_RobustConsistency_LQC_TongxinLi_Sigmetrics_2022_10.1145/3508038,Online_untrusted_predictions_rutten2022online}. These works have developed machine learning models to predict the system dynamic or control-relevant information which is utilized in deciding the control actions \cite{data_driven_control_tang2022data,robust_learning_MPC_aswani2013provably,learning_uncertainty_set_lilearning,L1_GP_gahlawat2020l1,learning_tube_MPC_fan2020deep,Deep_adaptive_control_joshi2019deep,Asynchronous_deep_adaptive_control_joshi2021asynchronous}. Recent works combine learning-based methods with system models in order to improve the safety or robustness of learning for control \cite{safe_learning_in_robotics_brunke2022safe,learning_tube_MPC_fan2020deep,L4safety_critical_control_taylor2020learning,Control_RobustConsistency_LQC_TongxinLi_Sigmetrics_2022_10.1145/3508038,Online_untrusted_predictions_rutten2022online,wabersich2021probabilistic}.
Among them, learning-augmented online algorithms
combine potentially untrusted ML predictions with
robust policies (i.e., control priors). 
Learning-augmented algorithms have been developed for online control/optimization by combining ML predictions and control priors through online switching \cite{Online_untrusted_predictions_rutten2022online,SOCO_MetricUntrustedPrediction_Google_ICML_2020_pmlr-v119-antoniadis20a} or adaptively setting a confidence on the ML prediction \cite{Control_RobustConsistency_LQC_TongxinLi_Sigmetrics_2022_10.1145/3508038,renewable_aggregation_li2021learning}. 
Compared to these studies, we make contributions by considering
a more challenging setting, i.e., non-linear and time-varying dynamic models that are sequentially revealed online. 
Although some of the existing  studies \cite{Control_RobustConsistency_LQC_TongxinLi_Sigmetrics_2022_10.1145/3508038,Online_untrusted_predictions_rutten2022online,optimal_christianson2023,expert_robustified_learning_infocom2023,Shaolei_L2O_ExpertCalibrated_SOCO_SIGMETRICS_Journal_2022} provide provable cost bounds, they cannot guarantee a flexible any-step safety constraint given an arbitrary control prior, but this is needed for real problems \cite{Cooling_control_luo2022controlling}.

\textbf{Safe/Constrained Reinforcement Learning}
Our algorithm is also relevant to the literature of safe/constrained Reinforcement Learning (RL). Some safe/constrained RL works focus on discrete actions and their regret scales with the size of action set \cite{constrained_MDP_efroni2020exploration,wei2021provably,liu2021learning} while others \cite{Conservative_RL_Bandits_LiweiWang_SimonDu_ICLR_2022_yang2022a,Conservative_RL_Bandits_LiweiWang_SimonDu_ICLR_2022_yang2022a,Conservative_ConstrainedPolicyOptimization_achiam2017constrained,Conservative_ProjectBasedConstrainedPolicyOptimizatino_ICLR_2020_Yang2020Projection-Based} apply to the continuous control problems. 
However, most of them only satisfy the constraints in expectation or with a high probability \cite{Conservative_RL_Bandits_LiweiWang_SimonDu_ICLR_2022_yang2022a,Conservative_RL_Bandits_LiweiWang_SimonDu_ICLR_2022_yang2022a,Conservative_ConstrainedPolicyOptimization_achiam2017constrained,Conservative_ProjectBasedConstrainedPolicyOptimizatino_ICLR_2020_Yang2020Projection-Based,SafeRL_AlmostSureViolationConstraint_JHU_2022_castellano2021reinforcement,SafeRL_LinearFunction_YangLin_UCLA_ICML_2021_pmlr-v139-amani21a,safe_exploartion_primal_dual_ding2021provably,constrained_RL_linear_approximation_ghosh2022provably,constrained_MDP_efroni2020exploration,constrained_RL_primal_dual_ding2020natural}. A recent work \cite{Saute_rl_sootla2022saute} tries to solve RL with safety constraints satisfied almost surely, but no theoretical constraint satisfaction is guaranteed. When these algorithms are applied to online control systems like water supply management, the safety constraints can still be violated for some adversarial sequences.
By contrast, our algorithm exploits the control priors and provides a theoretical guarantee for the safety constraint satisfaction.

\section{Problem Formulation}
\begin{figure}
	\centering
 \vspace{-0.4cm}
 \includegraphics[width=0.48\textwidth]{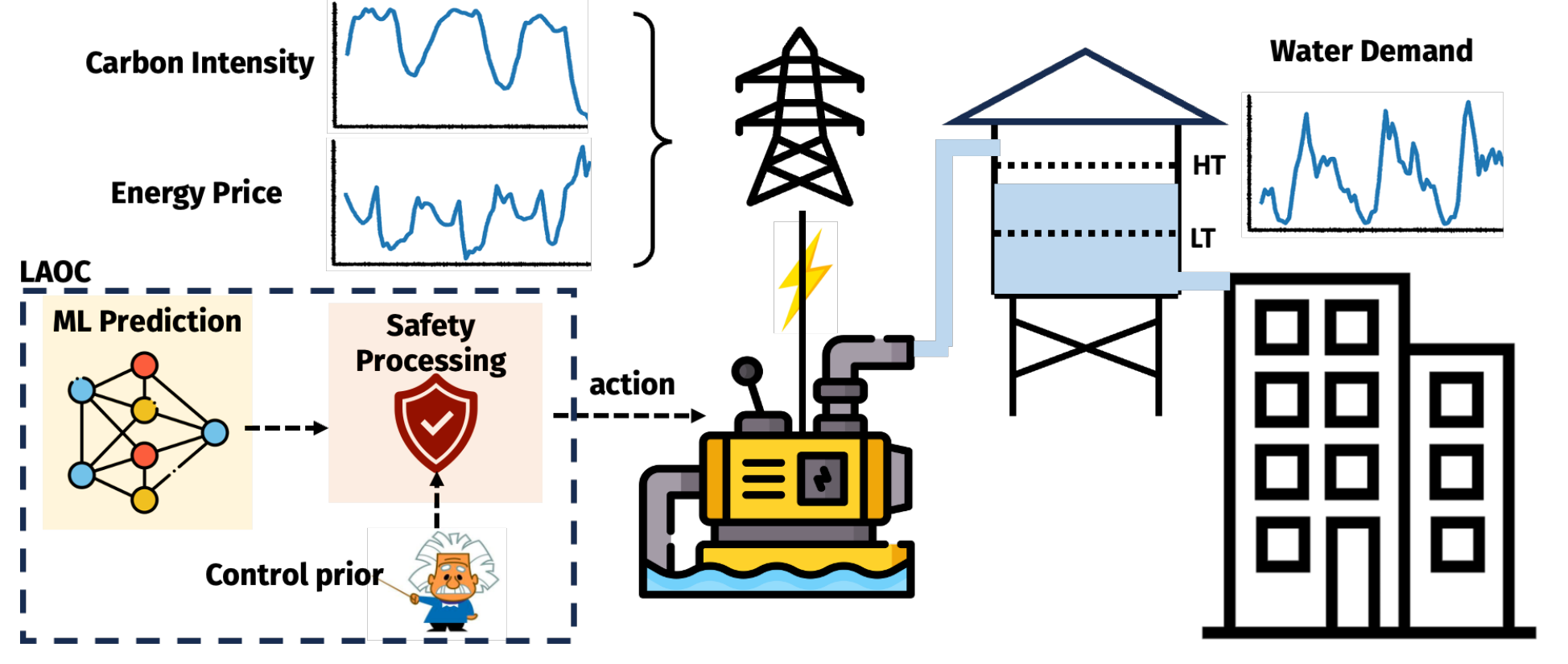}
\caption{Water Supply Infrastructure with ML predictions.}
\vspace{-0.2cm}	
\end{figure} 

In this section, with the water supply management as the application scenario, we present the safe online control model. Next, we show the safe online control model applies to broader applications by specifying the dynamics, loss and risk functions. Finally, we give the assumptions on the dynamics and risk functions required for the algorithm design and analysis.
\subsection{Safe Online Control for Water Supply}\label{sec:water_supply_control_model}
In this section, we formulate an online control problem with time-varying costs and dynamics that captures the task of water supply management. A problem instance consists of \(H\) time slots. At the beginning of each time slot \(h \in [H]\), the controller observes the water level state \(x_{h}\) and decides on an action \(u_h \in \mathbb{R}^d\) to schedule the activation time and/or the speed of the pumps. This action incurs a non-negative carbon emission \(c_e(u_h, e_h)\) related to the carbon intensity \(e_h\), and a monetary cost \(c_p(u_h, p_h)\) related to the energy price \(p_h\).
Given the water level state \(x_h\) and the action \(u_h\), the system transitions to \(x_{h+1}\) at the end of slot \(h\) following the dynamic function \(f\) defined as:
\begin{equation}\label{eqn:dynamic}
x_{h+1} = f_h(x_{h}, u_h) = x_{h} + g(u_h) - w_h, \quad h = 1, \ldots, H,
\end{equation}
where \(w_h\) is the water consumption within time slot \(h\), and \(g\) maps the control signal \(u_h\) to the amount of water supply within time slot \(h\). Note that $g(u_h)$ is a linear function if we only control the activation time of pumping, and it is a nonlinear continuous function if we control the speed of pumps \cite{variable_speed_pump}. The water level state is expected to remain close to a nominal water level $\bar{x}$ in the water tanks. Deviation from this nominal water level incurs a penalty cost denoted as \(c_w(x_h)\).

For convenience, we denote \(y_h \coloneqq (e_h, p_h, w_h)\), and so \(y_{1:H} = (y_1, \ldots, y_H)\) is the information for the entire episode. The total loss at slot \(h\) is expressed as:
\begin{equation}
c_h(x_h, u_h) = \gamma_1 \cdot c_w(x_h) + \gamma_2 \cdot c_e(u_h, e_h) + \gamma_3 \cdot c_p(u_h, p_h),
\end{equation}
where \(\gamma_1\), \(\gamma_2\), and \(\gamma_3\) are weights used to convert the costs to the same measurement.
An online control policy, denoted by \(\pi\), outputs the action \(u_h\). The cumulative loss within an episode of \(H\) time slots, following policy \(\pi\), is expressed as:
The offline optimal loss is denoted by \(J_H^*\).

\textbf{Safety Constraint.} Online control algorithms must guarantee safety performance. For critical infrastructures, water supply management should maintain a safe water level to ensure reliable supply during emergencies. Failure to maintain a safe water level incurs a safety risk, denoted as \(r_h(x_h, u_h)\).  Given a nominal water level $\bar{x}$, a concrete form of safety risk can be denoted as 
\begin{equation}\label{eqn:risk_func}
r_h(x_h,u_h) = \gamma_w \cdot \mathrm{dist}(x_h,\bar{x})+ \gamma_b\cdot b(u_h),
\end{equation}
where $\mathrm{dist}(x_h,\bar{x})$ is a measure of distance between the water level $x_h$ and the nominal water level $\bar{x}$, $b(u_h)$ penalizes the power load of the scheduling action $u_h$, and $\gamma_w$ and $\gamma_b$ are balancing weights for the two risk metrics. Note that the distance function $\mathrm{dist}(x_h,\bar{x})$ can be an asymmetric function which provides different penalties for $x_h-\bar{x}\leq 0$ and $x_h-\bar{x}>0$. The asymmetric distance measure is flexible to model different penalties of overly-high and overly-low water levels. 
We define the total safety risk of a policy \(\pi\) over an episode with \(H\) rounds as \(R_H^{\pi} = \sum_{h=1}^H r_h(x_h, u_h)\).

To evaluate whether a controller is safe, we require a safety benchmark. In this paper, we use the scaled safety risk of an existing safe control prior \(\pi^{\dagger}\) as our benchmark. This means that for any problem instance \(y_{1:h}\) and \(h \in [H]\), the controller \(\pi\) must satisfy the safety constraint expressed as:
\begin{equation}\label{eqn:safety_constraint}
R_h^{\pi} \leq (1 + \lambda) R_h^{\pi^{\dagger}},
\end{equation}
where \(R_h^{\pi^{\dagger}}\) is the safety risk of the safe control prior \(\pi^{\dagger}\) and \(\lambda > 0\) is a preset parameter indicating the safety requirement level. The constraint in \eqref{eqn:safety_constraint} is called \textit{$(1+\lambda)-$safety}.

The intuition behind the safety constraint is that if the control prior has a worst-case safety performance guarantee for any instance, then a policy satisfying this constraint also ensures a performance guarantee adjustable by \(1 + \lambda\). This constraint must be satisfied in each round to provide a strong worst-case guarantee. The safe control prior can be a human-crafted algorithm with a theoretical worst-case performance guarantee or a reliable heuristic implemented in real systems for a long time. 
In water infrastructures, the control prior can be a traditional controller that is designed to maintain the safe water level \cite{MPC_energy_water_management_wanjiru2016model, MPC_real_time_water_operation_wang2021minimizing}.

\textbf{Objective}. 
We exploit ML predictions to optimize the expected loss while guaranteeing safety constraint for any problem instance. Given a safety requirement \(\lambda > 0\), the objective is:
\begin{equation}\label{eqn:constraint}
\begin{split}  &\min_{\pi}\mathbb{E}_{y_{1:H}}\left[J_H^{\pi}\right]\\
&\text{s.t.} \;\; R_h^{\pi} \leq (1+\lambda) R_h^{\pi^{\dagger}}, \;\; \forall h \in [H], \;\; \forall y_{1:H} \in \mathcal{Y}.
\end{split}
\end{equation}
For convenience, we define the collection of all control policies that satisfies the safety requirement with $\lambda$ as \(\Pi_{\lambda} = \{ \pi \mid R_h^{\pi} \leq (1+\lambda) R_h^{\pi^{\dagger}}, \forall h \in [H], \;\forall y_{1:H} \in \mathcal{Y}\}\). If \(\lambda\) is larger, then the size of \(\Pi_{\lambda}\) is also larger, providing more flexibility to optimize the average loss.
To solve this objective, we need to integrate the control prior \(\pi^{\dagger}\) into the ML-based controller. 

\subsection{Broader Applications}
While the formulation is specifically given for water supply management, it applies to many other online control problems by replacing the dynamic function $f_h$, the cost function $c_h$ and the risk function $r_h$ with concrete expressions. Here, we give the following two application examples.
\begin{itemize}
    \item \textbf{Battery management of EV charging station}. The battery management of Electrical Vehicle (EV) charging station is an online control problem where the agent needs to decide the amount of battery charging or discharging $u_h$ at each round to maintain a nominal State of Charge (SoC) $x_h$ that satisfies the charging demand \cite{Control_RobustConsistency_LQC_TongxinLi_Sigmetrics_2022_10.1145/3508038,EV_charging_data_lee2019acn}. In this problem, the dynamic of SoC $x_h$ is modeled by the dynamic function $f_h(x_h,u_h)=x_h+u_h-w_h$ where $w_h$ is the charging demand, the loss function $c_h$ defines the cost of charging and discharging. The risk function $r_h$ defines the risk of not satisfying the charging demand. Classic controllers \cite{classic_control_anderson2012optimal,power_of_predictions_yu2020power} can serve as the control prior $\pi^{\dagger}$ with risk performance guarantee.
    \item \textbf{Cooling control for sustainable data centers}.  In this application, the target of the data center agent is to maintain a temperature range with high carbon efficiency by making online decisions of cooling equipment management \cite{Cooling_control_luo2022controlling,hvac_optimizing_wong2022optimizing,chervonyi2022semi}. Failure to maintain a suitable temperature range will overheat the devices and render the risk of critical services denial. The dynamic function $f_h(x_h,u_h)$ models the temperature dynamic where $w_h$ is the randomness factor affecting the temperature change. The cost function $c_h$ captures the losses of carbon emission and energy costs. The risk function $r_h$ measures the risk of deviating from the normal temperature range. The traditional rule-based heuristics \cite{Cooling_control_luo2022controlling} that have verified performance in maintaining a suitable temperature can serve as the control prior $\pi^{\dagger}$.
\end{itemize}

\subsection{Assumptions}
In this paper, we assume the following conditions on the dynamic functions and the risk functions.

\begin{assumption}[Lipschitz dynamics]\label{assumption:lipschitz}
For each time \(h\), the function \(f_h\) is Lipschitz continuous with respect to \(x_h\) and \(u_h\) with Lipschitz constants \(\sigma_x \geq 0\) and \(\sigma_u \geq 0\), respectively, i.e., for any \((x, u)\) and \((x', u')\), \(f_h\) satisfies
\[
\begin{split}
&\left\| f_h(x, u) - f_h(x', u) \right\| \leq \sigma_x \| x - x' \| \\
&\left\| f_h(x, u) - f_h(x, u') \right\| \leq \sigma_u \| u - u' \|.
\end{split}
\]
\end{assumption}

\begin{assumption}[Well-conditioned risk functions]\label{assumption:smooth}
For each time \(h\), the risk function \(r_h\) is non-negative, \(\alpha\)-strongly convex, and \(\beta\)-smooth with respect to \((x_h, u_h)\).
\end{assumption}

The first assumption is the Lipschitz continuity of the dynamic functions, which is common in finite-horizon control models \cite{Control_RobustConsistency_LQC_TongxinLi_Sigmetrics_2022_10.1145/3508038,online_optimal_control_li2019online,competitive_control_yu2022competitive}. For water supply management, the dynamic function $f_h$ in \eqref{eqn:dynamic} is clearly Lipschitz continuous as $g$ is a Lipschitz continuous function.

The second assumption is the non-negativity, convexity and smoothness of the risk functions, which is a common regularity condition in control system costs \cite{online_optimal_control_li2019online,Control_PerturbationExponential_YihengLin_AdamWiermand_NIPS_2021_NEURIPS2021_298f5874,CompetitiveControl_GuanyaShi_CISS_2021_9400281,Decentralized_OCO_lin2022decentralized}. 
We are flexible to choose different risk functions that satisfy Assumption \ref{assumption:smooth}. For example, we can choose an asymmetric $\mathrm{dist}$ function as $\mathrm{dist}(x_h,\bar{x})=\gamma_{w,1}(\bar{x}-x_h)^2$ if $\bar{x}\geq x_h$ and $\mathrm{dist}(x_h,\bar{x})=\gamma_{w,2}(x_h-\bar{x})^2$ if $\bar{x}<x_h$ and a quadratic penalty of the power load, and the obtained risk function satisfies Assumption 3.2.

\section{Learning-Augmented Online Control (\ouralg)}\label{sec:potentil_competitive}

In this section, we present and analyze an algorithm, \ouralg, to solve the online control problem introduced in the previous section. Before stating the algorithm, we highlight the challenges created by the safety requirements.  

\subsection{Challenges Due to Safety Requirements}\label{sec:challenges}

Our goal is to find a policy satisfying safety constraints in \eqref{eqn:constraint} while exploiting the ML predictions to achieve a low loss. However, this is very challenging for online control where the future contexts are unknown to the agent. 

One might hope that a straightforward design that considers a linear combination of the control prior and the pure ML action (\lin) would be sufficient.  Formally, \lin is defined as $\pi=\rho \tilde{\pi} + (1-\rho) \pi^{\dagger}$, where $\tilde{\pi}$ is the pure ML policy and $\pi^{\dagger}$ is the policy prior. However, Proposition \ref{thm:necessary_competitiveness} shows that,
unless we completely ignore the ML policy (i.e.,
$\rho=0$),
\lin cannot guarantee $(1+\lambda)-$safety given any ML policy.

\begin{proposition}\label{thm:necessary_competitiveness}
Define the quality of pure ML as the normalized difference between the ML advice and the offline optimal action ${\|\tilde{u}-u^*\|^2}/{J_H^*}$. If
the pure ML have an arbitrarily low quality
(i.e., ${\|\tilde{u}-u^*\|^2}/{J_H^*}\to\infty$), \lin with $\rho\in (0,1]$ cannot guarantee $(1+\lambda)-$ safety for any finite $\lambda>0$.
\end{proposition}

Proposition \ref{thm:necessary_competitiveness} is proven by constructing a contradictory example that if $(1+\lambda)-$safety with finite $\lambda$ is satisfied by \lin with $\rho\in(0,1]$, the quality of ML advice $\frac{\|\tilde{u}-u^*\|^2}{J_H^*}$ must be bounded by a finite value. In
other words, $(1+\lambda)-$safety cannot be satisfied
by \lin with a potentially unsafe ML model in the worst case.

Overcoming the limitation of \lin with respect to safety guarantees requires a more flexible combination of pure ML and the control prior. 
Thus, we give a second natural approach maps the ML advice into a safe action set defined by the safety constraint for each round $h$ as
\begin{equation}\label{eqn:naive_design}
\overline{\mathcal{U}}_{\lambda,h}\!\!=\!\!\left\{u_h \mid R_h\leq(1+\lambda)R_h^{\pi^{\dagger}}\right\},
\end{equation}
where $R_h=\sum_{i=0}^h r_i(x_i,u_i)$ and $R_h^{\pi^{\dagger}}=\sum_{i=0}^h r_i(x_i^{\dagger},u_i^{\dagger})$.
The mapping can be a linear combination that selects the action as $u_h=\bar{\rho}_h \tilde{u}_h + (1-\bar{\rho}_h) u^{\dagger}_h$ where $\bar{\rho}_h=1$ if $\tilde{u}_h\in\overline{\mathcal{U}}_{\lambda,h}$ and $\bar{\rho}_h$ is the solution of $\smash{r_h(x_h,\bar{\rho}_h \tilde{u}_h + (1-\bar{\rho}_h) u^{\dagger}_h) = (1+\lambda)R_h^{\pi^{\dagger}}-R_{h-1}}$ if $\tilde{u}_h\notin\overline{\mathcal{U}}_{\lambda,h}$. We refer to this policy as \linplus.

\linplus uses a time-varying combination variable $\bar{\rho}_h$, so it is much more flexible than \lin and can strictly guarantee $(1+\lambda)$-safety given any instance as long as the safe action set in \eqref{eqn:naive_design} is non-empty. Unfortunately, the naive design of safe action set $\overline{\mathcal{U}}_{\lambda,h}$ in \eqref{eqn:naive_design} can be empty, which results in no feasible actions.  This is illustrated by the following example.

\textit{Example 4.1.} \textit{Suppose that $\sum_{i=0}^{h}r_i(x_i,u_i)=(1+\lambda)\sum_{i=0}^{h}r_i(x_i^{\dagger},u_i^{\dagger})$ is satisfied at time $h$. If $\smash{x_{h+1}=x_{h+1}^{\dagger}}$ holds at round $h+1$, the agent can always choose ${u_{h+1}=u_{h+1}^{\dagger}}$ to satisfy \eqref{eqn:naive_design} at round $h+1$. However, when $\smash{x_{h+1}\neq x_{h+1}^{\dagger}}$, 
it is possible that the control prior has a low loss for its state $x_{h+1}^{\dagger}$
at time $h+1$ such that for any action $u\in\mathcal{U}$ the true loss $c_{h+1}(x_{h+1},u)$ is lager than the scaled prior loss $\smash{(1+\lambda)c_{h+1}(x_{h+1}^{\dagger},u_{h+1}^{\dagger})}$. In such a case, the naive safe action set $\overline{\mathcal{U}}_{\lambda,h+1}$ is empty, and the control agent cannot maintain the
inequality in \eqref{eqn:naive_design}, thus potentially violating the subsequent safety constraints.}

The failures of the intuitive policies \lin and \linplus show that for a policy to combine the ML advice and the control prior, it must be flexible and conservative enough to guarantee that feasible actions exist to meet the safety constraints.  In the next section, we give the design that can theoretically guarantee the $(1+\lambda)$-safety for any sequence and ML advice. 

\subsection{Algorithm Design}\label{sec:competitivenessguarantee} 
In this section, we give an overview of the design of Learning Augmented Online Control (\ouralg). 

First, we highlight the design of the safe action set used in the algorithm. 
Instead of directly guaranteeing the inequality in \eqref{eqn:naive_design}, we ensure that the resulting cumulative loss satisfies $\sum_{i=0}^{h}r_i(x_i,u_i)+\phi_h\leq (1+\lambda)\sum_{i=0}^{h}r_i(x_i^{\dagger},u_i^{\dagger})$ with an added reservation $\phi_h\geq 0$ for hedging. 
With a proper design of the reservation, $\mathcal{U}_{\lambda,h}, h\in[h,H]$ can be guaranteed to be not empty for all the possible future control environments $y_{h:H}$.
To this end, we design a reservation in the next proposition,
whose proof is deferred to Appendix~\ref{sec:robustnessproof}.

\begin{proposition}\label{thm:potentaildesign_concrete}
Define a safe set $\mathcal{U}_{\lambda,h}, \lambda>0$ as 
\begin{equation}\label{eqn:robust_constraint}
\mathcal{U}_{\lambda,h}\!\coloneqq\!\left\{u_h \in \mathcal{U} \mid \!R_{h}+\phi_h(u_h)\leq (1+\lambda) R^{\pi^{\dagger}}_{h}\right\},
\end{equation}
where $R_{h}=
\sum_{i=0}^{h}r_i(x_i,u_i)$ and $R^{\pi^{\dagger}}_{h}=\sum_{i=0}^{h}r_i(x_i^{\dagger},u_i^{\dagger})$ are
the true loss and the loss of control prior, respectively. Moreover, 
{
\(\phi_h(u)\!\!=\!\!q_h\|x_{h+1}\!-\!x_{h+1}^{\dagger}\|^2\!\!=\!\!q_h\|f_h(x_{h},u)\!-\!f_h(x_{h}^{\dagger},u_h^{\dagger})\|^2\)} is a reservation function,
where $
q_h=C_1(1+\frac{1}{\lambda})\frac{\beta}{2}\sum_{h'=0}^{H-h-1}(C_2\sigma_x^{2})^{h'}
$
for constants $C_1\geq 1$ and $C_2\geq 1$.
With Assumptions~\ref{assumption:lipschitz} and~\ref{assumption:smooth}, if $\mathcal{U}_{\lambda,h-1}$ is not empty and the action $u_{h-1}$ at round $h-1$ is selected from the safe set $\mathcal{U}_{\lambda,h-1}$,
then $\mathcal{U}_{\lambda,h}$ is not empty and always includes $u_h^{\dagger}$. 
\end{proposition}

The key insight behind the formulation of the reservation  $\phi_h(u)$ in \eqref{eqn:robust_constraint}
is to hedge against the possible violation of safety constraints in future rounds.  If the resulting state difference $\|x_{h+1}-x_{h+1}^{\dagger}\|^2$ from choosing $u_h$ is greater, the possible loss difference $\sum_{i=h+1}^H r_i(x_i,u_i)-(1+\lambda)r_i(x_i^{\dagger},u_i^{\dagger})$ in the following rounds can also be greater in the worst case. Thus, the reservation $\phi_h(u)$ is designed as the scaled state difference to account for the worst-case future risk difference
between the true control policy and 
the control prior $\pi^{\dagger}$.

\begin{algorithm}[t]
\caption{Learning Augmented Online Control (\ouralg)}
 \begin{algorithmic}[1]
\label{alg:expert_robust_Q}
\REQUIRE{ML model $\tilde{\pi}$ and control prior $\pi^{\dagger}$}
 
\FOR{\text{time horizon }$h=0,\cdots,H-1$}

\STATE Observe state $x_h$, information $\{r_h,f_h\}$, and last-step context $w_{h-1}$.

\STATE  Update the policy prior's state
$x_{h}^{\dagger}=f_{h-1}(x_{h-1}^{\dagger},u_{h-1}^{\dagger}) + w_{h-1}$

\STATE Obtain an action $u_h^{\dagger}$ by the prior $\pi^{\dagger}$, and update prior risk $R_{h}^{\pi^{\dagger}}=R^{\pi^{\dagger}}_{h-1}+r_h(x_{h}^{\dagger},u^{\dagger}_h)$

\STATE Obtain the ML action $\tilde{u}_h$ via the ML model $\tilde{\pi}$

\STATE \textbf{if } the ML action $\tilde{u}_h\in \mathcal{U}_{\lambda,h}$ \textbf{then } take $u_h=\tilde{u}_h$ 

\STATE \textbf{else } take $u_h=m(\tilde{u}_h)$ \textbf{end if} {\color{gray}\textit{// Map to a safe action set $\mathcal{U}_{\lambda,h}$~\eqref{eqn:robust_constraint} by \eqref{eqn:projection} or \eqref{eqn:linear_combination}}}

\STATE Update true loss $J_{h}=J_{h-1}+c_h(x_{h},u_h)$ and risk $R_{h}=R_{h-1}+r_h(x_{h},u_h)$
\ENDFOR
\end{algorithmic}
\end{algorithm}

As a consequence of Proposition \ref{thm:potentaildesign_concrete}, if $u_h$ is selected from $\mathcal{U}_{\lambda,h}$ for each round $h$, there always exists a non-empty safe action set $\mathcal{U}_{\lambda, h}$ in the subsequent steps, and thus $(1+\lambda)-$safety is strictly satisfied for each round. Based on Proposition \ref{thm:potentaildesign_concrete}, given an ML policy $\tilde{\pi}$ and a control prior $\pi^{\dagger}$, we design the online learning-augmented control policy \ouralg as shown in Algorithm~\ref{alg:expert_robust_Q}.  
At each round $h$ within an episode, the controller first evaluates the loss of the control prior. To achieve this, after observing the true state $x_h$ and $f_{h-1}, w_{h-1}$, 
we first calculate a ``virtual state'' corresponding to the control prior for the same online information $y_{0:h-1}$, denoted by $\smash{x_h^{\dagger}=f_{h-1}(x_{h-1}^{\dagger},u_{h-1}^{\dagger})+w_{h-1}}$. Next, we query the control prior $\pi^{\dagger}$ with a state $x_h^{\dagger}$ and obtain an action $u_h^{\dagger}$, which can be used to update the cumulative risk $R_{h}^{\dagger}$ at round $h$. By doing so, a safe action set $\mathcal{U}_{\lambda,h}$ can be constructed by Proposition~\ref{thm:potentaildesign_concrete}.
 
To utilize the ML advice for loss performance, we select an action that is close enough to the pure ML action from the safe action set. If the ML action $\tilde{u}_h$ is in the safe action set $\mathcal{U}_{\lambda,h}$, then we simply select $u_h=\tilde{u}_h$. Otherwise, we can use a mapping function $m:\mathbb{R}^d\rightarrow \mathcal{U}_{\lambda,h}$ that maps the ML action $\tilde{u}_h$ into an action in the safe action set.  
One choice of $m$ is the projection operation which selects action as
 \begin{equation}\label{eqn:projection}
 u_h=m(\tilde{u}_h,\mathcal{U}_{\lambda,h})=\arg\min_{u\in\mathcal{U}_{\lambda,h}} \|\tilde{u}_h-u\|.
 \end{equation}
When the safe action set is a convex set (e.g. the dynamic functions $\{f_h:h\in [H]\}$ are linear  \cite{competitive_control_goel2022competitive,competitive_control_yu2022competitive,goel2022best,regret_LQR_predictions_zhang2021regret}), the projection can be efficiently solved. Otherwise,  the complexity can be high especially for high dimensional actions \cite{Homeomorphic_Projection_chen2021_low_complexity,differentiable_projection_chen2021enforcing}. Under such cases, we can choose $m$ as a linear combination as below \begin{equation}\label{eqn:linear_combination}u_h=m(\tilde{u}_h,\mathcal{U}_{\lambda,h})=\rho_h \tilde{u}_h+(1-\rho_h)u_h^{\dagger},
 \end{equation}
 where we need to solve an one-dimensional combination variable $\rho_h\in [0,1]$
 as a solution of $R_{h-1}+r_h(x_h,\rho_h \tilde{u}_h+(1-\rho_h)u_h^{\dagger})+\phi_h(\rho_h \tilde{u}_h+(1-\rho_h)u_h^{\dagger})= (1+\lambda) R^{\dagger}_{h}$.
 We will prove in Theorem \ref{thm:optimal} that \ouralg with both mapping functions in \eqref{eqn:projection} and \eqref{eqn:linear_combination} share the same expected loss bound.

The time complexity of LAOC is $O(H (T_\mathrm{ML}+T_\mathrm{prior}+T_\mathrm{map}) )$ where $T_\mathrm{ML}$, $T_\mathrm{prior}$ and $T_\mathrm{map}$ are the time complexities of the ML inference, the control prior and the mapping operations, respectively. $T_\mathrm{ML}$ is determined by the ML architecture. The time complexity of the control prior $T_\mathrm{prior}$ usually increases with the complexity of the control problem. Take the control prior ROBD \cite{SOCO_OBD_R-OBD_Goel_Adam_NIPS_2019_NEURIPS2019_9f36407e} as an example, the complexity to solve the optimization in ROBD scales with the dimension of the action. Furthermore, the mapping complexity $T_\mathrm{map}$ depends on the action-state dimensions and the complexity of the control model. If the safe action set in \eqref{eqn:robust_constraint} is convex (e.g. linear dynamic leads to a convex safe action set), we can use a convex optimization solver to efficiently solve the projection in \eqref{eqn:projection}. When the safe action set is non-convex, the projection into a non-convex action set has a high time complexity. In such cases, we can map the ML action to the safe action set by solving an one-dimensional combination variable in \eqref{eqn:linear_combination}.
 
\textbf{Safety-aware finetuning.}
If we have access to the pure ML model, we can finetune it based on available sequence data to further improve average loss performance with the safety guarantee. 
Specifically, 
given the pure ML model $\tilde{\pi}$ which outputs the ML action $\tilde{u}_h$ and the safe action set $\mathcal{U}_{\lambda, h}$, 
we finetune the ML model by minimizing the empirical loss of safe actions:
\begin{equation}\label{eqn:training}
\tilde{\pi}_{\lambda}^{(n)}\!\!=\!\arg\min_{\tilde{\pi}\in\Pi}\!\!\sum_{y_{1:H}\in\mathcal{D}_n}\sum_{h=1}^Hc_h\big(x_h,m\left(\tilde{u}_h,\mathcal{U}_{\lambda,h}\right)\big),
\end{equation} 
where $\mathcal{D}_n$ is the finetuning dataset with $n$ sequences. 
To finetune the ML model with~\eqref{eqn:training}, we can directly 
perform the back-propagation through the online process where all the operations are differentiable. The projection in \eqref{eqn:projection} can be implicitly differentiated as  shown in \cite{agrawal2019differentiable}. The linear mapping in \eqref{eqn:linear_combination} is also differentiable 
by differentiating the equation to solve $\rho_h$.
\subsection{Performance Bounds}\label{sec:analysis_cost_ratio}\label{sec:regret}
We provide the performance analysis of \ouralg in this section. We first present the conclusion that the the safety constraint in \eqref{eqn:safety_constraint} is always satisfied by \ouralg. Next, we give the average performance bound of \ouralg under the safety guarantee. Last but not least, we provide the performance bound by safety-aware finetuning in Eqn.\eqref{eqn:training}.

\subsubsection{Safety constraint satisfaction}
In Proposition \ref{thm:potentaildesign_concrete}, we prove that the safe set $\mathcal{U}_{\lambda,h}$ in \eqref{eqn:robust_constraint} is not empty for each round $h$. Since \ouralg (Algorithm \ref{alg:expert_robust_Q}) guarantees that the action $u_h$ lies in the safe set at each round, we can get the conclusion of safety constraint satisfaction in the next theorem. 
\begin{theorem}\label{thm:safety}
By \ouralg (Algorithm \ref{alg:expert_robust_Q}) with safety set $\mathcal{U}_{\lambda,h}$ in \eqref{eqn:robust_constraint}, for any problem sequence $y_{1:H}$ and any round $h\in [H]$, we can guarantee that the safety risk constraint in \eqref{eqn:safety_constraint} is satisfied.
\end{theorem}

Theorem \eqref{thm:safety} highlights that \ouralg can strictly guarantee $(1+\lambda)$-safety for any problem instance even when the ML policy $\tilde{\pi}$ has an arbitrarily bad performance. Under the safety guarantee, we are concerned about the expected loss performance given in the next theorem.

\subsubsection{Average performance}
The expected loss relies heavily on the choices of $C_1$ and $C_2$ in \eqref{eqn:robust_constraint} of Proposition \ref{thm:potentaildesign_concrete}.  To see this, if $C_1$ or $C_2$ is larger, the reservation $\phi_h(u)$ becomes larger, so the safe action set $\mathcal{U}_{\lambda,h}$ contains less feasible actions. Thus, the policy cannot utilize the ML model to improve the average loss performance effectively. 
On the contrary, if $C_1$ and $C_2$ approach 1, it can happen in the earlier rounds that the sizes of the safe action sets $\mathcal{U}_{\lambda,h}$ are too large and the selected action is too far from the prior action. This results in large state differences between $x_h$ and $x_h^{\dagger}$ in future rounds, resulting in small safe action set $\mathcal{U}_{\lambda,h}$ and impeding the exploitation of ML advice.
 The following analysis formally shows the factors that affect the expected performance and suggest the choices of $C_1$ and $C_2$. 

\begin{theorem}\label{thm:optimal}
Assume that the ML policy $\tilde{\pi}$  is $L_{\pi}$-Lipschitz continuous and the function $c_h$ is $L_c$-Lipschitz continuous,
by optimally choosing $C_1=1+\frac{1}{\sqrt{1+\lambda}}$ and $C_2=\arg\min_{c\geq 1} \{\frac{c}{c-1}\sigma_u^2(1-(c\sigma_x^2)^{H-h})/(1-c\sigma_x^2)\} $ in \eqref{eqn:robust_constraint},
the expected loss of \ouralg $\pi_{\lambda}$ that guarantees $(1+\lambda)$-safety is bounded by\begin{equation}\label{eqn:cr_bound_theorem}
 \begin{split}\mathbb{E}\left[J_H^{\pi_{\lambda}}\right]\leq \mathbb{E}\left[J_H^{\tilde{\pi}}\right]\!+\!B\mathbb{E}\left[\sum_{h=0}^{H-1} \left[\delta_h-(\sqrt{1+\lambda}-1)^2Gr_h^{\dagger}\right]^+ \right],
  \end{split}
  \end{equation}
 where 
 $\delta_h=\|\tilde{\pi}(\tilde{s}_h)-u_h^{\dagger}\|$ is the
 action discrepancy between the pure ML action and the control prior, $G\coloneqq 2(L_c(1+\frac{C_2}{C_2-1}\sigma_u^2(1-(C_2\sigma_x^2)^{H-h})/(1-C_2\sigma_x^2)))^{-1}$ 
 and $B\coloneqq
 L_c(1+(1+2L_{\pi})\sigma_u\sum_{i=0}^{H-1}(\sigma_x+2\sigma_uL_{\pi})^{h-i-1})$
 are constants of the control system,
  in which $\beta$ is the smoothness parameter of the risk function $r_h$, $A$ is the size of the state-action set, $L_{\pi}$ is the Lipschitz constant of the ML advice policy $\tilde{\pi}$,
 $\sigma_x$ and $\sigma_u$ are the Lipschitz constants
 of the dynamics model $f_h$. 
\end{theorem}

The expected loss bound in Theorem~\ref{thm:optimal} relies on the choices of $C_1$ and $C_2$. 
When $\lambda$ becomes larger, the safety constraint is more relaxed, so a smaller $C_1$ is chosen to get a smaller reservation $\phi_u(h)$ in Proposition \ref{thm:potentaildesign_concrete}, allowing more flexibility to follow the ML advice. Also, $C_2$ is selected to alleviate the impact of the dynamic sensitivity measured by $\sigma_x$ and $\sigma_u$ (Assumption~\ref{assumption:lipschitz}) on the expected loss.

The expected loss bound in Theorem~\ref{thm:optimal} can be interpreted as follows. 
First, the safety constraint naturally 
creates a gap of expected loss between  \ouralg  $\pi_{\lambda}$ and the ML advice $\tilde{\pi}$. More specifically,
given a control prior $\pi^{\dagger}$, when $\lambda>0$ becomes
smaller, the safety constraint is more stringent,
which thus makes the actions of \ouralg $\pi_{\lambda}$ potentially deviate more from those of
the ML advice policy $\tilde{\pi}$ and increases
the bound in \eqref{eqn:cr_bound_theorem}.
On the contrary,
when $\lambda>0$ becomes larger, the safety constraint is more relaxed,
reducing the expected loss of
 \ouralg $\pi_{\lambda}$.
In particular, if $\lambda$ is sufficiently large,
the term $[\delta_h-(\sqrt{1+\lambda}-1)^2Gr_h^{\dagger}]^+$  can reduce to zero, voiding
the safety constraint and resulting in the same expected loss as pure ML. 
Additionally, the expected loss is affected by the action discrepancy $\delta_h$ 
because a larger $\delta_h$ means larger difference between the prior $\pi^{\dagger}$ and the ML model $\tilde{\pi}$,
naturally making it more difficult for \ouralg to approach
ML $\tilde{\pi}$ while satisfying safety constraints.

\subsubsection{Generalization performance of safety-aware finetuning}
In this section, we consider the case in which the ML policy in \ouralg
is trained on \eqref{eqn:training}.
We bound the average loss gap between \ouralg policy and the  
unconstrained-optimal policy $\pi^*$.
We denote $\pi_{\lambda}^{(n)}(s_h)=m(\tilde{\pi}_{\lambda}^{(n)}(s_h),\mathcal{U}_{\lambda,t})$ as \ouralg policy by Algorithm \ref{alg:expert_robust_Q} with the ML model $\tilde{\pi}_{\lambda}^{(n)}$ trained by \eqref{eqn:training} on a dataset with $n$ traces,  and bound the expected loss $\mathbb{E}[J_H^{\pi_{\lambda}^{(n)}}]$ in the following theorem.

\begin{theorem}\label{thm:sublinear_main}
If ML policy is trained by the loss function in Eqn.~\eqref{eqn:training} with a training dataset with $n$ samples,
with probability at least $1-\delta, \delta\in(0,1)$, 
the expected loss of our competitiveness-constrained policy $\pi_{\lambda}^{(n)}$
is bounded by
\begin{equation}\nonumber
\begin{split}
\mathbb{E}\left[J_H^{\pi_{\lambda}^{(n)}}\right]\leq&
\mathbb{E}\left[J_H^{\pi^{*}}\right]+B\mathbb{E}\left[\sum_{h=0}^{H-1} \left[\delta_h-(\sqrt{1+\lambda}-1)^2Gr_h^{\dagger}\right]^+ \right]\\
&+\mathcal{O}\left(\sqrt{\frac{1}{n}\ln\frac{N(\epsilon, \Pi_{\lambda}, \hat{L}_1^n)}{\delta}}\right),
\end{split}
\end{equation}
where the system-related parameters $B, G$ and $\delta_h$ have the same definition as in Theorem \ref{thm:optimal},  $N(\epsilon, \Pi_{\lambda}, \hat{L}_1^n)$ is the $\epsilon$-covering number of the competitive policy space $\Pi_{\lambda}$ with $L_1$-norm as the distance measure
(the distance of two policies
$\pi$ and $\pi'$ is $\|\pi-\pi'\|_{\hat{L}_1^n}=\frac{1}{n}\sum_{t=1}^n\sum_{h=1}^H\|\pi(s_h^{(t)})-\pi'(s_h^{(t)})\|_1$) on the training dataset $\mathcal{D}_n$, and $\mathcal{O}$ indicates the scaling with the loss upper bound $P$, the horizon $H$, and the size of action-state space $\mathcal{X}\times\mathcal{U}$. 
\end{theorem}

Theorem~\ref{thm:sublinear_main}  shows that as the number
of training samples $n\to\infty$, the expected loss is bounded by the unconstrained-optimal expected loss $\mathbb{E}\left[J_H^{\pi^{*}}\right]$ plus an additional term relying on the expected loss of the prior $\pi^{\dagger}$ and the 
parameter $\lambda>0$.
This additional term is because the policy is optimized under the safety constraint in \eqref{eqn:safety_constraint}. 
When $\lambda$ becomes larger, the constraint is more relaxed and the expected loss is closer to the unconstrained-optimal expected loss.   
Also, Theorem~\ref{thm:sublinear_main} shows that our policy with the online-trained ML model  converges with a rate of $\sqrt{1/n}$. In particular, the convergence rate is affected by $\lambda$ through the covering number $N(\epsilon, \Pi_{\lambda}, \hat{L}_1^n)$ which indicates the richness of the competitive policy class $\Pi_{\lambda}$. 
Comparing to the unconstrained policy set $\Pi_{\infty}$, the covering number of the competitive policy class $\Pi_{\lambda}$ is smaller. This is because with the same ML model, 
the safety constraint reduces the set of feasible actions --- with a smaller $\lambda>0$,
the safe policy space becomes smaller, making
it easier for the convergence of \ouralg.

\section{Case Study}\label{sec:water_experiment}

In this section, we evaluate the performance of \ouralg by experiments on a concrete water supply case and compare \ouralg with different control baselines. 
\subsection{Setup}
In this section, we provide the experimental setups on the water supply management. We first present the architecture of water supply system with roof top water tanks. Next, we introduce the datasets used in the experiments including the traces of water demand, carbon intensity, and energy price. Following that, we define the concerned performance metrics in the experiments. Finally, we provide the settings of \ouralg and the baselines.
\subsubsection{Water supply system with roof top water tanks}
The water supply systems of many modern buildings are equipped with roof top water tanks. The roof top water tanks have large water storage capacities and exploit the gravity in the elevated level to supply water for building users. Water is pumped from municipal water sources to these roof top water tanks to maintain a water level. The water tanks can play an important part in sustaining water supply system because the manager can pump less water (by decreasing the activation time and/or the speed of pumps) to the water tanks when the carbon intensity and energy price are high while still satisfying the demand using the water stored in the water tanks. Beyond that, the water tanks are crucial for the safety of the buildings because they are equipped to supply water for fire protection systems and the mission-critical functions of the buildings. Therefore, we must make sure that the water level in a water tank is not far from its nominal water level to meet the safety requirements.

To sustain the water supply system for a building with roof top water tanks, we need to know the energy consumption of its pumping system to pump water to the water tanks. In this paper, we estimate the power $P_{\mathrm{pump}} (\mathrm{kW})$ of water pumping by the following formula converted from the horsepower formula used in engineering practice \cite{Horsepower_required_to_pump_water}:
\begin{equation}\label{eqn:power_pumping}
P_{\mathrm{pump}} = \frac{\mathrm{WF}\times \mathrm{HD} \times \mathrm{SG} }{102\cdot \eta_{\mathrm{pump}}},
\end{equation}
where $\mathrm{WF} (\mathrm{L/s})$ is the water volume flow, $\mathrm{HD} (\mathrm{m})$ is the height of the water tank, and $\mathrm{SG}=1$ is the water specific gravity, and $\eta_{\mathrm{pump}}$ is the power efficiency of the pumping system.
We develop a controller that decides the amount of water $u_h (\mathrm{m}^3)$ pumped into the water tanks in each hour round $h$. 
The effective water flow is $u_h/3.6 \;(\mathrm{L/s})$ \footnote{The amount of water supply per hour can be adjusted by either controlling the activation time of the pumps or the speed of the pumps. For ease of computation, we assume the speed of the pumps is constant within each hour.}, which corresponds to an energy consumption of $(u_h\times \mathrm{HD} \times \mathrm{SG})/(367.2\cdot \eta_{\mathrm{pump}})$ (in $\mathrm{kWh}$) by \eqref{eqn:power_pumping}. Here, we define the energy efficiency as the energy consumption to pump a unit $\mathrm{m}^3$ of water to the water tank in one hour and denote it as
\begin{equation}\label{energy_efficiency_pumping}
\eta = \frac{\mathrm{HD} \times \mathrm{SG} }{367.2\cdot \eta_{\mathrm{pump}}}.
\end{equation}

The setups in the experiment are given as below.
The buildings are $75\:\mathrm{m}$ high and have water tanks with a total volume of $80 m^3$ on the roof. Each control horizon has a span of 24 hours.  By \eqref{energy_efficiency_pumping}, the power consumption to pump a unit $\mathrm{m}^3$ of water is $\eta = 0.272\:\mathrm{kWh/m^3}$ by choosing the energy efficiency of the pumps as $\eta_{\mathrm{pump}}=75\%$ according to \cite{volk2013pump}.  
The controller decides the amount of water pumped into the water tanks in each hour as $u_h$ ($m^3$), so the energy consumption at hour round $h$ is $\eta \cdot u_h$ ($\mathrm{kWh}$).

\subsubsection{Performance Metrics}\label{sec:metrics_experiments}
In water supply management, the concerned performance metrics include the carbon and energy costs and the safety risk. Given the water supply system with roof top water tanks, the expressions of the objectives are given as below.

\textbf{Carbon cost}.
Given an action $u_h$ which is the the amount of pumped water at hour round $h$, the energy consumption is $\eta \cdot u_h$ ($\mathrm{kWh}$). Therefore, given the carbon intensity $e_h$ ($\mathrm{g/kWh}$) at round $h$, the carbon emission at round $h$ is $c_e(u_h, e_h)= e_h\cdot (\eta\cdot u_h)$. 

 \textbf{Energy cost}.
With the energy price $p_h$ for round $h$, the total energy cost at round $h$ is $c_p(u_h, e_h)=p_h\cdot (\eta \cdot u_h)$. 

\textbf{Deviation from the safe level}.
 We choose the nominal safe water level as $\bar{x}=40 \mathrm{m^3}$ (half of the total water tank capacity). We choose a quadratic penalty for water level deviation which restrains large deviation. Thus, the deviation is measured by the quadratic deviation from the nominal level $\bar{x}$, i.e. $c_w(x_h)= (x_h-\bar{x})^2$.

The loss function is a weighted combination of the deviation and the carbon and energy costs which is expressed as
\begin{equation}
c_h(x_h,u_h)= \gamma_1 \cdot (x_h-\bar{x})^2+\gamma_2\cdot  e_h\cdot (\eta\cdot u_h)+\gamma_3\cdot p_h\cdot (\eta \cdot u_h).
\end{equation}
We consider the expected loss $\mathbb{E}_{y_{1:H}}[J_H^{\pi}]=\mathbb{E}_{y_{1:H}}\left[\sum_{h=1}^H c_h(x_h,u_h)\right]$ where the expectation is taken on the distribution of the water demand, carbon intensity and energy price traces. 

\textbf{Safety risk}. The safety risk is determined by the deviation and the hourly energy consumption. A high deviation will increase the risk of not satisfying the water demand and a large energy consumption can add too much power load to the energy system. We consider a quadratic penalty to restrain large deviation and hourly energy consumption and the safety risk is expressed as
\begin{equation}\label{eqn:risk_func_exp}
r_h(x_h,u_h) = \gamma_w \cdot (x_h-\bar{x})^2+ \gamma_b\cdot (\eta \cdot u_h)^2.
\end{equation}
In some scenarios, we need to consider different penalties for overly-low and overly-high water levels and model the deviation as an asymmetric function. If the asymmetric function satisfies Assumption \ref{assumption:smooth} (e.g. an asymmetric $\mathrm{dist}$ function as $\mathrm{dist}(x_h,\bar{x})=\gamma_{w,1}(\bar{x}-x_h)^2$ if $\bar{x}\geq x_h$ and $\mathrm{dist}(x_h,\bar{x})=\gamma_{w,2}(x_h-\bar{x})^2$ if $\bar{x}<x_h$), the theoretical conclusions of \ouralg still hold, and the key observations of experiments also generalize to such asymmetric penalties.

Given a control prior $\pi^{\dagger}$, we consider $(1+\lambda)-$safety which guarantees for any sequence that the safety risk is always bounded by the scaled safety risk of $\pi^{\dagger}$, i.e. $\forall h \in [H], \forall y_{1:H} \in \mathcal{Y}, R_h^{\pi} \leq (1+\lambda) R_h^{\pi^{\dagger}}$. We also directly evaluate the safety risk performance by the maximum risk ratio on the testing dataset $\max_{y_{1:H}\in \mathcal{D}_{\mathrm{test}}} \left(R_H^\pi/R_H^{\pi^{\dagger}}\right)$, which is a commonly used metric for worst case performance. \cite{competitive_control_goel2022competitive,CompetitiveControl_GuanyaShi_CISS_2021_9400281,SOCO_ROBD_Adam_NeurIPS_2019_10.5555/3454287.3454455}.

\subsubsection{Water demand, carbon intensity and energy price}
The experiments are conducted based on some public datasets. We provide the details for the traces of water demand, carbon intensity and energy price as below.

\textbf{Water demand trace}. 
The consumed water at each hour round $h$ is $w_h$ and affects the water level dynamics through \eqref{eqn:dynamic}. 
In our experiments, we use the water demand dataset measured for university buildings in \cite{water_consumption_predixtion_bejarano2019swap}. 
For each building, the trace contains hourly water consumption from August 1st, 2018 to December 8th, 2018. Since the traces are measured on low-rise university buildings, we scale up the hourly water consumption by 10 to simulate the high-rise building with dense occupancy. 
The water consumption data of four residence hall is used for training the ML model for water supply management. We augment the water consumption data of another two residence halls and get the 1-year demand traces for 20 buildings which are held out for testing.  

To further evaluate the robustness of the algorithms, we also create an Out-Of-Distribution (OOD) testing dataset on the basis of the original testing dataset. We generate the OOD demand dataset by adding Gaussian noise to each sample in the original dataset.  The standard deviation of the Gaussian noise is set as 30\% of the maximum demand value.

\textbf{Carbon intensity trace}
The carbon intensity datasets are from California Independent System Operator (CAISO) which are published on the website of Electricity Maps \cite{electricity_maps}.
The carbon intensity datasets contain the hourly carbon intensity of a city in California. We use the carbon traces in 2022 to train the ML model, and we hold out the carbon traces in 2023 for validation and testing.  

\textbf{Energy price trace}
The electricity price datasets are from CAISO which are published on the website of Energy Online \cite{energy_price_data}. Each price trace in the dataset contains the energy price value every 5 mins. We convert the original traces into hourly price traces by calculating the average price within each hour. We use the price data in 2022 to train the ML model while holding out the price data in 2023 for validation and testing.

\begin{table*}[!t]
    \scriptsize
    \caption{Cost and risk performance}   \label{table:comparison_all_results}
    \vspace{-0.2cm}
    \centering
    \begin{tabular}{c|c|c|c|c|c|c|c|c|c|c} 
    \toprule
    \multirow{2}{*}{\textbf{Metrics}} & \multicolumn{4}{c|}{Control priors} & \multicolumn{2}{c|}{ML} & \multicolumn{4}{|c}{Learning-augmented designs}\\ 
    \cline{2-11}
                          & OGD & ROBD  & MPC-LSTM& TMPC & ML & CRL & \lin-0.5 & \lin-0.2 &\bf \ouralg ($\lambda=0.4$)&\bf \ouralg ($\lambda=0.8$) \\ 
\midrule
     \textbf{Avg. energy} (US\$)                         &7344 & 7347& 7008&7640& \bf 6169&6584&6169 &7190 &6872 & 6690\\
         \textbf{Avg. carbon} (kg)                         &17994 &18278 & 18007& 18820&\bf 16123 &16820& 17178& 17668 &17037 & 16526\\
  \textbf{Max risk ratio}    &2.04 &\bf 1.14 & 6.19&3.608&6.17 &4.60& 4.56&2.44 &  2.76 & 3.40\\ 
    \bottomrule
    \end{tabular}
\end{table*}

\begin{table*}[!t]
    \scriptsize
    \caption{Cost and risk performance under OOD setting}   \label{table:comparison_all_results_ood}
    \vspace{-0.2cm}
    \centering
    \begin{tabular}{c|c|c|c|c|c|c|c|c|c|c} 
    \toprule
    \multirow{2}{*}{\textbf{Metrics}} & \multicolumn{4}{c|}{Control priors} & \multicolumn{2}{c|}{ML} & \multicolumn{4}{|c}{Learning-augmented designs}\\ 
    \cline{2-11}
                          & OGD & ROBD  & MPC-LSTM& TMPC & ML & CRL & \lin-0.5 & \lin-0.2 &\bf \ouralg ($\lambda=0.4$)&\bf \ouralg ($\lambda=0.8$) \\ 
\midrule
     \textbf{Avg. energy} (US\$)                         &7610 & 7580& 7136&7937& \bf 6516&6854&7063 &7391 &6901 & 6798\\
         \textbf{Avg. carbon} (kg)                         &20494 &20577& 18630& 21143&\bf 18476 &18711& 19485& 20090 &19326 & 19060\\
  \textbf{Max risk ratio}    &4.09 &\bf 1.24 & 41.2&12.1&11.00 &7.68& 5.79&3.32 &  4.36 & 5.68\\ 
    \bottomrule
    \end{tabular}
\end{table*}

\subsubsection{Settings of \ouralg}
To implement \ouralg in Algorithm \ref{alg:expert_robust_Q}, we need an ML model $\tilde{\pi}$ and a control prior $\pi^{\dagger}$ as inputs. Also, we can perform safety-aware finetuning in Eqn.\eqref{eqn:training} to learn an ML model for Algorithm \ref{alg:expert_robust_Q}. Thus, we summarize the variants of \ouralg as follows.
\begin{itemize}
\item \ouralg($\lambda, \tilde{\pi}, \pi^{\dagger}$): We use ML model $\tilde{\pi}$ and control prior $\pi^{\dagger}$ as the inputs of \ouralg.   $\tilde{\pi}$ and $\pi^{\dagger}$ can be replaced with a concrete ML model and a concrete control prior, respectively. If not specified, \ouralg uses Online Gradient Descent (\ogd) as the control prior by default. If not specified, \ouralg uses the ML model purely trained without considering the safety constraint by default. $\lambda$ determines the $(1+\lambda)-$safety in \eqref{eqn:safety_constraint}. 
\item \ouralg-F($\lambda, \pi^{\dagger}$): We use control prior $\pi^{\dagger}$ and an ML model obtained by safety-aware finetuning in Eqn.\eqref{eqn:training} as the inputs of \ouralg.  If not specified, \ouralg-F uses Online Gradient Descent (\ogd) as the control prior by default. $\lambda$ determines the $(1+\lambda)-$safety in \eqref{eqn:safety_constraint}.
\end{itemize}

\textbf{ML model}.
The ML model for \ouralg is a recurrent neural network. It takes available information about demand, carbon intensity, and electricity price as inputs and outputs the action for each round. By default, the ML model has 2 hidden layers and each hidden layer has 12 neurons. The ML model is trained by the Adam optimizer with a learning rate $5\times 10^{-4}$ for 400 epochs.

\textbf{Control prior}.
The control prior can be selected from some controllers that focus on reducing the safety risk. \cite{MPC_energy_water_management_wanjiru2016model, MPC_real_time_water_operation_wang2021minimizing}. 
Some robust online optimization algorithms such as Online Gradient Descent ({\ogd}) \cite{SOCO_Prediction_Error_Meta_ZhenhuaLiu_SIGMETRICS_2019_10.1145/3322205.3311087}, Online Balanced Descent ({\robd}) 
\cite{SOCO_OBD_R-OBD_Goel_Adam_NIPS_2019_NEURIPS2019_9f36407e}  can be applied to optimize the safety risk, so they can serve as the control prior. Alternatively, we can apply Model Predictive Control ({\mpc}) \cite{MPC_energy_water_management_wanjiru2016model, MPC_real_time_water_operation_wang2021minimizing} to minimize the safety risk which is commonly used for water supply control as a control prior.  

Regarding the safe set in \eqref{eqn:robust_constraint}, the safety requirement parameter $\lambda$ is chosen from $[0,2]$. $C_1$ and $C_2$ are chosen based on Theorem \ref{thm:optimal}.

\subsubsection{Baselines}

We compare \ouralg with \ogd \cite{SOCO_Prediction_Error_Meta_ZhenhuaLiu_SIGMETRICS_2019_10.1145/3322205.3311087}, \robd
\cite{SOCO_OBD_R-OBD_Goel_Adam_NIPS_2019_NEURIPS2019_9f36407e} and \mpc \cite{MPC_energy_water_management_wanjiru2016model}  that focus on the safety risk, the pure ML that is trained on the average loss, and the naive learning-augmented design \lin. 
$\bullet$ Offline Optimal Policy ({\opt}): This is the optimal offline policy that knows all the information in advance and obtains
the optimal action for each episode. 

$\bullet$ Online Gradient Descent ({\ogd}): Online gradient descent \cite{SOCO_Prediction_Error_Meta_ZhenhuaLiu_SIGMETRICS_2019_10.1145/3322205.3311087} is an  online algorithm to minimize the safety risk without relying on any predictions. \ogd has provable regret bound and competitive ratio with proper choice of step size. We use \ogd as a control policy prior by default.

$\bullet$ Regularized Online Balanced Descent ({\robd}):
\robd is an online optimization algorithm to minimize the safety risk with one-step demand prediction. It enjoys provable competitive ratio given perfect one-step prediction \cite{SOCO_OBD_R-OBD_Goel_Adam_NIPS_2019_NEURIPS2019_9f36407e,SOCO_Memory_FeedbackDelay_Nonlinear_Adam_Sigmetrics_2022_10.1145/3508037}. We use \robd as a control policy prior.
We set the parameters for \robd
optimally according to \cite{SOCO_OBD_R-OBD_Goel_Adam_NIPS_2019_NEURIPS2019_9f36407e}.

$\bullet$ Model Predictive Control ({\mpc}): \mpc\cite{MPC_camacho2013model}
solves the control problem by leveraging predictions of the
future information. Here, we assume that at round $h$, the information $w_{h:H}$ is predicted as $\hat{w}_{h:H}$, and the per-round prediction error normalized by the maximum input range is $\epsilon=\mathbb{E}\left[\frac{1}{(H-h+1)(w_{\max})}\|w_{h:H}-\hat{w}_{h:H}\|\right]$.  In this paper, we use \mpc with a window size of $4$ hours as a control prior to minimize the risk. 
$\mpc-\epsilon$ is \mpc with a generated prediction error of $\epsilon$.

$\bullet$ Model Predictive Control with LSTM ({\mpclstm}): Due to the powerful time series prediction ability, Long Short-Term Memory (LSTM) has been utilized as a prediction model in \mpc in recent studies \cite{MPC_LSTM_jung2023model,huang2022lstm,LSTM_MPC_zarzycki2024lstm}. In the water supply control problem, we implement a LSTM model as the predictor and apply it in \mpc, which is called \mpclstm. The LSTM model has one LSTM layer with 60 hidden neurons and can predict the demand in the future 4 hours. The same training dataset for \rl is used for LSTM training. 

$\bullet$ Tube-based Model Predictive Control ({\tmpc}): 
\tmpc\cite{dynamic_tube_MPC_lopez2019dynamic,system_TMPC_sieber2021system} is a computationally efficient robust \mpc approach which creates state constraints (tube) based on a nominal dynamic model. \tmpc makes sure that the true state of \mpc stays within the tube. Since the nominal states are assumed to satisfy the constraint, \tmpc can also guarantee a constraint. In our experiments, we design a tube based on a nominal dynamic model exploiting the expected demand information. On the basis of existing \tmpc \cite{dynamic_tube_MPC_lopez2019dynamic,system_TMPC_sieber2021system}, we utilize the LSTM predictor in deciding an action while guaranteeing the action stays in the created tube.

$\bullet$ Machine Learning ({\rl}): This is
the purely-trained \rl model without safety constraints for any episode. 
 For fair comparison, we use the same neural architecture for pure ML and \ouralg. 

$\bullet$ Constrained Reinforcement Learning ({\crl}): As an important safe reinforcement learning algorithm, \crl \cite{CRL_ghosh2022provably,paternain2019constrained} has been applied for control problems. Most \crl methods guarantee a constraint in expectation or with a high probability. In our experiments, we implement \crl to satisfy the expected safety risk constraint $\mathbb{E}[R_h^{\pi}-(1 + \lambda) R_h^{\pi^{\dagger}}]\leq 0$ with $\lambda=0$. Not like original model-free \crl, we exploit the dynamic model information for value estimation in RL.

$\bullet$ Linear Combination (\textsf{\lin}): \lin-$\rho$ is the policy in Proposition \ref{thm:necessary_competitiveness} that linearly combines ML advice $\tilde{\pi}$ and \textsf{\robd} $\pi^{\dagger}$ as $\pi=\rho \tilde{\pi} + (1-\rho) \pi^{\dagger}$ with a combination factor $\rho\in [0,1]$. 

\begin{figure*}[t]	
	\centering
  \subfigure[Safety violation w/ $\lambda$]{
	\includegraphics[width=0.31\textwidth]{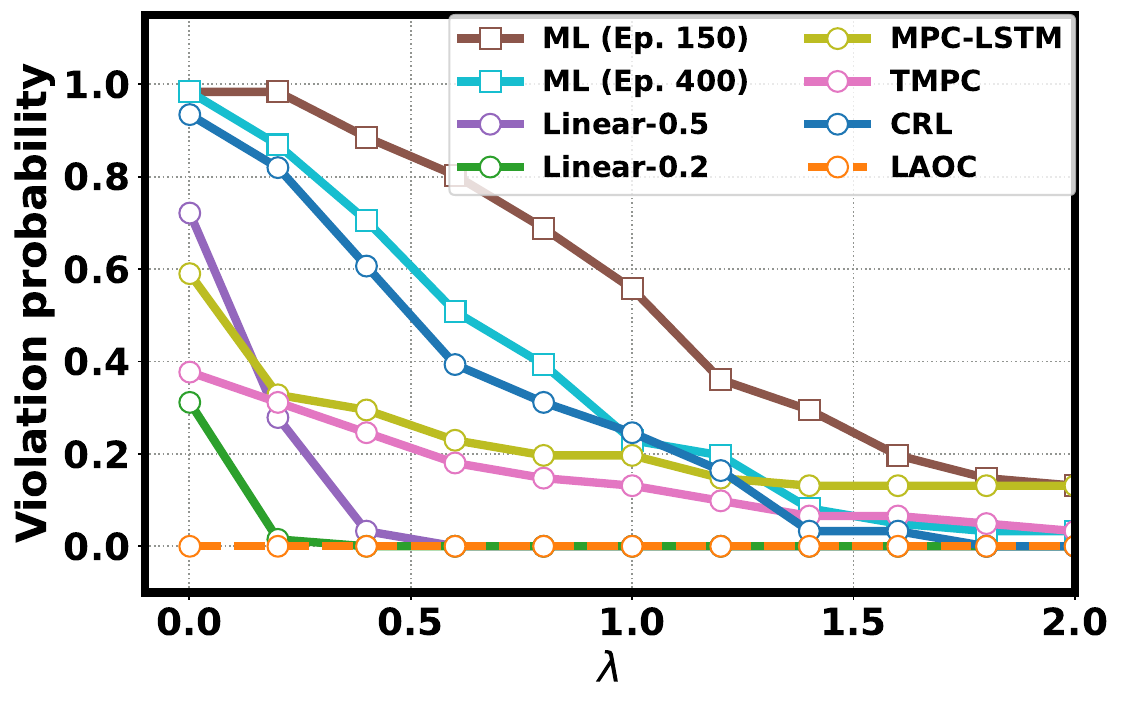}\label{fig:violation}
	}
  \subfigure[Average carbon cost]{
	\includegraphics[width=0.31\textwidth]{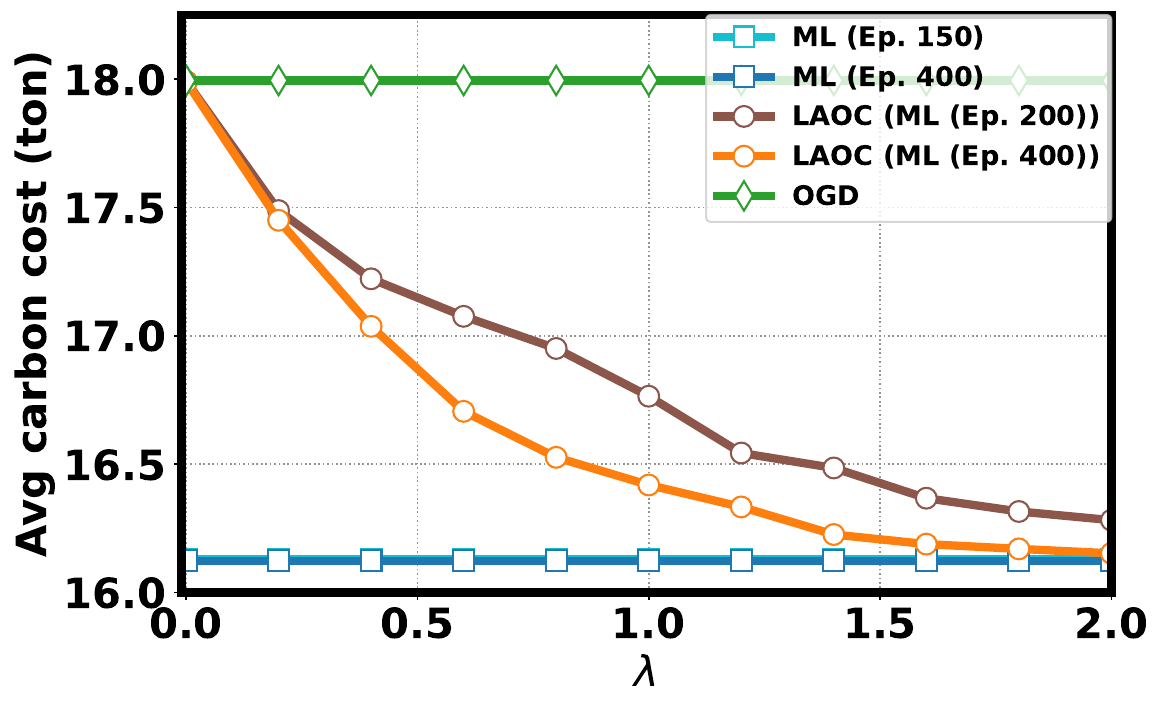}\label{fig:avg_carbon}
	}
 	 \subfigure[Average  energy cost]{
	\includegraphics[width=0.31\textwidth]{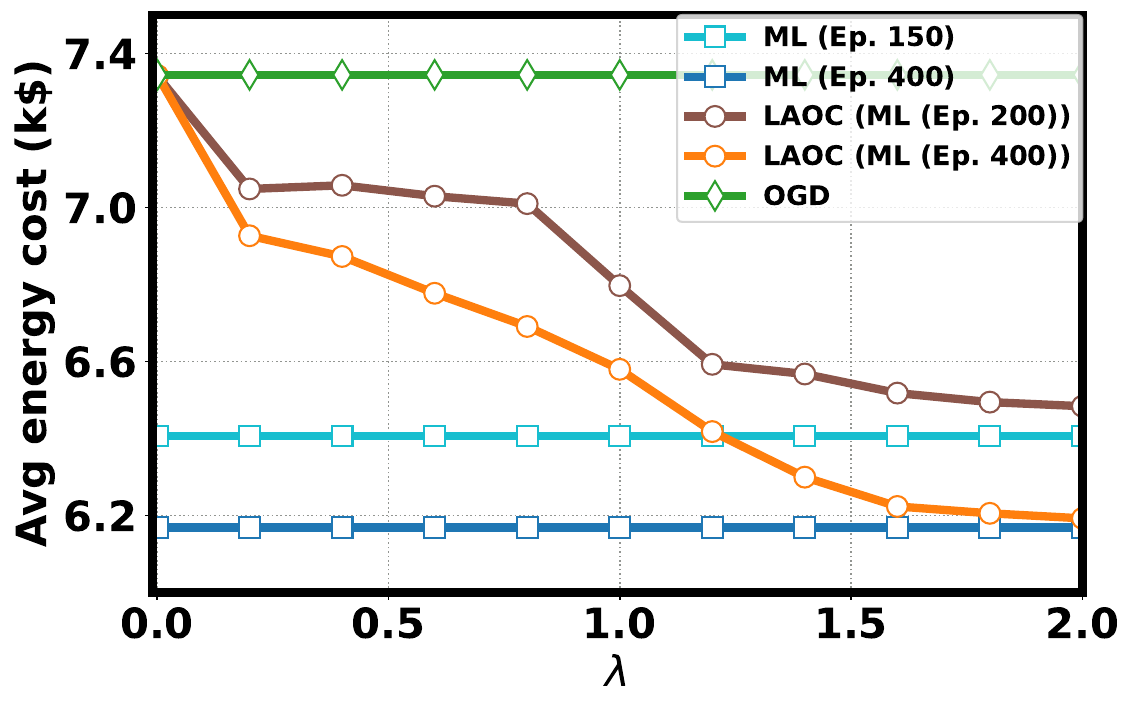}\label{fig:avg_energy}
	}
\vspace{-0.2cm}	
\caption{\small{Safety constraint violation and average costs. By default, \ogd is the control prior for \ouralg.  ML (Ep. $N$) is the ML model at the $N$th epoch. \ouralg (ML (Ep. $N$)) is \ouralg using the purely-trained ML model at the $N$th epoch.}}\label{fig:cost_standard}
\vspace{-0.5cm}
\end{figure*}

\subsection{Results}\label{sec:results}
We provide our main results for the default setting in Table \ref{table:comparison_all_results}. We give the average energy cost, the average carbon cost, and the maximum risk ratio for the control priors, the pure ML model and the learning-augmented algorithms including \lin and \ouralg. The results are evaluated for 20 buildings in one year.

First, we can find that the control prior \robd achieves the lowest safety risk which requires an accurate one-step prediction of the water demand. The control prior \ogd which does not rely on any prediction also achieves relatively low risk. However, the average energy costs and carbon emission are relatively large. Assuming a nearly-accurate predictor with a prediction error of $0.03$, \mpc-0.03 can achieve a maximum risk ratio of 2.52, an average carbon cost of 17782 kg, and an average energy cost of 6924 \$. Thus, \mpc has good risk and cost performances when a nearly-accurate predictor is applied. However, a real predictor such as LSTM in our experiments can have large prediction error. The LSTM in our experiment has an average prediction error of 0.05. This results in a higher safety risk and larger average energy cost and carbon emission as is shown by the performance of \mpclstm in Table \ref{table:comparison_all_results}. This shows that the performance of \mpc is largely affected by the quality of the predictor.  Furthermore, we can observe from Table \ref{table:comparison_all_results} that as a robust \mpc method, \tmpc can effectively reduce the safety risk, but it has much higher average energy costs and carbon emission. 

Different from control priors, the pure ML policy has the lowest energy cost and carbon emission, but it has much higher safety risk ratio. This is because the ML models are trained to optimize the average performances, but they can have arbitrarily bad performance when the adversarial instances exist. With the expected safety constraints, \crl can reduce the safety risk while sacrificing some average cost performance. However, we can observe from Table \ref{table:comparison_all_results} that the worst-case risk of \crl can still be very high, which is due to the existence of adversarial instances. The vulnerability of \rl and \crl impedes their deployments in real water supply systems which are critical for the safety of the buildings.

The learning-augmented designs are given to achieve a tradeoff between the average performance and the worst-case risk. As a naive learning-augmented design, \lin can reduce the safety risk to some extent by choosing a proper combination weight $\rho$. However, we can find that if we choose a large weight for ML model (e.g. $\rho =0.5$ in Table \ref{table:comparison_all_results}), the average performance can be good but the maximum risk ratio is still very high. Actually, by \lin-0.5, the $(1+\lambda)-$safety constraint is violated with a high probability as we will show in Figure \ref{fig:violation}. If we set a small weight for ML model (e.g. $\rho =0.2$ in Table \ref{table:comparison_all_results}), we can get a low safety risk ratio, but the average costs becomes very large. \ouralg is designed to optimizing the average performance while guaranteeing the $(1+\lambda)-$safety constraint in \eqref{eqn:safety_constraint}. We can observe that with a higher safety requirement $(e.g. \lambda=0.4)$, the safety risk ratio is low and close to that of control priors while the average costs are much lower than those of \lin. Also, with a lower safety requirement $(e.g. \lambda=0.8)$, the average costs of \ouralg are low and close to those of pure ML, and the risk ratio is also much lower than pure ML because the $(1+\lambda)-$safety constraint is always satisfied.
Next, we provide more details as below.

\subsubsection{Safety Violation}
The safety violation probability on testing dataset is given in Figure \ref{fig:violation}. The violation probability is the ratio of the number of safety constraint violation instances to the total testing instance number. A higher $\lambda$ in ($1+\lambda$)-safety in \eqref{eqn:safety_constraint} gives a less strict safety constraint, so the violation probability decreases with $\lambda$.   We can observe that ML can have a high safety violation probability even when $\lambda$ is large. If the ML model is not sufficiently trained (e.g. ML model at Epoch 150 in Figure \ref{fig:violation}), the safety violation probability is even higher. These show that pure ML is not safe enough for water supply systems. Although \crl reduces the safety violation probability comparing to \rl, it still has a high safety violation rate. Moreover, \mpclstm has a high safety violation rates due to the lack of prediction performance guarantee, and \tmpc has a reduced but non-zero safety violation probability.   As a learning-augmented design, \lin can also violates safety constraint especially when the safety requirement is high (small $\lambda$). Decreasing the combination weight for ML model from 0.5 to 0.2 can reduce the violation probability, but this results in a large increase of average costs shown in Table \ref{table:comparison_all_results}.   By contrast, \ouralg never violates safety constraint given any problem instance and any safety requirement parameter $\lambda$, which validates the effectiveness of \ouralg in strictly guaranteeing the safety constraint as proved in Theorem \ref{thm:safety}.   

\subsubsection{Cost-safety tradeoff}
Figure \ref{fig:avg_carbon} and Figure \ref{fig:avg_energy} demonstrates the tradeoff between the average costs and safety requirement for \ouralg. The preset parameter $\lambda$ in the safety constraint \eqref{eqn:safety_constraint} indicates the level of safety requirement: with smaller $\lambda$, the safety constraint becomes more strict.  When $\lambda=0$, the safety constraint is so strict that \ouralg reduces to the control prior \ogd which is the default control prior used in \ouralg. Thus, the carbon and energy costs of \ouralg is the same as those of \ogd.  When $\lambda$ becomes larger, $(1+\lambda)-$safety constraint \eqref{eqn:safety_constraint} becomes less strict, the average costs of \ouralg approaches the average costs of corresponding pure ML models, so \ouralg can achieve less carbon and energy costs. When $\lambda$ becomes large enough, we can observe that the average costs of \ouralg are the same as those of pure \rl model. The carbon and energy costs also show the impacts of the ML quality. The ML model at Epoch 400 is better than the ML model at Epoch 200, so the average costs of \ouralg with ML model at Epoch 400 are lower than those of \ouralg with ML model at Epoch 200.   These observations coincide with Theorem \ref{thm:optimal} which theoretically shows the tradeoff between the average costs and safety requirement.

\subsubsection{OOD Testing}
To further evaluate the robustness of the algorithms, we give results under the Out of Distribution (OOD) setting in Table \ref{table:comparison_all_results_ood}. We generate the OOD testing demand sequences by adding Gaussian noise to the original demand sequences. We can find that the control priors \robd and \ogd both achieve low enough safety risk ratio, but their average costs are very high. The LSTM predictor is largely affected by the OOD testing, causing a very high safety risk for \mpclstm. \tmpc can be applied to reduce the safety risk to some extent, but the worst-case safety risk is still very high. This is because the nominal model in \tmpc cannot define a robust enough tube in OOD setting. 

The pure ML policy and \crl policy are also largely affected by OOD testing. We can observe from Table \ref{table:comparison_all_results_ood} that \rl and \crl both have high safety risk in the worst case. The expected safety constraint satisfaction in \crl does not help a lot in OOD testing because \crl is trained on a distribution that is very different from the testing distribution. That being said, \rl still achieves the lowest average energy and carbon costs.

The learning-augmented designs that combine \rl with control priors can take an effect in achieving a low enough safety risk. Even \lin can achieve a low safety risk by choosing a good combination weight $\rho$. However, \lin has high average energy and carbon costs because it is limited in exploiting the ML predictions. Comparably, \ouralg (e.g. $\lambda=0.4$) not only guarantees a small enough risk for any problem instance, but also achieves low average energy and carbon costs.

\section{Concluding Remarks}

This work considers an online control problem for water supply management. Besides minimizing the average energy cost, we consider
the safety constraint against a given control prior.
We design a learning-augmented algorithm, \ouralg, that strictly ensure safety constraint.
Our analysis reveals the tradeoff between the cost performance and the safety requirement.
We evaluate the performance for a case study of building water supply, showing the superiority of \ouralg in reducing energy cost and carbon emission and guaranteeing the safety requirements. In the future, the proposed design can be extended to broader applications such as EV charging and sustainable data centers to improve the efficiency and provide safety guarantee for these systems.

\section{Acknowledgment}
Pengfei Li and Shaolei Ren were supported by NSF grants CCF-2324916 and CCF-2324941. Adam Wierman was supported by NSF grants CCF-2326609, CNS-2146814, CPS-2136197, CNS-2106403, and NGSDI-2105648 and the funding from the Resnick Sustainability Institute.

\bibliographystyle{ACM-Reference-Format}

\begin{thebibliography}{103}


\ifx \showCODEN    \undefined \def \showCODEN     #1{\unskip}     \fi
\ifx \showDOI      \undefined \def \showDOI       #1{#1}\fi
\ifx \showISBNx    \undefined \def \showISBNx     #1{\unskip}     \fi
\ifx \showISBNxiii \undefined \def \showISBNxiii  #1{\unskip}     \fi
\ifx \showISSN     \undefined \def \showISSN      #1{\unskip}     \fi
\ifx \showLCCN     \undefined \def \showLCCN      #1{\unskip}     \fi
\ifx \shownote     \undefined \def \shownote      #1{#1}          \fi
\ifx \showarticletitle \undefined \def \showarticletitle #1{#1}   \fi
\ifx \showURL      \undefined \def \showURL       {\relax}        \fi
\providecommand\bibfield[2]{#2}
\providecommand\bibinfo[2]{#2}
\providecommand\natexlab[1]{#1}
\providecommand\showeprint[2][]{arXiv:#2}

\bibitem[sus(2023)]%
        {sustainable_water_infrastructure}
 \bibinfo{year}{2023}\natexlab{}.
\newblock \bibinfo{title}{Planning for Sustainable Water Infrastructure}.
\newblock
  \bibinfo{howpublished}{\url{https://www.epa.gov/sustainable-water-infrastructure/planning-sustainable-water-infrastructure}}.
\newblock


\bibitem[ene(2024a)]%
        {energy_efficiency_water_supply}
 \bibinfo{year}{2024}\natexlab{a}.
\newblock \bibinfo{title}{Energy Efficiency for Water Utilities}.
\newblock
  \bibinfo{howpublished}{\url{https://www.epa.gov/sustainable-water-infrastructure/energy-efficiency-water-utilities}}.
\newblock


\bibitem[IEA(2024)]%
        {IEA_energy_and_water}
 \bibinfo{year}{2024}\natexlab{}.
\newblock \bibinfo{title}{Exploring the interdependence of two critical
  resources Energy and Water}.
\newblock
  \bibinfo{howpublished}{\url{https://www.iea.org/topics/energy-and-water}}.
\newblock


\bibitem[Hor(2024)]%
        {Horsepower_required_to_pump_water}
 \bibinfo{year}{2024}\natexlab{}.
\newblock \bibinfo{title}{Horsepower required to pump water.}
\newblock
  \bibinfo{howpublished}{\url{https://www.engineeringtoolbox.com/pumping-water-horsepower-d_753.html}}.
\newblock


\bibitem[ene(2024b)]%
        {energy_price_data}
 \bibinfo{year}{2024}\natexlab{b}.
\newblock \bibinfo{title}{Real-time Price by CAISO}.
\newblock
  \bibinfo{howpublished}{\url{http://www.energyonline.com/Data/GenericData.aspx?DataId=19&CAISO___Real-time_Price/}}.
\newblock


\bibitem[sav(2024)]%
        {saving_energy_public_water}
 \bibinfo{year}{2024}\natexlab{}.
\newblock \bibinfo{title}{STRATEGIES FOR SAVING ENERGY AT PUBLIC WATER
  SYSTEMS}.
\newblock
  \bibinfo{howpublished}{\url{https://www.epa.gov/sites/default/files/2015-04/documents/epa816f13004.pdf}}.
\newblock


\bibitem[Achiam et~al\mbox{.}(2017)]%
        {Conservative_ConstrainedPolicyOptimization_achiam2017constrained}
\bibfield{author}{\bibinfo{person}{Joshua Achiam}, \bibinfo{person}{David
  Held}, \bibinfo{person}{Aviv Tamar}, {and} \bibinfo{person}{Pieter Abbeel}.}
  \bibinfo{year}{2017}\natexlab{}.
\newblock \showarticletitle{Constrained Policy Optimization}. In
  \bibinfo{booktitle}{\emph{International conference on machine learning}}.
  PMLR, \bibinfo{pages}{22--31}.
\newblock


\bibitem[Agrawal et~al\mbox{.}(2019)]%
        {agrawal2019differentiable}
\bibfield{author}{\bibinfo{person}{Akshay Agrawal}, \bibinfo{person}{Brandon
  Amos}, \bibinfo{person}{Shane Barratt}, \bibinfo{person}{Stephen Boyd},
  \bibinfo{person}{Steven Diamond}, {and} \bibinfo{person}{J~Zico Kolter}.}
  \bibinfo{year}{2019}\natexlab{}.
\newblock \showarticletitle{Differentiable convex optimization layers}.
\newblock \bibinfo{journal}{\emph{Advances in neural information processing
  systems}}  \bibinfo{volume}{32} (\bibinfo{year}{2019}).
\newblock


\bibitem[Alomrani et~al\mbox{.}(2022)]%
  {L2O_OnlineBipartiteMatching_Toronto_ArXiv_2021_DBLP:journals/corr/abs-2109-10380}
\bibfield{author}{\bibinfo{person}{Mohammad~Ali Alomrani},
  \bibinfo{person}{Reza Moravej}, {and} \bibinfo{person}{Elias~Boutros
  Khalil}.} \bibinfo{year}{2022}\natexlab{}.
\newblock \showarticletitle{Deep Policies for Online Bipartite Matching: A
  Reinforcement Learning Approach}.
\newblock \bibinfo{journal}{\emph{Transactions on Machine Learning Research}}
  (\bibinfo{year}{2022}).
\newblock
\showISSN{2835-8856}
\urldef\tempurl%
\url{https://openreview.net/forum?id=mbwm7NdkpO}
\showURL{%
\tempurl}


\bibitem[Amani et~al\mbox{.}(2021)]%
        {SafeRL_LinearFunction_YangLin_UCLA_ICML_2021_pmlr-v139-amani21a}
\bibfield{author}{\bibinfo{person}{Sanae Amani}, \bibinfo{person}{Christos
  Thrampoulidis}, {and} \bibinfo{person}{Lin Yang}.}
  \bibinfo{year}{2021}\natexlab{}.
\newblock \showarticletitle{Safe Reinforcement Learning with Linear Function
  Approximation}. In \bibinfo{booktitle}{\emph{Proceedings of the 38th
  International Conference on Machine Learning}}
  \emph{(\bibinfo{series}{Proceedings of Machine Learning Research},
  Vol.~\bibinfo{volume}{139})}, \bibfield{editor}{\bibinfo{person}{Marina
  Meila} {and} \bibinfo{person}{Tong Zhang}} (Eds.). \bibinfo{publisher}{PMLR},
  \bibinfo{pages}{243--253}.
\newblock
\urldef\tempurl%
\url{https://proceedings.mlr.press/v139/amani21a.html}
\showURL{%
\tempurl}


\bibitem[Anderson and Moore(2012)]%
        {classic_control_anderson2012optimal}
\bibfield{author}{\bibinfo{person}{BD Anderson} {and} \bibinfo{person}{John~B
  Moore}.} \bibinfo{year}{2012}\natexlab{}.
\newblock \showarticletitle{Optimal Filtering. Courier Corporation}.
\newblock \bibinfo{journal}{\emph{Courier Corporation}} (\bibinfo{year}{2012}).
\newblock


\bibitem[Antoniadis et~al\mbox{.}(2020)]%
  {SOCO_MetricUntrustedPrediction_Google_ICML_2020_pmlr-v119-antoniadis20a}
\bibfield{author}{\bibinfo{person}{Antonios Antoniadis},
  \bibinfo{person}{Christian Coester}, \bibinfo{person}{Marek Elias},
  \bibinfo{person}{Adam Polak}, {and} \bibinfo{person}{Bertrand Simon}.}
  \bibinfo{year}{2020}\natexlab{}.
\newblock \showarticletitle{Online Metric Algorithms with Untrusted
  Predictions}. In \bibinfo{booktitle}{\emph{ICML}}.
\newblock


\bibitem[Aswani et~al\mbox{.}(2013)]%
        {robust_learning_MPC_aswani2013provably}
\bibfield{author}{\bibinfo{person}{Anil Aswani}, \bibinfo{person}{Humberto
  Gonzalez}, \bibinfo{person}{S~Shankar Sastry}, {and} \bibinfo{person}{Claire
  Tomlin}.} \bibinfo{year}{2013}\natexlab{}.
\newblock \showarticletitle{Provably safe and robust learning-based model
  predictive control}.
\newblock \bibinfo{journal}{\emph{Automatica}} \bibinfo{volume}{49},
  \bibinfo{number}{5} (\bibinfo{year}{2013}), \bibinfo{pages}{1216--1226}.
\newblock


\bibitem[Bejarano et~al\mbox{.}(2019)]%
        {water_consumption_predixtion_bejarano2019swap}
\bibfield{author}{\bibinfo{person}{Gissella Bejarano}, \bibinfo{person}{Adita
  Kulkarni}, \bibinfo{person}{Raushan Raushan}, \bibinfo{person}{Anand
  Seetharam}, {and} \bibinfo{person}{Arti Ramesh}.}
  \bibinfo{year}{2019}\natexlab{}.
\newblock \showarticletitle{Swap: Probabilistic graphical and deep learning
  models for water consumption prediction}. In
  \bibinfo{booktitle}{\emph{Proceedings of the 6th ACM International Conference
  on Systems for Energy-Efficient Buildings, Cities, and Transportation}}.
  \bibinfo{pages}{233--242}.
\newblock


\bibitem[Bousquet et~al\mbox{.}(2004)]%
        {statistical_learning_bousquet2004introduction}
\bibfield{author}{\bibinfo{person}{Olivier Bousquet},
  \bibinfo{person}{St{\'e}phane Boucheron}, {and} \bibinfo{person}{G{\'a}bor
  Lugosi}.} \bibinfo{year}{2004}\natexlab{}.
\newblock \showarticletitle{Introduction to statistical learning theory}.
\newblock \bibinfo{journal}{\emph{Advanced Lectures on Machine Learning: ML
  Summer Schools 2003, Canberra, Australia, February 2-14, 2003, T{\"u}bingen,
  Germany, August 4-16, 2003, Revised Lectures}} (\bibinfo{year}{2004}),
  \bibinfo{pages}{169--207}.
\newblock


\bibitem[Brunke et~al\mbox{.}(2022)]%
        {safe_learning_in_robotics_brunke2022safe}
\bibfield{author}{\bibinfo{person}{Lukas Brunke}, \bibinfo{person}{Melissa
  Greeff}, \bibinfo{person}{Adam~W Hall}, \bibinfo{person}{Zhaocong Yuan},
  \bibinfo{person}{Siqi Zhou}, \bibinfo{person}{Jacopo Panerati}, {and}
  \bibinfo{person}{Angela~P Schoellig}.} \bibinfo{year}{2022}\natexlab{}.
\newblock \showarticletitle{Safe learning in robotics: From learning-based
  control to safe reinforcement learning}.
\newblock \bibinfo{journal}{\emph{Annual Review of Control, Robotics, and
  Autonomous Systems}} \bibinfo{volume}{5}, \bibinfo{number}{1}
  (\bibinfo{year}{2022}), \bibinfo{pages}{411--444}.
\newblock


\bibitem[Camacho and Alba(2013)]%
        {MPC_camacho2013model}
\bibfield{author}{\bibinfo{person}{Eduardo~F Camacho} {and}
  \bibinfo{person}{Carlos~Bordons Alba}.} \bibinfo{year}{2013}\natexlab{}.
\newblock \bibinfo{booktitle}{\emph{Model predictive control}}.
\newblock \bibinfo{publisher}{Springer}.
\newblock


\bibitem[Castellano et~al\mbox{.}(2022)]%
  {SafeRL_AlmostSureViolationConstraint_JHU_2022_castellano2021reinforcement}
\bibfield{author}{\bibinfo{person}{Agustin Castellano},
  \bibinfo{person}{Hancheng Min}, \bibinfo{person}{Juan Bazerque}, {and}
  \bibinfo{person}{Enrique Mallada}.} \bibinfo{year}{2022}\natexlab{}.
\newblock \showarticletitle{Reinforcement Learning with Almost Sure
  Constraints}. In \bibinfo{booktitle}{\emph{Learning for Dynamics and
  Control}}.
\newblock


\bibitem[Cheh et~al\mbox{.}(2024)]%
        {Water_Pump_Operation_eEnergy24}
\bibfield{author}{\bibinfo{person}{Carmen Cheh}, \bibinfo{person}{Justin
  Albrethsen}, \bibinfo{person}{Zhen~Wei Ng}, \bibinfo{person}{Binbin Chen},
  \bibinfo{person}{Xin Lou}, \bibinfo{person}{Zaki Masood}, {and}
  \bibinfo{person}{David~KY Yau}.} \bibinfo{year}{2024}\natexlab{}.
\newblock \showarticletitle{Water Pump Operation Optimization under Dynamic
  Market and Consumer Behaviour}.
\newblock  (\bibinfo{year}{2024}), \bibinfo{pages}{335--346}.
\newblock


\bibitem[Chen et~al\mbox{.}(2021)]%
        {differentiable_projection_chen2021enforcing}
\bibfield{author}{\bibinfo{person}{Bingqing Chen}, \bibinfo{person}{Priya~L
  Donti}, \bibinfo{person}{Kyri Baker}, \bibinfo{person}{J~Zico Kolter}, {and}
  \bibinfo{person}{Mario Berg{\'e}s}.} \bibinfo{year}{2021}\natexlab{}.
\newblock \showarticletitle{Enforcing policy feasibility constraints through
  differentiable projection for energy optimization}. In
  \bibinfo{booktitle}{\emph{Proceedings of the Twelfth ACM International
  Conference on Future Energy Systems}}. \bibinfo{pages}{199--210}.
\newblock


\bibitem[Chervonyi et~al\mbox{.}(2022a)]%
        {RL_cooling_chervonyi2022semi}
\bibfield{author}{\bibinfo{person}{Yuri Chervonyi}, \bibinfo{person}{Praneet
  Dutta}, \bibinfo{person}{Piotr Trochim}, \bibinfo{person}{Octavian Voicu},
  \bibinfo{person}{Cosmin Paduraru}, \bibinfo{person}{Crystal Qian},
  \bibinfo{person}{Emre Karagozler}, \bibinfo{person}{Jared~Quincy Davis},
  \bibinfo{person}{Richard Chippendale}, \bibinfo{person}{Gautam Bajaj},
  {et~al\mbox{.}}} \bibinfo{year}{2022}\natexlab{a}.
\newblock \showarticletitle{Semi-analytical industrial cooling system model for
  reinforcement learning}.
\newblock \bibinfo{journal}{\emph{arXiv preprint arXiv:2207.13131}}
  (\bibinfo{year}{2022}).
\newblock


\bibitem[Chervonyi et~al\mbox{.}(2022b)]%
        {chervonyi2022semi}
\bibfield{author}{\bibinfo{person}{Yuri Chervonyi}, \bibinfo{person}{Praneet
  Dutta}, \bibinfo{person}{Piotr Trochim}, \bibinfo{person}{Octavian Voicu},
  \bibinfo{person}{Cosmin Paduraru}, \bibinfo{person}{Crystal Qian},
  \bibinfo{person}{Emre Karagozler}, \bibinfo{person}{Jared~Quincy Davis},
  \bibinfo{person}{Richard Chippendale}, \bibinfo{person}{Gautam Bajaj},
  {et~al\mbox{.}}} \bibinfo{year}{2022}\natexlab{b}.
\newblock \showarticletitle{Semi-analytical industrial cooling system model for
  reinforcement learning}.
\newblock \bibinfo{journal}{\emph{arXiv preprint arXiv:2207.13131}}
  (\bibinfo{year}{2022}).
\newblock


\bibitem[Christianson et~al\mbox{.}(2023)]%
        {optimal_christianson2023}
\bibfield{author}{\bibinfo{person}{Nicolas Christianson},
  \bibinfo{person}{Junxuan Shen}, {and} \bibinfo{person}{Adam Wierman}.}
  \bibinfo{year}{2023}\natexlab{}.
\newblock \showarticletitle{Optimal robustness-consistency tradeoffs for
  learning-augmented metrical task systems}. In \bibinfo{booktitle}{\emph{AI
  STATS}}.
\newblock


\bibitem[Comden et~al\mbox{.}(2019)]%
  {SOCO_Prediction_Error_Meta_ZhenhuaLiu_SIGMETRICS_2019_10.1145/3322205.3311087}
\bibfield{author}{\bibinfo{person}{Joshua Comden}, \bibinfo{person}{Sijie Yao},
  \bibinfo{person}{Niangjun Chen}, \bibinfo{person}{Haipeng Xing}, {and}
  \bibinfo{person}{Zhenhua Liu}.} \bibinfo{year}{2019}\natexlab{}.
\newblock \showarticletitle{Online Optimization in Cloud Resource Provisioning:
  Predictions, Regrets, and Algorithms}.
\newblock \bibinfo{journal}{\emph{Proc. ACM Meas. Anal. Comput. Syst.}}
  \bibinfo{volume}{3}, \bibinfo{number}{1}, Article \bibinfo{articleno}{16}
  (\bibinfo{date}{March} \bibinfo{year}{2019}), \bibinfo{numpages}{30}~pages.
\newblock
\urldef\tempurl%
\url{https://doi.org/10.1145/3322205.3311087}
\showDOI{\tempurl}


\bibitem[Ding et~al\mbox{.}(2021)]%
        {safe_exploartion_primal_dual_ding2021provably}
\bibfield{author}{\bibinfo{person}{Dongsheng Ding}, \bibinfo{person}{Xiaohan
  Wei}, \bibinfo{person}{Zhuoran Yang}, \bibinfo{person}{Zhaoran Wang}, {and}
  \bibinfo{person}{Mihailo Jovanovic}.} \bibinfo{year}{2021}\natexlab{}.
\newblock \showarticletitle{Provably efficient safe exploration via primal-dual
  policy optimization}. In \bibinfo{booktitle}{\emph{International Conference
  on Artificial Intelligence and Statistics}}. PMLR,
  \bibinfo{pages}{3304--3312}.
\newblock


\bibitem[Ding et~al\mbox{.}(2020)]%
        {constrained_RL_primal_dual_ding2020natural}
\bibfield{author}{\bibinfo{person}{Dongsheng Ding}, \bibinfo{person}{Kaiqing
  Zhang}, \bibinfo{person}{Tamer Basar}, {and} \bibinfo{person}{Mihailo
  Jovanovic}.} \bibinfo{year}{2020}\natexlab{}.
\newblock \showarticletitle{Natural policy gradient primal-dual method for
  constrained markov decision processes}.
\newblock \bibinfo{journal}{\emph{Advances in Neural Information Processing
  Systems}}  \bibinfo{volume}{33} (\bibinfo{year}{2020}),
  \bibinfo{pages}{8378--8390}.
\newblock


\bibitem[Du et~al\mbox{.}(2019)]%
  {L2O_OnlineResource_PriceCloud_ChuanWu_AAAI_2019_10.1609/aaai.v33i01.33017570}
\bibfield{author}{\bibinfo{person}{Bingqian Du}, \bibinfo{person}{Chuan Wu},
  {and} \bibinfo{person}{Zhiyi Huang}.} \bibinfo{year}{2019}\natexlab{}.
\newblock \showarticletitle{Learning Resource Allocation and Pricing for Cloud
  Profit Maximization}. In \bibinfo{booktitle}{\emph{Proceedings of the
  Thirty-Third AAAI Conference on Artificial Intelligence and Thirty-First
  Innovative Applications of Artificial Intelligence Conference and Ninth AAAI
  Symposium on Educational Advances in Artificial Intelligence}} (Honolulu,
  Hawaii, USA) \emph{(\bibinfo{series}{AAAI'19/IAAI'19/EAAI'19})}.
  \bibinfo{publisher}{AAAI Press}, Article \bibinfo{articleno}{929},
  \bibinfo{numpages}{8}~pages.
\newblock
\showISBNx{978-1-57735-809-1}
\urldef\tempurl%
\url{https://doi.org/10.1609/aaai.v33i01.33017570}
\showDOI{\tempurl}


\bibitem[D’Ambrosio et~al\mbox{.}(2015)]%
        {water_network_optimization_d2015mathematical}
\bibfield{author}{\bibinfo{person}{Claudia D’Ambrosio},
  \bibinfo{person}{Andrea Lodi}, \bibinfo{person}{Sven Wiese}, {and}
  \bibinfo{person}{Cristiana Bragalli}.} \bibinfo{year}{2015}\natexlab{}.
\newblock \showarticletitle{Mathematical programming techniques in water
  network optimization}.
\newblock \bibinfo{journal}{\emph{European Journal of Operational Research}}
  \bibinfo{volume}{243}, \bibinfo{number}{3} (\bibinfo{year}{2015}),
  \bibinfo{pages}{774--788}.
\newblock


\bibitem[Efroni et~al\mbox{.}(2020)]%
        {constrained_MDP_efroni2020exploration}
\bibfield{author}{\bibinfo{person}{Yonathan Efroni}, \bibinfo{person}{Shie
  Mannor}, {and} \bibinfo{person}{Matteo Pirotta}.}
  \bibinfo{year}{2020}\natexlab{}.
\newblock \showarticletitle{Exploration-exploitation in constrained mdps}.
\newblock \bibinfo{journal}{\emph{arXiv preprint arXiv:2003.02189}}
  (\bibinfo{year}{2020}).
\newblock


\bibitem[Fan et~al\mbox{.}(2020)]%
        {learning_tube_MPC_fan2020deep}
\bibfield{author}{\bibinfo{person}{David~D Fan}, \bibinfo{person}{Ali-akbar
  Agha-mohammadi}, {and} \bibinfo{person}{Evangelos~A Theodorou}.}
  \bibinfo{year}{2020}\natexlab{}.
\newblock \showarticletitle{Deep learning tubes for tube mpc}.
\newblock \bibinfo{journal}{\emph{arXiv preprint arXiv:2002.01587}}
  (\bibinfo{year}{2020}).
\newblock


\bibitem[Franco et~al\mbox{.}(2006)]%
        {franco2006robust}
\bibfield{author}{\bibinfo{person}{Ana L{\`u}cia~D Franco},
  \bibinfo{person}{Henri Bourl{\`e}s}, \bibinfo{person}{Edson~R De~Pieri},
  {and} \bibinfo{person}{Herve Guillard}.} \bibinfo{year}{2006}\natexlab{}.
\newblock \showarticletitle{Robust nonlinear control associating robust
  feedback linearization and H/sub/spl infin//control}.
\newblock \bibinfo{journal}{\emph{IEEE transactions on automatic control}}
  \bibinfo{volume}{51}, \bibinfo{number}{7} (\bibinfo{year}{2006}),
  \bibinfo{pages}{1200--1207}.
\newblock


\bibitem[Freeman and Kokotovic(2008)]%
        {freeman2008robust}
\bibfield{author}{\bibinfo{person}{Randy Freeman} {and}
  \bibinfo{person}{Petar~V Kokotovic}.} \bibinfo{year}{2008}\natexlab{}.
\newblock \bibinfo{booktitle}{\emph{Robust nonlinear control design:
  state-space and Lyapunov techniques}}.
\newblock \bibinfo{publisher}{Springer Science \& Business Media}.
\newblock


\bibitem[Gahlawat et~al\mbox{.}(2020)]%
        {L1_GP_gahlawat2020l1}
\bibfield{author}{\bibinfo{person}{Aditya Gahlawat}, \bibinfo{person}{Pan
  Zhao}, \bibinfo{person}{Andrew Patterson}, \bibinfo{person}{Naira
  Hovakimyan}, {and} \bibinfo{person}{Evangelos Theodorou}.}
  \bibinfo{year}{2020}\natexlab{}.
\newblock \showarticletitle{L1-GP: L1 adaptive control with Bayesian learning}.
  In \bibinfo{booktitle}{\emph{Learning for dynamics and control}}. PMLR,
  \bibinfo{pages}{826--837}.
\newblock


\bibitem[Ghaddar et~al\mbox{.}(2015)]%
        {pump_scheduling_water_networks}
\bibfield{author}{\bibinfo{person}{Bissan Ghaddar}, \bibinfo{person}{Joe
  Naoum-Sawaya}, \bibinfo{person}{Akihiro Kishimoto}, \bibinfo{person}{Nicole
  Taheri}, {and} \bibinfo{person}{Bradley Eck}.}
  \bibinfo{year}{2015}\natexlab{}.
\newblock \showarticletitle{A Lagrangian decomposition approach for the pump
  scheduling problem in water networks}.
\newblock \bibinfo{journal}{\emph{European Journal of Operational Research}}
  \bibinfo{volume}{241}, \bibinfo{number}{2} (\bibinfo{year}{2015}),
  \bibinfo{pages}{490--501}.
\newblock


\bibitem[Ghelichi et~al\mbox{.}(2018)]%
        {robust_optimization_water_ghelichi2018novel}
\bibfield{author}{\bibinfo{person}{Zabih Ghelichi}, \bibinfo{person}{Javad
  Tajik}, {and} \bibinfo{person}{Mir~Saman Pishvaee}.}
  \bibinfo{year}{2018}\natexlab{}.
\newblock \showarticletitle{A novel robust optimization approach for an
  integrated municipal water distribution system design under uncertainty: A
  case study of Mashhad}.
\newblock \bibinfo{journal}{\emph{Computers \& Chemical Engineering}}
  \bibinfo{volume}{110} (\bibinfo{year}{2018}), \bibinfo{pages}{13--34}.
\newblock


\bibitem[Ghosh et~al\mbox{.}(2022a)]%
        {constrained_RL_linear_approximation_ghosh2022provably}
\bibfield{author}{\bibinfo{person}{Arnob Ghosh}, \bibinfo{person}{Xingyu Zhou},
  {and} \bibinfo{person}{Ness Shroff}.} \bibinfo{year}{2022}\natexlab{a}.
\newblock \showarticletitle{Provably efficient model-free constrained rl with
  linear function approximation}.
\newblock \bibinfo{journal}{\emph{arXiv preprint arXiv:2206.11889}}
  (\bibinfo{year}{2022}).
\newblock


\bibitem[Ghosh et~al\mbox{.}(2022b)]%
        {CRL_ghosh2022provably}
\bibfield{author}{\bibinfo{person}{Arnob Ghosh}, \bibinfo{person}{Xingyu Zhou},
  {and} \bibinfo{person}{Ness Shroff}.} \bibinfo{year}{2022}\natexlab{b}.
\newblock \showarticletitle{Provably efficient model-free constrained rl with
  linear function approximation}.
\newblock \bibinfo{journal}{\emph{Advances in Neural Information Processing
  Systems}}  \bibinfo{volume}{35} (\bibinfo{year}{2022}),
  \bibinfo{pages}{13303--13315}.
\newblock


\bibitem[Goel et~al\mbox{.}(2022)]%
        {goel2022best}
\bibfield{author}{\bibinfo{person}{Gautam Goel}, \bibinfo{person}{Naman
  Agarwal}, \bibinfo{person}{Karan Singh}, {and} \bibinfo{person}{Elad Hazan}.}
  \bibinfo{year}{2022}\natexlab{}.
\newblock \showarticletitle{Best of Both Worlds in Online Control: Competitive
  Ratio and Policy Regret}.
\newblock \bibinfo{journal}{\emph{arXiv preprint arXiv:2211.11219}}
  (\bibinfo{year}{2022}).
\newblock


\bibitem[Goel and Hassibi(2021)]%
  {CompetitiveControl_Hassibi_arXiv_2021_https://doi.org/10.48550/arxiv.2107.13657}
\bibfield{author}{\bibinfo{person}{Gautam Goel} {and} \bibinfo{person}{Babak
  Hassibi}.} \bibinfo{year}{2021}\natexlab{}.
\newblock \bibinfo{title}{Competitive Control}.
\newblock
\newblock
\urldef\tempurl%
\url{https://doi.org/10.48550/ARXIV.2107.13657}
\showDOI{\tempurl}


\bibitem[Goel and Hassibi(2022)]%
        {competitive_control_goel2022competitive}
\bibfield{author}{\bibinfo{person}{Gautam Goel} {and} \bibinfo{person}{Babak
  Hassibi}.} \bibinfo{year}{2022}\natexlab{}.
\newblock \showarticletitle{Competitive control}.
\newblock \bibinfo{journal}{\emph{IEEE Trans. Automat. Control}}
  (\bibinfo{year}{2022}).
\newblock


\bibitem[Goel et~al\mbox{.}(2019a)]%
        {SOCO_OBD_R-OBD_Goel_Adam_NIPS_2019_NEURIPS2019_9f36407e}
\bibfield{author}{\bibinfo{person}{Gautam Goel}, \bibinfo{person}{Yiheng Lin},
  \bibinfo{person}{Haoyuan Sun}, {and} \bibinfo{person}{Adam Wierman}.}
  \bibinfo{year}{2019}\natexlab{a}.
\newblock \showarticletitle{Beyond Online Balanced Descent: An Optimal
  Algorithm for Smoothed Online Optimization}. In
  \bibinfo{booktitle}{\emph{NeurIPS}}, Vol.~\bibinfo{volume}{32}.
\newblock
\urldef\tempurl%
\url{https://proceedings.neurips.cc/paper/2019/file/9f36407ead0629fc166f14dde7970f68-Paper.pdf}
\showURL{%
\tempurl}


\bibitem[Goel et~al\mbox{.}(2019b)]%
        {SOCO_ROBD_Adam_NeurIPS_2019_10.5555/3454287.3454455}
\bibfield{author}{\bibinfo{person}{Gautam Goel}, \bibinfo{person}{Yiheng Lin},
  \bibinfo{person}{Haoyuan Sun}, {and} \bibinfo{person}{Adam Wierman}.}
  \bibinfo{year}{2019}\natexlab{b}.
\newblock \showarticletitle{Beyond online balanced descent: an optimal
  algorithm for smoothed online optimization}. In
  \bibinfo{booktitle}{\emph{Proceedings of the 33rd International Conference on
  Neural Information Processing Systems}}. \bibinfo{publisher}{Curran
  Associates Inc.}, \bibinfo{address}{Red Hook, NY, USA}, Article
  \bibinfo{articleno}{168}, \bibinfo{numpages}{11}~pages.
\newblock


\bibitem[Goel and Wierman(2019)]%
        {SOCO_OBD_LQR_Abstract_Goel_Adam_Caltech_2019_10.1145/3374888.3374892}
\bibfield{author}{\bibinfo{person}{Gautam Goel} {and} \bibinfo{person}{Adam
  Wierman}.} \bibinfo{year}{2019}\natexlab{}.
\newblock \showarticletitle{An Online Algorithm for Smoothed Online Convex
  Optimization}.
\newblock \bibinfo{journal}{\emph{SIGMETRICS Perform. Eval. Rev.}}
  \bibinfo{volume}{47}, \bibinfo{number}{2} (\bibinfo{date}{Dec.}
  \bibinfo{year}{2019}), \bibinfo{pages}{6–8}.
\newblock


\bibitem[Goryashko and Nemirovski(2014)]%
        {robust_energy_optimization_WDS_goryashko2014robust}
\bibfield{author}{\bibinfo{person}{Alexander~P Goryashko} {and}
  \bibinfo{person}{Arkadi~S Nemirovski}.} \bibinfo{year}{2014}\natexlab{}.
\newblock \showarticletitle{Robust energy cost optimization of water
  distribution system with uncertain demand}.
\newblock \bibinfo{journal}{\emph{Automation and Remote Control}}
  \bibinfo{volume}{75} (\bibinfo{year}{2014}), \bibinfo{pages}{1754--1769}.
\newblock


\bibitem[Hardt and Simchowitz(2018)]%
        {convex_course}
\bibfield{author}{\bibinfo{person}{Moritz Hardt} {and} \bibinfo{person}{Max
  Simchowitz}.} \bibinfo{year}{2018}\natexlab{}.
\newblock \bibinfo{title}{Convex Optimization and Approximation}.
\newblock
  \bibinfo{howpublished}{\url{https://ee227c.github.io/notes/ee227c-notes.pdf}}.
\newblock


\bibitem[Hewing et~al\mbox{.}(2020)]%
        {learning_MPC_hewing2020learning}
\bibfield{author}{\bibinfo{person}{Lukas Hewing}, \bibinfo{person}{Kim~P
  Wabersich}, \bibinfo{person}{Marcel Menner}, {and} \bibinfo{person}{Melanie~N
  Zeilinger}.} \bibinfo{year}{2020}\natexlab{}.
\newblock \showarticletitle{Learning-based model predictive control: Toward
  safe learning in control}.
\newblock \bibinfo{journal}{\emph{Annual Review of Control, Robotics, and
  Autonomous Systems}} \bibinfo{volume}{3}, \bibinfo{number}{1}
  (\bibinfo{year}{2020}), \bibinfo{pages}{269--296}.
\newblock


\bibitem[Huang et~al\mbox{.}(2022)]%
        {huang2022lstm}
\bibfield{author}{\bibinfo{person}{Keke Huang}, \bibinfo{person}{Ke Wei},
  \bibinfo{person}{Fanbiao Li}, \bibinfo{person}{Chunhua Yang}, {and}
  \bibinfo{person}{Weihua Gui}.} \bibinfo{year}{2022}\natexlab{}.
\newblock \showarticletitle{LSTM-MPC: A deep learning based predictive control
  method for multimode process control}.
\newblock \bibinfo{journal}{\emph{IEEE Transactions on Industrial Electronics}}
  \bibinfo{volume}{70}, \bibinfo{number}{11} (\bibinfo{year}{2022}),
  \bibinfo{pages}{11544--11554}.
\newblock


\bibitem[Joshi and Chowdhary(2019)]%
        {Deep_adaptive_control_joshi2019deep}
\bibfield{author}{\bibinfo{person}{Girish Joshi} {and} \bibinfo{person}{Girish
  Chowdhary}.} \bibinfo{year}{2019}\natexlab{}.
\newblock \showarticletitle{Deep model reference adaptive control}. In
  \bibinfo{booktitle}{\emph{2019 IEEE 58th Conference on Decision and Control
  (CDC)}}. IEEE, \bibinfo{pages}{4601--4608}.
\newblock


\bibitem[Joshi et~al\mbox{.}(2021)]%
        {Asynchronous_deep_adaptive_control_joshi2021asynchronous}
\bibfield{author}{\bibinfo{person}{Girish Joshi}, \bibinfo{person}{Jasvir
  Virdi}, {and} \bibinfo{person}{Girish Chowdhary}.}
  \bibinfo{year}{2021}\natexlab{}.
\newblock \showarticletitle{Asynchronous deep model reference adaptive
  control}. In \bibinfo{booktitle}{\emph{Conference on Robot Learning}}. PMLR,
  \bibinfo{pages}{984--1000}.
\newblock


\bibitem[Jung et~al\mbox{.}(2023)]%
        {MPC_LSTM_jung2023model}
\bibfield{author}{\bibinfo{person}{Marvin Jung}, \bibinfo{person}{Paulo~Renato
  da Costa~Mendes}, \bibinfo{person}{Magnus {\"O}nnheim}, {and}
  \bibinfo{person}{Emil Gustavsson}.} \bibinfo{year}{2023}\natexlab{}.
\newblock \showarticletitle{Model Predictive Control when utilizing LSTM as
  dynamic models}.
\newblock \bibinfo{journal}{\emph{Engineering Applications of Artificial
  Intelligence}}  \bibinfo{volume}{123} (\bibinfo{year}{2023}),
  \bibinfo{pages}{106226}.
\newblock


\bibitem[Khalil et~al\mbox{.}(1996)]%
        {khalil1996robust}
\bibfield{author}{\bibinfo{person}{IS Khalil}, \bibinfo{person}{JC Doyle},
  {and} \bibinfo{person}{K Glover}.} \bibinfo{year}{1996}\natexlab{}.
\newblock \bibinfo{booktitle}{\emph{Robust and optimal control}}.
\newblock \bibinfo{publisher}{Prentice hall}.
\newblock


\bibitem[Kiesel and Kusterman(2016)]%
        {electricity_markets}
\bibfield{author}{\bibinfo{person}{R{\"u}diger Kiesel} {and}
  \bibinfo{person}{Michael Kusterman}.} \bibinfo{year}{2016}\natexlab{}.
\newblock \showarticletitle{Structural models for coupled electricity markets}.
\newblock \bibinfo{journal}{\emph{Journal of Commodity Markets}}
  \bibinfo{volume}{3}, \bibinfo{number}{1} (\bibinfo{year}{2016}),
  \bibinfo{pages}{16--38}.
\newblock


\bibitem[Lee et~al\mbox{.}(2019)]%
        {EV_charging_data_lee2019acn}
\bibfield{author}{\bibinfo{person}{Zachary~J Lee}, \bibinfo{person}{Tongxin
  Li}, {and} \bibinfo{person}{Steven~H Low}.} \bibinfo{year}{2019}\natexlab{}.
\newblock \showarticletitle{ACN-data: Analysis and applications of an open EV
  charging dataset}. In \bibinfo{booktitle}{\emph{Proceedings of the tenth ACM
  international conference on future energy systems}}.
  \bibinfo{pages}{139--149}.
\newblock


\bibitem[Li et~al\mbox{.}(2022a)]%
        {Shaolei_L2O_ExpertCalibrated_SOCO_SIGMETRICS_2022}
\bibfield{author}{\bibinfo{person}{Pengfei Li}, \bibinfo{person}{Jianyi Yang},
  {and} \bibinfo{person}{Shaolei Ren}.} \bibinfo{year}{2022}\natexlab{a}.
\newblock \showarticletitle{Expert-Calibrated Learning for Online Optimization
  with Switching Costs}. In \bibinfo{booktitle}{\emph{SIGMETRICS}}.
\newblock


\bibitem[Li et~al\mbox{.}(2022b)]%
        {Shaolei_L2O_ExpertCalibrated_SOCO_SIGMETRICS_Journal_2022}
\bibfield{author}{\bibinfo{person}{Pengfei Li}, \bibinfo{person}{Jianyi Yang},
  {and} \bibinfo{person}{Shaolei Ren}.} \bibinfo{year}{2022}\natexlab{b}.
\newblock \showarticletitle{Expert-Calibrated Learning for Online Optimization
  with Switching Costs}.
\newblock \bibinfo{journal}{\emph{Proc. ACM Meas. Anal. Comput. Syst.}}
  \bibinfo{volume}{6}, \bibinfo{number}{2}, Article \bibinfo{articleno}{28}
  (\bibinfo{date}{Jun} \bibinfo{year}{2022}), \bibinfo{numpages}{35}~pages.
\newblock


\bibitem[Li et~al\mbox{.}(2023a)]%
        {learning_OBM_ICML23}
\bibfield{author}{\bibinfo{person}{Pengfei Li}, \bibinfo{person}{Jianyi Yang},
  {and} \bibinfo{person}{Shaolei Ren}.} \bibinfo{year}{2023}\natexlab{a}.
\newblock \showarticletitle{Learning for Edge-Weighted Online Bipartite
  Matching with Robustness Guarantees}.
\newblock \bibinfo{journal}{\emph{ICML}} (\bibinfo{year}{2023}).
\newblock


\bibitem[Li et~al\mbox{.}(2023b)]%
        {expert_robustified_learning_infocom2023}
\bibfield{author}{\bibinfo{person}{Pengfei Li}, \bibinfo{person}{Jianyi Yang},
  {and} \bibinfo{person}{Shaolei Ren}.} \bibinfo{year}{2023}\natexlab{b}.
\newblock \showarticletitle{Robustified Learning for Online Optimization with
  Memory Costs}.
\newblock \bibinfo{journal}{\emph{INDOCOM}} (\bibinfo{year}{2023}).
\newblock


\bibitem[Li et~al\mbox{.}(2021)]%
        {renewable_aggregation_li2021learning}
\bibfield{author}{\bibinfo{person}{Tongxin Li}, \bibinfo{person}{Bo Sun},
  \bibinfo{person}{Yue Chen}, \bibinfo{person}{Zixin Ye},
  \bibinfo{person}{Steven~H Low}, {and} \bibinfo{person}{Adam Wierman}.}
  \bibinfo{year}{2021}\natexlab{}.
\newblock \showarticletitle{Learning-based Predictive Control via Real-time
  Aggregate Flexibility}.
\newblock \bibinfo{journal}{\emph{IEEE Transactions on Smart Grid}}
  \bibinfo{volume}{12}, \bibinfo{number}{6} (\bibinfo{year}{2021}),
  \bibinfo{pages}{4897--4913}.
\newblock


\bibitem[Li et~al\mbox{.}(2022c)]%
  {Control_RobustConsistency_LQC_TongxinLi_Sigmetrics_2022_10.1145/3508038}
\bibfield{author}{\bibinfo{person}{Tongxin Li}, \bibinfo{person}{Ruixiao Yang},
  \bibinfo{person}{Guannan Qu}, \bibinfo{person}{Guanya Shi},
  \bibinfo{person}{Chenkai Yu}, \bibinfo{person}{Adam Wierman}, {and}
  \bibinfo{person}{Steven Low}.} \bibinfo{year}{2022}\natexlab{c}.
\newblock \showarticletitle{Robustness and Consistency in Linear Quadratic
  Control with Untrusted Predictions}.
\newblock \bibinfo{journal}{\emph{Proc. ACM Meas. Anal. Comput. Syst.}}
  \bibinfo{volume}{6}, \bibinfo{number}{1}, Article \bibinfo{articleno}{18}
  (\bibinfo{date}{feb} \bibinfo{year}{2022}), \bibinfo{numpages}{35}~pages.
\newblock
\urldef\tempurl%
\url{https://doi.org/10.1145/3508038}
\showDOI{\tempurl}


\bibitem[Li et~al\mbox{.}(2019)]%
        {online_optimal_control_li2019online}
\bibfield{author}{\bibinfo{person}{Yingying Li}, \bibinfo{person}{Xin Chen},
  {and} \bibinfo{person}{Na Li}.} \bibinfo{year}{2019}\natexlab{}.
\newblock \showarticletitle{Online optimal control with linear dynamics and
  predictions: Algorithms and regret analysis}.
\newblock \bibinfo{journal}{\emph{Advances in Neural Information Processing
  Systems}}  \bibinfo{volume}{32} (\bibinfo{year}{2019}).
\newblock


\bibitem[Li and Li(2020)]%
  {SOCO_Prediction_Error_RHIG_NaLi_Harvard_NIPS_2020_NEURIPS2020_a6e4f250}
\bibfield{author}{\bibinfo{person}{Yingying Li} {and} \bibinfo{person}{Na Li}.}
  \bibinfo{year}{2020}\natexlab{}.
\newblock \showarticletitle{Leveraging Predictions in Smoothed Online Convex
  Optimization via Gradient-based Algorithms}. In
  \bibinfo{booktitle}{\emph{NeurIPS}}, Vol.~\bibinfo{volume}{33}.
\newblock
\urldef\tempurl%
\url{https://proceedings.neurips.cc/paper/2020/file/a6e4f250fb5c56aaf215a236c64e5b0a-Paper.pdf}
\showURL{%
\tempurl}


\bibitem[Li et~al\mbox{.}({[n.\,d.]})]%
        {learning_uncertainty_set_lilearning}
\bibfield{author}{\bibinfo{person}{Yingying Li}, \bibinfo{person}{Jing Yu},
  \bibinfo{person}{Lauren Conger}, \bibinfo{person}{Taylan Kargin}, {and}
  \bibinfo{person}{Adam Wierman}.} \bibinfo{year}{[n.\,d.]}\natexlab{}.
\newblock \showarticletitle{Learning the Uncertainty Sets of Linear Control
  Systems via Set Membership: A Non-asymptotic Analysis}. In
  \bibinfo{booktitle}{\emph{Forty-first International Conference on Machine
  Learning}}.
\newblock


\bibitem[Liang et~al\mbox{.}(2023)]%
        {Homeomorphic_Projection_chen2021_low_complexity}
\bibfield{author}{\bibinfo{person}{Enming Liang}, \bibinfo{person}{Minghua
  Chen}, {and} \bibinfo{person}{Steven~H. Low}.}
  \bibinfo{year}{2023}\natexlab{}.
\newblock \showarticletitle{Low Complexity Homeomorphic Projection to Ensure
  Neural-Network Solution Feasibility for Optimization over (Non-)Convex Set}.
  In \bibinfo{booktitle}{\emph{ICML}}.
\newblock


\bibitem[Lin et~al\mbox{.}(2022)]%
        {Decentralized_OCO_lin2022decentralized}
\bibfield{author}{\bibinfo{person}{Yiheng Lin}, \bibinfo{person}{Judy Gan},
  \bibinfo{person}{Guannan Qu}, \bibinfo{person}{Yash Kanoria}, {and}
  \bibinfo{person}{Adam Wierman}.} \bibinfo{year}{2022}\natexlab{}.
\newblock \showarticletitle{Decentralized Online Convex Optimization in
  Networked Systems}. In \bibinfo{booktitle}{\emph{International Conference on
  Machine Learning}}. PMLR, \bibinfo{pages}{13356--13393}.
\newblock


\bibitem[Lin et~al\mbox{.}(2021)]%
  {Control_PerturbationExponential_YihengLin_AdamWiermand_NIPS_2021_NEURIPS2021_298f5874}
\bibfield{author}{\bibinfo{person}{Yiheng Lin}, \bibinfo{person}{Yang Hu},
  \bibinfo{person}{Guanya Shi}, \bibinfo{person}{Haoyuan Sun},
  \bibinfo{person}{Guannan Qu}, {and} \bibinfo{person}{Adam Wierman}.}
  \bibinfo{year}{2021}\natexlab{}.
\newblock \showarticletitle{Perturbation-based Regret Analysis of Predictive
  Control in Linear Time Varying Systems}. In
  \bibinfo{booktitle}{\emph{Advances in Neural Information Processing
  Systems}}, \bibfield{editor}{\bibinfo{person}{M.~Ranzato},
  \bibinfo{person}{A.~Beygelzimer}, \bibinfo{person}{Y.~Dauphin},
  \bibinfo{person}{P.S. Liang}, {and} \bibinfo{person}{J.~Wortman Vaughan}}
  (Eds.), Vol.~\bibinfo{volume}{34}. \bibinfo{publisher}{Curran Associates,
  Inc.}, \bibinfo{pages}{5174--5185}.
\newblock
\urldef\tempurl%
\url{https://proceedings.neurips.cc/paper/2021/file/298f587406c914fad5373bb689300433-Paper.pdf}
\showURL{%
\tempurl}


\bibitem[Liu et~al\mbox{.}(2021)]%
        {liu2021learning}
\bibfield{author}{\bibinfo{person}{Tao Liu}, \bibinfo{person}{Ruida Zhou},
  \bibinfo{person}{Dileep Kalathil}, \bibinfo{person}{Panganamala Kumar}, {and}
  \bibinfo{person}{Chao Tian}.} \bibinfo{year}{2021}\natexlab{}.
\newblock \showarticletitle{Learning policies with zero or bounded constraint
  violation for constrained mdps}.
\newblock \bibinfo{journal}{\emph{Advances in Neural Information Processing
  Systems}}  \bibinfo{volume}{34} (\bibinfo{year}{2021}),
  \bibinfo{pages}{17183--17193}.
\newblock


\bibitem[Lopez et~al\mbox{.}(2019)]%
        {dynamic_tube_MPC_lopez2019dynamic}
\bibfield{author}{\bibinfo{person}{Brett~T Lopez},
  \bibinfo{person}{Jean-Jacques~E Slotine}, {and} \bibinfo{person}{Jonathan~P
  How}.} \bibinfo{year}{2019}\natexlab{}.
\newblock \showarticletitle{Dynamic tube MPC for nonlinear systems}. In
  \bibinfo{booktitle}{\emph{2019 American Control Conference (ACC)}}. IEEE,
  \bibinfo{pages}{1655--1662}.
\newblock


\bibitem[Luna et~al\mbox{.}(2019)]%
        {luna2019improving}
\bibfield{author}{\bibinfo{person}{Tiago Luna}, \bibinfo{person}{Jo{\~a}o
  Ribau}, \bibinfo{person}{David Figueiredo}, {and} \bibinfo{person}{Rita
  Alves}.} \bibinfo{year}{2019}\natexlab{}.
\newblock \showarticletitle{Improving energy efficiency in water supply systems
  with pump scheduling optimization}.
\newblock \bibinfo{journal}{\emph{Journal of cleaner production}}
  \bibinfo{volume}{213} (\bibinfo{year}{2019}), \bibinfo{pages}{342--356}.
\newblock


\bibitem[Luo et~al\mbox{.}(2022)]%
        {Cooling_control_luo2022controlling}
\bibfield{author}{\bibinfo{person}{Jerry Luo}, \bibinfo{person}{Cosmin
  Paduraru}, \bibinfo{person}{Octavian Voicu}, \bibinfo{person}{Yuri
  Chervonyi}, \bibinfo{person}{Scott Munns}, \bibinfo{person}{Jerry Li},
  \bibinfo{person}{Crystal Qian}, \bibinfo{person}{Praneet Dutta},
  \bibinfo{person}{Jared~Quincy Davis}, \bibinfo{person}{Ningjia Wu},
  {et~al\mbox{.}}} \bibinfo{year}{2022}\natexlab{}.
\newblock \showarticletitle{Controlling Commercial Cooling Systems Using
  Reinforcement Learning}.
\newblock \bibinfo{journal}{\emph{arXiv preprint arXiv:2211.07357}}
  (\bibinfo{year}{2022}).
\newblock


\bibitem[Maps(2024)]%
        {electricity_maps}
\bibfield{author}{\bibinfo{person}{Electricity Maps}.}
  \bibinfo{year}{2024}\natexlab{}.
\newblock \showarticletitle{Carbon Intensity Data (Version January 17, 2024)}.
\newblock \bibinfo{journal}{\emph{Electricity Maps Data Portal}}
  (\bibinfo{year}{2024}).
\newblock
\urldef\tempurl%
\url{https://www.electricitymaps.com/data-portal}
\showURL{%
\tempurl}


\bibitem[Oikonomou and Parvania(2018)]%
        {optimal_operation_oikonomou2018optimal}
\bibfield{author}{\bibinfo{person}{Konstantinos Oikonomou} {and}
  \bibinfo{person}{Masood Parvania}.} \bibinfo{year}{2018}\natexlab{}.
\newblock \showarticletitle{Optimal coordination of water distribution energy
  flexibility with power systems operation}.
\newblock \bibinfo{journal}{\emph{IEEE Transactions on Smart Grid}}
  \bibinfo{volume}{10}, \bibinfo{number}{1} (\bibinfo{year}{2018}),
  \bibinfo{pages}{1101--1110}.
\newblock


\bibitem[Pan et~al\mbox{.}(2022)]%
  {SOCO_Memory_FeedbackDelay_Nonlinear_Adam_Sigmetrics_2022_10.1145/3508037}
\bibfield{author}{\bibinfo{person}{Weici Pan}, \bibinfo{person}{Guanya Shi},
  \bibinfo{person}{Yiheng Lin}, {and} \bibinfo{person}{Adam Wierman}.}
  \bibinfo{year}{2022}\natexlab{}.
\newblock \showarticletitle{Online Optimization with Feedback Delay and
  Nonlinear Switching Cost}.
\newblock \bibinfo{journal}{\emph{Proc. ACM Meas. Anal. Comput. Syst.}}
  \bibinfo{volume}{6}, \bibinfo{number}{1}, Article \bibinfo{articleno}{17}
  (\bibinfo{date}{Feb} \bibinfo{year}{2022}), \bibinfo{numpages}{34}~pages.
\newblock
\urldef\tempurl%
\url{https://doi.org/10.1145/3508037}
\showDOI{\tempurl}


\bibitem[Paternain et~al\mbox{.}(2019)]%
        {paternain2019constrained}
\bibfield{author}{\bibinfo{person}{Santiago Paternain}, \bibinfo{person}{Luiz
  Chamon}, \bibinfo{person}{Miguel Calvo-Fullana}, {and}
  \bibinfo{person}{Alejandro Ribeiro}.} \bibinfo{year}{2019}\natexlab{}.
\newblock \showarticletitle{Constrained reinforcement learning has zero duality
  gap}.
\newblock \bibinfo{journal}{\emph{Advances in Neural Information Processing
  Systems}}  \bibinfo{volume}{32} (\bibinfo{year}{2019}).
\newblock


\bibitem[Polycarpou and Ioannou(1993)]%
        {polycarpou1993robust}
\bibfield{author}{\bibinfo{person}{Marios~M Polycarpou} {and}
  \bibinfo{person}{Petros~A Ioannou}.} \bibinfo{year}{1993}\natexlab{}.
\newblock \showarticletitle{A robust adaptive nonlinear control design}. In
  \bibinfo{booktitle}{\emph{1993 American control conference}}. IEEE,
  \bibinfo{pages}{1365--1369}.
\newblock


\bibitem[Quintiliani and Creaco(2019)]%
        {pump_control_optimization}
\bibfield{author}{\bibinfo{person}{Claudia Quintiliani} {and}
  \bibinfo{person}{Enrico Creaco}.} \bibinfo{year}{2019}\natexlab{}.
\newblock \showarticletitle{Using additional time slots for improving pump
  control optimization based on trigger levels}.
\newblock \bibinfo{journal}{\emph{Water Resources Management}}
  \bibinfo{volume}{33} (\bibinfo{year}{2019}), \bibinfo{pages}{3175--3186}.
\newblock


\bibitem[Rutten et~al\mbox{.}(2022)]%
        {Online_untrusted_predictions_rutten2022online}
\bibfield{author}{\bibinfo{person}{Daan Rutten}, \bibinfo{person}{Nico
  Christianson}, \bibinfo{person}{Debankur Mukherjee}, {and}
  \bibinfo{person}{Adam Wierman}.} \bibinfo{year}{2022}\natexlab{}.
\newblock \showarticletitle{Online Optimization with Untrusted Predictions}.
\newblock \bibinfo{journal}{\emph{arXiv preprint arXiv:2202.03519}}
  (\bibinfo{year}{2022}).
\newblock


\bibitem[Salahuddin et~al\mbox{.}(2016)]%
        {RL_resource_provisioning_salahuddin2016reinforcement}
\bibfield{author}{\bibinfo{person}{Mohammad~A Salahuddin}, \bibinfo{person}{Ala
  Al-Fuqaha}, {and} \bibinfo{person}{Mohsen Guizani}.}
  \bibinfo{year}{2016}\natexlab{}.
\newblock \showarticletitle{Reinforcement learning for resource provisioning in
  the vehicular cloud}.
\newblock \bibinfo{journal}{\emph{IEEE Wireless Communications}}
  \bibinfo{volume}{23}, \bibinfo{number}{4} (\bibinfo{year}{2016}),
  \bibinfo{pages}{128--135}.
\newblock


\bibitem[Shi(2021)]%
        {CompetitiveControl_GuanyaShi_CISS_2021_9400281}
\bibfield{author}{\bibinfo{person}{Guanya Shi}.}
  \bibinfo{year}{2021}\natexlab{}.
\newblock \showarticletitle{Competitive Control via Online Optimization with
  Memory, Delayed Feedback, and Inexact Predictions}. In
  \bibinfo{booktitle}{\emph{2021 55th Annual Conference on Information Sciences
  and Systems (CISS)}}.
\newblock


\bibitem[Shi et~al\mbox{.}(2020)]%
        {competitive_control_memory_shi2020online}
\bibfield{author}{\bibinfo{person}{Guanya Shi}, \bibinfo{person}{Yiheng Lin},
  \bibinfo{person}{Soon-Jo Chung}, \bibinfo{person}{Yisong Yue}, {and}
  \bibinfo{person}{Adam Wierman}.} \bibinfo{year}{2020}\natexlab{}.
\newblock \showarticletitle{Online optimization with memory and competitive
  control}.
\newblock \bibinfo{journal}{\emph{Advances in Neural Information Processing
  Systems}}  \bibinfo{volume}{33} (\bibinfo{year}{2020}),
  \bibinfo{pages}{20636--20647}.
\newblock


\bibitem[Sieber et~al\mbox{.}(2021)]%
        {system_TMPC_sieber2021system}
\bibfield{author}{\bibinfo{person}{Jerome Sieber}, \bibinfo{person}{Samir
  Bennani}, {and} \bibinfo{person}{Melanie~N Zeilinger}.}
  \bibinfo{year}{2021}\natexlab{}.
\newblock \showarticletitle{A system level approach to tube-based model
  predictive control}.
\newblock \bibinfo{journal}{\emph{IEEE Control Systems Letters}}
  \bibinfo{volume}{6} (\bibinfo{year}{2021}), \bibinfo{pages}{776--781}.
\newblock


\bibitem[Singh and Kekatos(2019)]%
        {optimal_scheduling_water_singh2019optimal}
\bibfield{author}{\bibinfo{person}{Manish~K Singh} {and}
  \bibinfo{person}{Vassilis Kekatos}.} \bibinfo{year}{2019}\natexlab{}.
\newblock \showarticletitle{Optimal scheduling of water distribution systems}.
\newblock \bibinfo{journal}{\emph{IEEE Transactions on Control of Network
  Systems}} \bibinfo{volume}{7}, \bibinfo{number}{2} (\bibinfo{year}{2019}),
  \bibinfo{pages}{711--723}.
\newblock


\bibitem[Sootla et~al\mbox{.}(2022)]%
        {Saute_rl_sootla2022saute}
\bibfield{author}{\bibinfo{person}{Aivar Sootla}, \bibinfo{person}{Alexander~I
  Cowen-Rivers}, \bibinfo{person}{Taher Jafferjee}, \bibinfo{person}{Ziyan
  Wang}, \bibinfo{person}{David~H Mguni}, \bibinfo{person}{Jun Wang}, {and}
  \bibinfo{person}{Haitham Ammar}.} \bibinfo{year}{2022}\natexlab{}.
\newblock \showarticletitle{Saut{\'e} rl: Almost surely safe reinforcement
  learning using state augmentation}. In
  \bibinfo{booktitle}{\emph{International Conference on Machine Learning}}.
  PMLR, \bibinfo{pages}{20423--20443}.
\newblock


\bibitem[Sopasakis et~al\mbox{.}(2018)]%
        {uncertainty_aware_water_management_sopasakis2018uncertainty}
\bibfield{author}{\bibinfo{person}{Pantelis Sopasakis}, \bibinfo{person}{Ajay~K
  Sampathirao}, \bibinfo{person}{Alberto Bemporad}, {and}
  \bibinfo{person}{Panagiotis Patrinos}.} \bibinfo{year}{2018}\natexlab{}.
\newblock \showarticletitle{Uncertainty-aware demand management of water
  distribution networks in deregulated energy markets}.
\newblock \bibinfo{journal}{\emph{Environmental modelling \& software}}
  \bibinfo{volume}{101} (\bibinfo{year}{2018}), \bibinfo{pages}{10--22}.
\newblock


\bibitem[Stuhlmacher and Mathieu(2020a)]%
        {chance_constrained_pumpung_stuhlmacher2020chance}
\bibfield{author}{\bibinfo{person}{Anna Stuhlmacher} {and}
  \bibinfo{person}{Johanna~L Mathieu}.} \bibinfo{year}{2020}\natexlab{a}.
\newblock \showarticletitle{Chance-constrained water pumping to manage water
  and power demand uncertainty in distribution networks}.
\newblock \bibinfo{journal}{\emph{Proc. IEEE}} \bibinfo{volume}{108},
  \bibinfo{number}{9} (\bibinfo{year}{2020}), \bibinfo{pages}{1640--1655}.
\newblock


\bibitem[Stuhlmacher and Mathieu(2020b)]%
        {chance_constrained_water_distribution_stuhlmacher2020water}
\bibfield{author}{\bibinfo{person}{Anna Stuhlmacher} {and}
  \bibinfo{person}{Johanna~L Mathieu}.} \bibinfo{year}{2020}\natexlab{b}.
\newblock \showarticletitle{Water distribution networks as flexible loads: A
  chance-constrained programming approach}.
\newblock \bibinfo{journal}{\emph{Electric Power Systems Research}}
  \bibinfo{volume}{188} (\bibinfo{year}{2020}), \bibinfo{pages}{106570}.
\newblock


\bibitem[Sun et~al\mbox{.}(2021)]%
        {EV_charging_sun2021data}
\bibfield{author}{\bibinfo{person}{Chenxi Sun}, \bibinfo{person}{Tongxin Li},
  {and} \bibinfo{person}{Xiaoying Tang}.} \bibinfo{year}{2021}\natexlab{}.
\newblock \showarticletitle{Data-driven Electric Vehicle Charging Station
  Placement for Incentivizing Potential Demand}. In
  \bibinfo{booktitle}{\emph{2021 IEEE International Conference on
  Communications, Control, and Computing Technologies for Smart Grids
  (SmartGridComm)}}. IEEE, \bibinfo{pages}{27--32}.
\newblock


\bibitem[Sun and Huang(2022)]%
        {predicting_carbon_intensity_sun2022predictions}
\bibfield{author}{\bibinfo{person}{Wei Sun} {and} \bibinfo{person}{Chenchen
  Huang}.} \bibinfo{year}{2022}\natexlab{}.
\newblock \showarticletitle{Predictions of carbon emission intensity based on
  factor analysis and an improved extreme learning machine from the perspective
  of carbon emission efficiency}.
\newblock \bibinfo{journal}{\emph{Journal of Cleaner Production}}
  \bibinfo{volume}{338} (\bibinfo{year}{2022}), \bibinfo{pages}{130414}.
\newblock


\bibitem[Tang and Daoutidis(2022)]%
        {data_driven_control_tang2022data}
\bibfield{author}{\bibinfo{person}{Wentao Tang} {and}
  \bibinfo{person}{Prodromos Daoutidis}.} \bibinfo{year}{2022}\natexlab{}.
\newblock \showarticletitle{Data-driven control: Overview and perspectives}. In
  \bibinfo{booktitle}{\emph{2022 American Control Conference (ACC)}}. IEEE,
  \bibinfo{pages}{1048--1064}.
\newblock


\bibitem[Taylor et~al\mbox{.}(2020)]%
        {L4safety_critical_control_taylor2020learning}
\bibfield{author}{\bibinfo{person}{Andrew Taylor}, \bibinfo{person}{Andrew
  Singletary}, \bibinfo{person}{Yisong Yue}, {and} \bibinfo{person}{Aaron
  Ames}.} \bibinfo{year}{2020}\natexlab{}.
\newblock \showarticletitle{Learning for safety-critical control with control
  barrier functions}. In \bibinfo{booktitle}{\emph{Learning for Dynamics and
  Control}}. PMLR, \bibinfo{pages}{708--717}.
\newblock


\bibitem[Tutterow and McKane(2004)]%
        {variable_speed_pump}
\bibfield{author}{\bibinfo{person}{Vestal Tutterow} {and}
  \bibinfo{person}{Aimee~T McKane}.} \bibinfo{year}{2004}\natexlab{}.
\newblock \bibinfo{title}{Variable speed pumping: A guide to successful
  applications}.
\newblock \bibinfo{howpublished}{US DOE. Lawrence Berkeley National Laboratory.
  \url{https://escholarship.org/uc/item/4691d71q}}.
\newblock


\bibitem[Volk(2013)]%
        {volk2013pump}
\bibfield{author}{\bibinfo{person}{Michael Volk}.}
  \bibinfo{year}{2013}\natexlab{}.
\newblock \bibinfo{booktitle}{\emph{Pump characteristics and applications}}.
\newblock \bibinfo{publisher}{CRC Press}.
\newblock


\bibitem[Wabersich et~al\mbox{.}(2021)]%
        {wabersich2021probabilistic}
\bibfield{author}{\bibinfo{person}{Kim~P Wabersich}, \bibinfo{person}{Lukas
  Hewing}, \bibinfo{person}{Andrea Carron}, {and} \bibinfo{person}{Melanie~N
  Zeilinger}.} \bibinfo{year}{2021}\natexlab{}.
\newblock \showarticletitle{Probabilistic model predictive safety certification
  for learning-based control}.
\newblock \bibinfo{journal}{\emph{IEEE Trans. Automat. Control}}
  \bibinfo{volume}{67}, \bibinfo{number}{1} (\bibinfo{year}{2021}),
  \bibinfo{pages}{176--188}.
\newblock


\bibitem[Wang et~al\mbox{.}(2021)]%
        {MPC_real_time_water_operation_wang2021minimizing}
\bibfield{author}{\bibinfo{person}{Ye Wang}, \bibinfo{person}{Kevin Too~Yok},
  \bibinfo{person}{Wenyan Wu}, \bibinfo{person}{Angus~R Simpson},
  \bibinfo{person}{Erik Weyer}, {and} \bibinfo{person}{Chris Manzie}.}
  \bibinfo{year}{2021}\natexlab{}.
\newblock \showarticletitle{Minimizing pumping energy cost in real-time
  operations of water distribution systems using economic model predictive
  control}.
\newblock \bibinfo{journal}{\emph{Journal of Water Resources Planning and
  Management}} \bibinfo{volume}{147}, \bibinfo{number}{7}
  (\bibinfo{year}{2021}), \bibinfo{pages}{04021042}.
\newblock


\bibitem[Wanjiru et~al\mbox{.}(2016)]%
        {MPC_energy_water_management_wanjiru2016model}
\bibfield{author}{\bibinfo{person}{Evan~M Wanjiru}, \bibinfo{person}{Lijun
  Zhang}, {and} \bibinfo{person}{Xiaohua Xia}.}
  \bibinfo{year}{2016}\natexlab{}.
\newblock \showarticletitle{Model predictive control strategy of energy-water
  management in urban households}.
\newblock \bibinfo{journal}{\emph{Applied Energy}}  \bibinfo{volume}{179}
  (\bibinfo{year}{2016}), \bibinfo{pages}{821--831}.
\newblock


\bibitem[Wei et~al\mbox{.}(2021)]%
        {wei2021provably}
\bibfield{author}{\bibinfo{person}{Honghao Wei}, \bibinfo{person}{Xin Liu},
  {and} \bibinfo{person}{Lei Ying}.} \bibinfo{year}{2021}\natexlab{}.
\newblock \showarticletitle{A provably-efficient model-free algorithm for
  constrained markov decision processes}.
\newblock \bibinfo{journal}{\emph{arXiv preprint arXiv:2106.01577}}
  (\bibinfo{year}{2021}).
\newblock


\bibitem[Wong et~al\mbox{.}(2022)]%
        {hvac_optimizing_wong2022optimizing}
\bibfield{author}{\bibinfo{person}{William Wong}, \bibinfo{person}{Praneet
  Dutta}, \bibinfo{person}{Octavian Voicu}, \bibinfo{person}{Yuri Chervonyi},
  \bibinfo{person}{Cosmin Paduraru}, {and} \bibinfo{person}{Jerry Luo}.}
  \bibinfo{year}{2022}\natexlab{}.
\newblock \showarticletitle{Optimizing industrial hvac systems with
  hierarchical reinforcement learning}.
\newblock \bibinfo{journal}{\emph{arXiv preprint arXiv:2209.08112}}
  (\bibinfo{year}{2022}).
\newblock


\bibitem[Yang and Ren(2022)]%
        {LAAU_yang2022learning}
\bibfield{author}{\bibinfo{person}{Jianyi Yang} {and} \bibinfo{person}{Shaolei
  Ren}.} \bibinfo{year}{2022}\natexlab{}.
\newblock \showarticletitle{Learning-Assisted Algorithm Unrolling for Online
  Optimization with Budget Constraints}.
\newblock \bibinfo{journal}{\emph{AAAI}} (\bibinfo{year}{2022}).
\newblock


\bibitem[Yang et~al\mbox{.}(2020)]%
  {Conservative_ProjectBasedConstrainedPolicyOptimizatino_ICLR_2020_Yang2020Projection-Based}
\bibfield{author}{\bibinfo{person}{Tsung-Yen Yang}, \bibinfo{person}{Justinian
  Rosca}, \bibinfo{person}{Karthik Narasimhan}, {and} \bibinfo{person}{Peter~J.
  Ramadge}.} \bibinfo{year}{2020}\natexlab{}.
\newblock \showarticletitle{Projection-Based Constrained Policy Optimization}.
  In \bibinfo{booktitle}{\emph{International Conference on Learning
  Representations}}.
\newblock
\urldef\tempurl%
\url{https://openreview.net/forum?id=rke3TJrtPS}
\showURL{%
\tempurl}


\bibitem[Yang et~al\mbox{.}(2022)]%
        {Conservative_RL_Bandits_LiweiWang_SimonDu_ICLR_2022_yang2022a}
\bibfield{author}{\bibinfo{person}{Yunchang Yang}, \bibinfo{person}{Tianhao
  Wu}, \bibinfo{person}{Han Zhong}, \bibinfo{person}{Evrard Garcelon},
  \bibinfo{person}{Matteo Pirotta}, \bibinfo{person}{Alessandro Lazaric},
  \bibinfo{person}{Liwei Wang}, {and} \bibinfo{person}{Simon~Shaolei Du}.}
  \bibinfo{year}{2022}\natexlab{}.
\newblock \showarticletitle{A Reduction-Based Framework for Conservative
  Bandits and Reinforcement Learning}. In
  \bibinfo{booktitle}{\emph{International Conference on Learning
  Representations}}.
\newblock
\urldef\tempurl%
\url{https://openreview.net/forum?id=AcrlgZ9BKed}
\showURL{%
\tempurl}


\bibitem[Yu et~al\mbox{.}(2020)]%
        {power_of_predictions_yu2020power}
\bibfield{author}{\bibinfo{person}{Chenkai Yu}, \bibinfo{person}{Guanya Shi},
  \bibinfo{person}{Soon-Jo Chung}, \bibinfo{person}{Yisong Yue}, {and}
  \bibinfo{person}{Adam Wierman}.} \bibinfo{year}{2020}\natexlab{}.
\newblock \showarticletitle{The power of predictions in online control}.
\newblock \bibinfo{journal}{\emph{Advances in Neural Information Processing
  Systems}}  \bibinfo{volume}{33} (\bibinfo{year}{2020}),
  \bibinfo{pages}{1994--2004}.
\newblock


\bibitem[Yu et~al\mbox{.}(2022)]%
        {competitive_control_yu2022competitive}
\bibfield{author}{\bibinfo{person}{Chenkai Yu}, \bibinfo{person}{Guanya Shi},
  \bibinfo{person}{Soon-Jo Chung}, \bibinfo{person}{Yisong Yue}, {and}
  \bibinfo{person}{Adam Wierman}.} \bibinfo{year}{2022}\natexlab{}.
\newblock \showarticletitle{Competitive control with delayed imperfect
  information}. In \bibinfo{booktitle}{\emph{2022 American Control Conference
  (ACC)}}. IEEE, \bibinfo{pages}{2604--2610}.
\newblock


\bibitem[Zarzycki and {\L}awry{\'n}czuk(2024)]%
        {LSTM_MPC_zarzycki2024lstm}
\bibfield{author}{\bibinfo{person}{Krzysztof Zarzycki} {and}
  \bibinfo{person}{Maciej {\L}awry{\'n}czuk}.} \bibinfo{year}{2024}\natexlab{}.
\newblock \showarticletitle{LSTM for Modelling and Predictive Control of
  Multivariable Processes}. In \bibinfo{booktitle}{\emph{International
  Conference on Innovative Techniques and Applications of Artificial
  Intelligence}}. Springer, \bibinfo{pages}{74--87}.
\newblock


\bibitem[Zhang et~al\mbox{.}(2021)]%
        {regret_LQR_predictions_zhang2021regret}
\bibfield{author}{\bibinfo{person}{Runyu Zhang}, \bibinfo{person}{Yingying Li},
  {and} \bibinfo{person}{Na Li}.} \bibinfo{year}{2021}\natexlab{}.
\newblock \showarticletitle{On the regret analysis of online LQR control with
  predictions}. In \bibinfo{booktitle}{\emph{2021 American Control Conference
  (ACC)}}. IEEE, \bibinfo{pages}{697--703}.
\newblock


\end{thebibliography}

\clearpage
\appendix
\section*{Appendix}
\section{Additional Numerical Results}
In this section, we give more numerical results of the case study in Section \ref{sec:water_experiment}. We first give more details on the testing loss for different training epochs and safety parameter $\lambda$ followed by the maximum risk ratio with different $\lambda$. Then, we provide an ablation study on the impact of different control priors on \ouralg.  Next, we show the safety violation probability under the OOD setting. Finally, we give an instance study to explain \ouralg intuitively. 
\subsection{Training and Testing Details}
\begin{figure*}[t]	
	\centering
 \subfigure[Average cost w/ Epoch]{
	\includegraphics[width=0.31\textwidth]{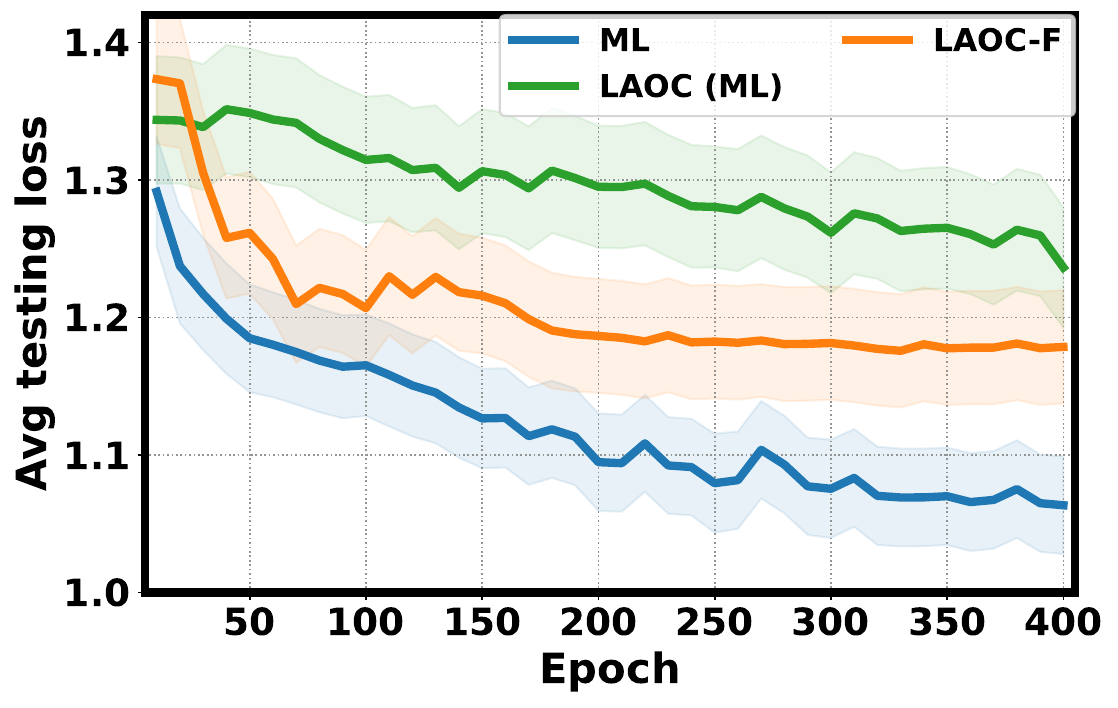}\label{fig:convergence}
	}
 	\subfigure[Average Loss w/ $\lambda$]{
	\includegraphics[width=0.31\textwidth]{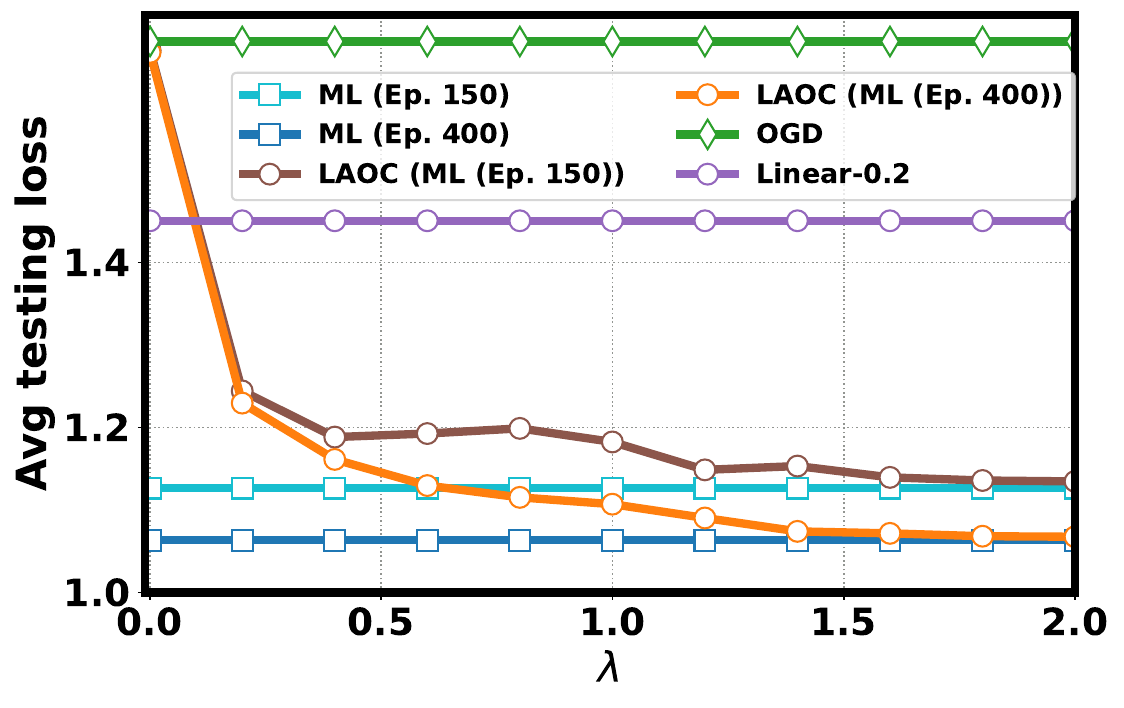}\label{fig:cost_ratio}
	}
 	\subfigure[Max risk ratio]{
	\includegraphics[width=0.31\textwidth]{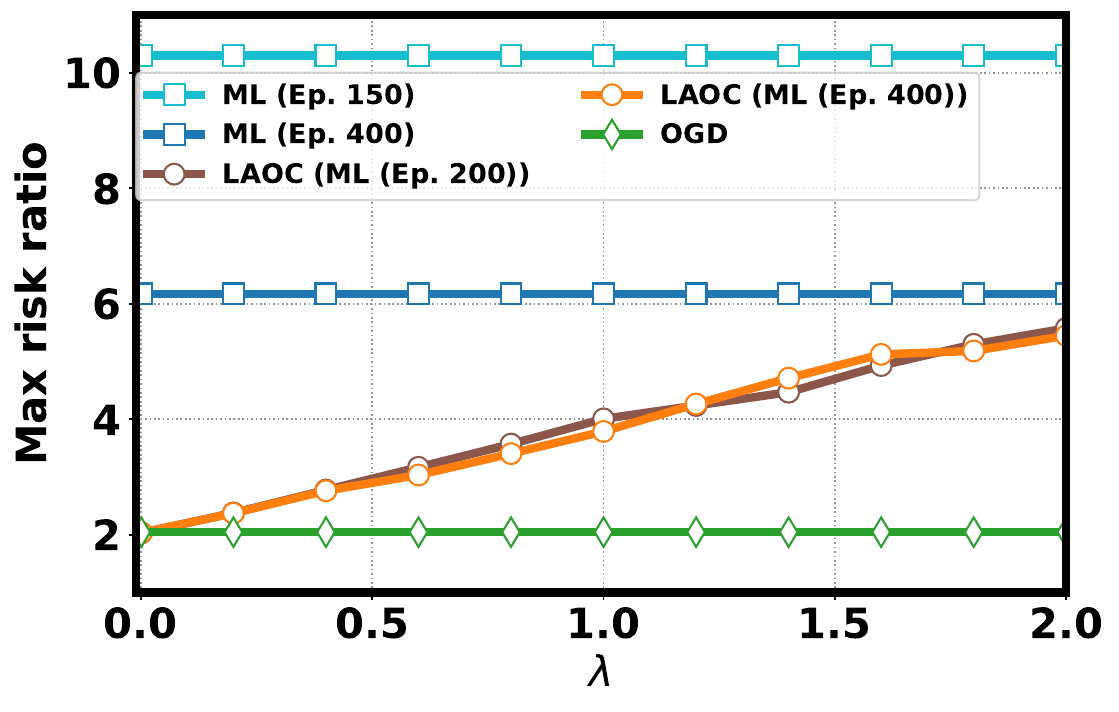}\label{fig:risk_ratio}
	}
\vspace{-0.2cm}	
\caption{\small{Average testing loss and the maximum risk ratio. By default, \ogd is the control prior for \ouralg.  ML (Ep. $N$) is the ML model at the $N$th epoch. \ouralg (ML (Ep. $N$)) is \ouralg using the purely-trained ML model at the $N$th epoch. \ouralg-F is \ouralg with safety-aware finetuning \eqref{eqn:training}.} }\label{fig:cost_split}
\vspace{-0.5cm}
\end{figure*} 
\subsubsection{Convergence}
In Figure \ref{fig:convergence}, we show the average testing losses as the training evolves. We show the training sequences of pure ML (blue curve) and the safety-aware fine-tuning of \ouralg-F (orange curve), respectively. \ouralg(ML) (green curve) takes the purely trained ML model at the corresponding epoch as input. 
The average loss is normalized by the average loss of the optimal policy, i.e. $\mathbb{E}[J_H^{\pi}]/\mathbb{E}[J_H^{*}]$. The testing losses converge after 400 epochs.  We can find that ML purely trained without considering safety has the best testing loss convergence. Due to the safety constraint, the testing loss of \ouralg(ML) with the purely-trained ML model as input increases a lot. By the safety-aware finetuning in \eqref{eqn:training}, \ouralg-F effectively reduces the testing loss of \ouralg because the safety-aware finetuning is performed on an objective that takes the safety set \eqref{eqn:robust_constraint} into consideration, which validates the conclusion in Theorem \ref{thm:sublinear_main}.

\subsubsection{Testing loss with respect to $\lambda$}
Figure \ref{fig:cost_ratio} shows the the average testing loss changing with the safety parameter $\lambda$ in the safety constraint \eqref{eqn:safety_constraint} for \ouralg. The average testing loss is the weighted combination of the energy cost, carbon cost and the deviation penalty, and is normalized by the average loss of the optimal policy. When $\lambda$ becomes larger, $(1+\lambda)-$safety constraint \eqref{eqn:safety_constraint} becomes less strict, the average loss of \ouralg approaches the average loss of corresponding purely-trained ML. When $\lambda=0$, the safety constraint is the strictest and \ouralg reduces to the control prior \ogd. These observations coincide with the cost bound in Theorem \ref{thm:optimal}. 
Additionally, we evaluate the average testing loss of \lin-0.2 and find that although \lin-0.2 has low safety violation probability, it is so conservative that average loss is very high. These validate the superiority of \ouralg in achieving a low enough average loss while guaranteeing safety.

\subsubsection{Maximum risk ratio with respect to $\lambda$}
In Figure \ref{fig:risk_ratio}, we show the worst-case risk ratio changing with the safety parameter $\lambda$ in the safety constraint \eqref{eqn:safety_constraint}. If the safety requirement parameter $\lambda$ becomes larger, \ouralg will take greater risks. Nevertheless, the risk is still lower than purely-trained ML even with very large $\lambda$. 
These results show the advantage of \ouralg in decarbonizing the water supply systems under the safety guarantee.

\begin{figure*}[t]	
	\centering
	\subfigure[Average loss]{
\includegraphics[width=0.31\textwidth]{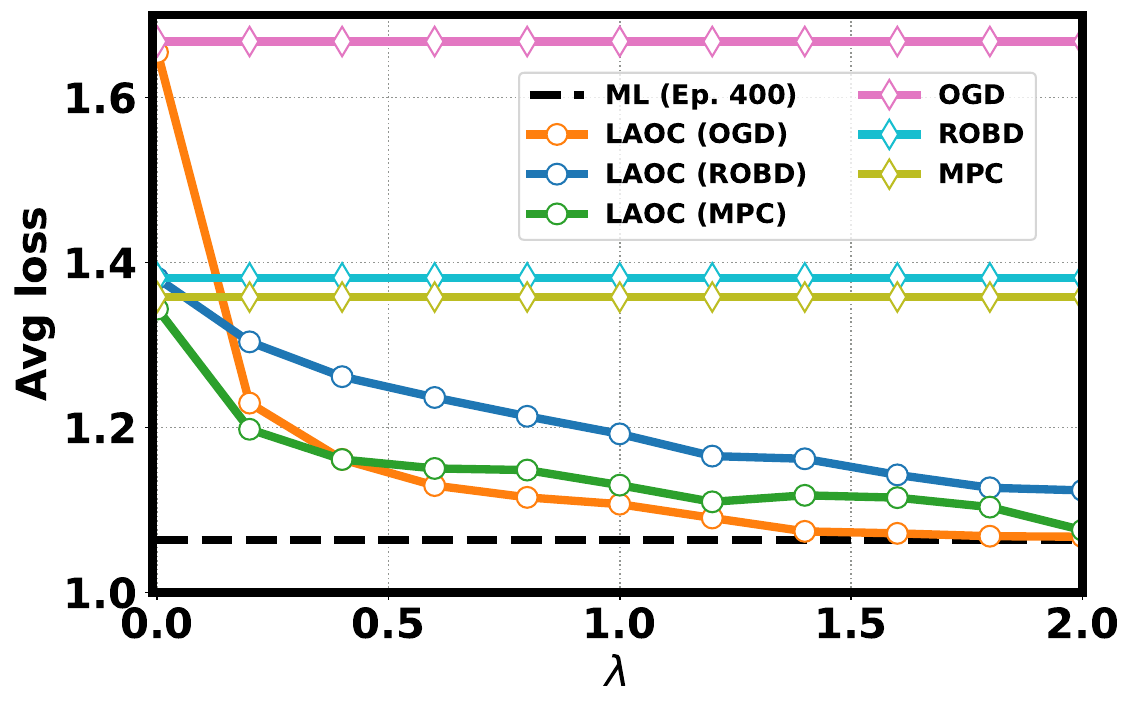}
	}
	\subfigure[Average carbon emission]{
	\includegraphics[width=0.31\textwidth]{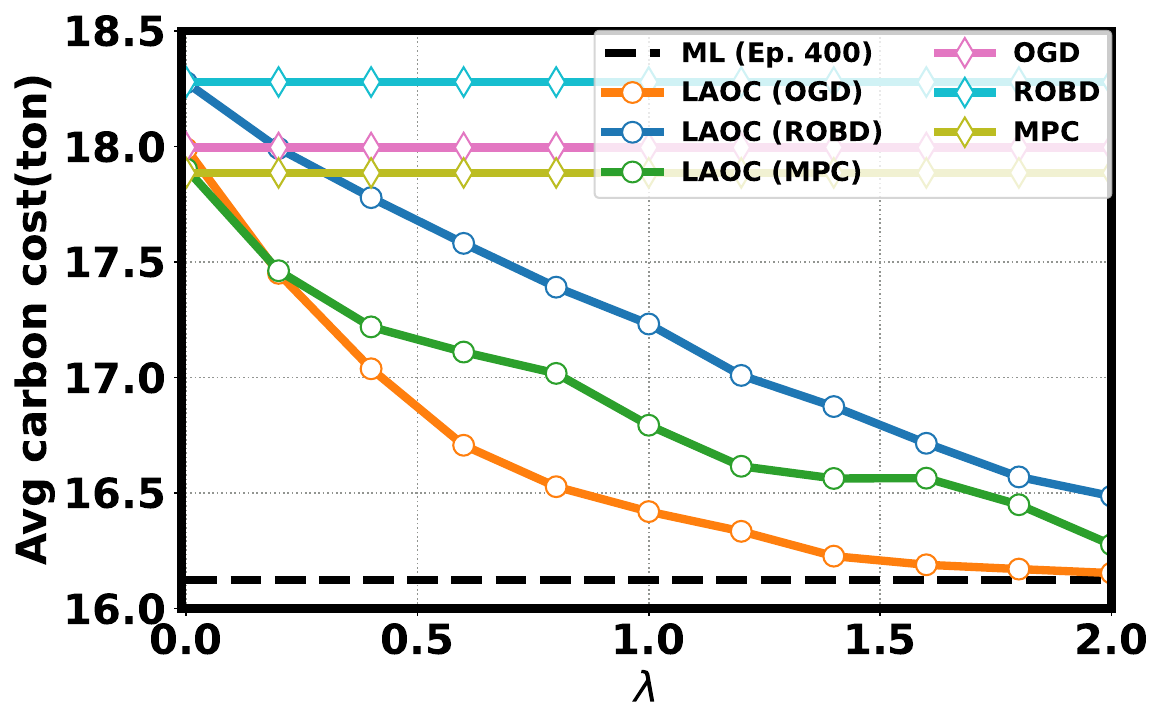}
	}
	\subfigure[Average energy cost]{
	\includegraphics[width=0.31\textwidth]{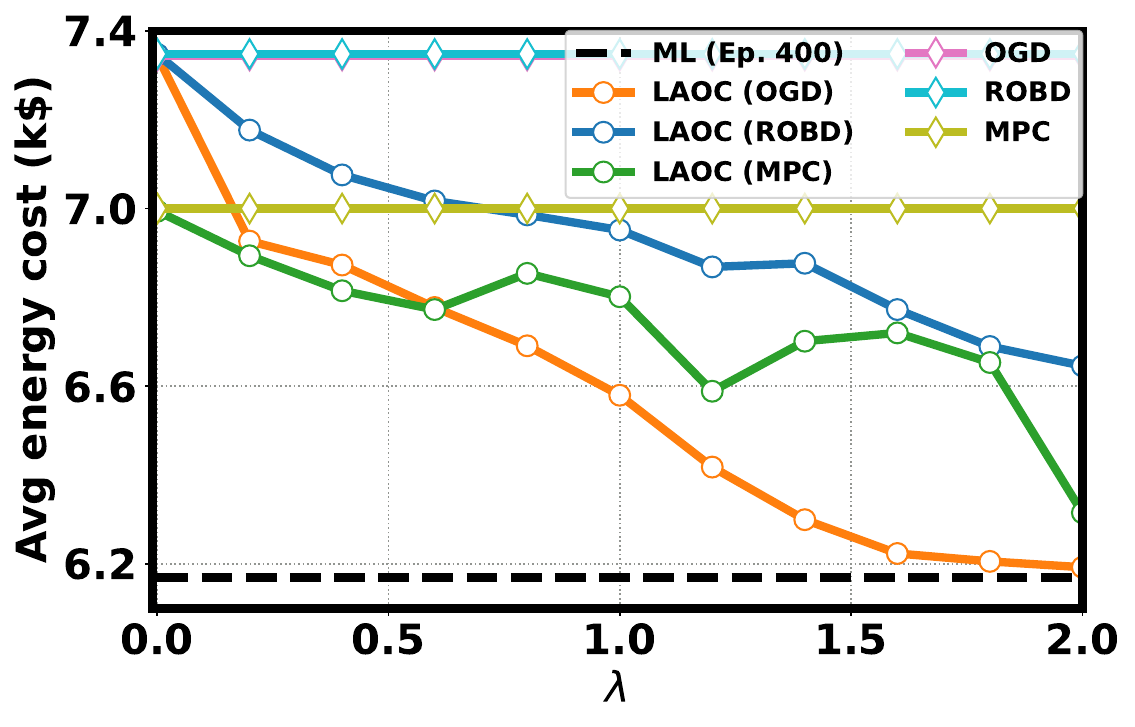}
	}
\vspace{-0.4cm}	
\caption{\small{Average loss, carbon cost and energy cost for different \ouralg algorithms and control priors. \ouralg algorithms use purely-trained ML model at Epoch 400. Here, \mpc represents \mpc-0.03 with a prediction error of 0.03. } }\label{fig:cost_priors}
\end{figure*}

\subsection{\ouralg with different control priors}
In Figure \ref{fig:cost_priors}, we give the average costs of \ouralg using different control priors (\ogd,\robd,\mpc). Here, \mpc represents \mpc-0.03 with a generated prediction error of 0.03. \mpc-0.03 can achieve a maximum risk ratio of 2.52, an average carbon cost of 17782 kg, and an average energy cost of 6924 \$. By the performance bound in Theorem \ref{thm:optimal}, the expected loss is affected by the per-round risk performance of the control prior $r_h^{\dagger}$ and the action discrepancy $\delta_h$ between the
pure ML action and the control prior.  As shown in Table \ref{table:comparison_all_results}, \robd has the lowest worst-case risk which defines the most stringent safety constraint, so \ouralg (\robd) has larger average loss and larger carbon/energy costs than \ouralg with the other two priors. 
We also observe that although \ogd has the largest average carbon/energy costs, \ouralg(\ogd) can achieve low carbon/energy costs a when $\lambda$ is slightly larger. This is because the safety constraint is defined by the risk of \ogd which is higher than that of \robd. No matter which control prior is considered, \ouralg can always guarantee the ($1+\lambda$)-safety constraint with respect to the control prior.

\subsection{Safety Violation Probability Under OOD Setting}
\begin{figure}
\includegraphics[width=0.3\textwidth]{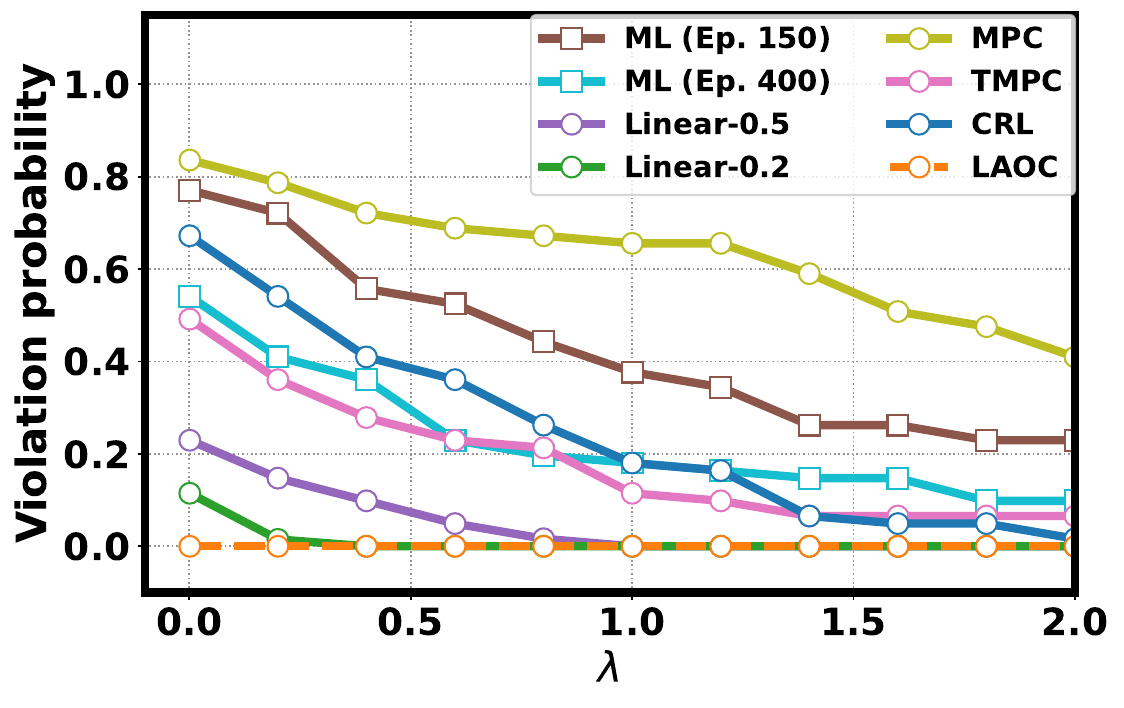}
\vspace{-0.3cm}
\caption{Safety Violation Probability Under OOD Setting. }
\label{fig:violation_rate_ood}
\vspace{-0.5cm}
\end{figure}
Under the OOD seeting, the violation rates of safety constraint \eqref{eqn:safety_constraint} with respect to the control prior \ogd are given in Figure \ref{fig:violation_rate_ood}. A higher $\lambda$ in ($1+\lambda$)-safety in \eqref{eqn:safety_constraint} gives a less strict safety constraint, so the violation probability decreases with $\lambda$. We can observe that \mpclstm is largely affected by the distribution shift and has the highest safety violation probability. \tmpc reduces the violation probability but still has a large violation probability. Both \rl and \crl have non-zero violation probability. We can find that the violation probability of \crl is even larger than the violation probability of \rl when $\lambda$ is small. The ineffectiveness of \crl is because \crl guarantees an expected constraint on the training distribution but the testing distribution has been very different from the training distribution. 

As a learning-augmented design, \lin can achieve low safety constraint violation rate by choosing a small enough combination weight, but this results in a large increase of average costs shown in Table \ref{table:comparison_all_results_ood}. By contrast, even in the OOD setting, \ouralg never violates safety constraint given any problem instance and any safety requirement parameter $\lambda$, which validates the effectiveness of \ouralg in strictly guaranteeing the safety constraint as proved in Theorem \ref{thm:safety}. 

\subsection{Instance study}
\begin{figure*}[t]	
	\centering
	\subfigure[Context sequence]{
\includegraphics[width=0.36\textwidth]{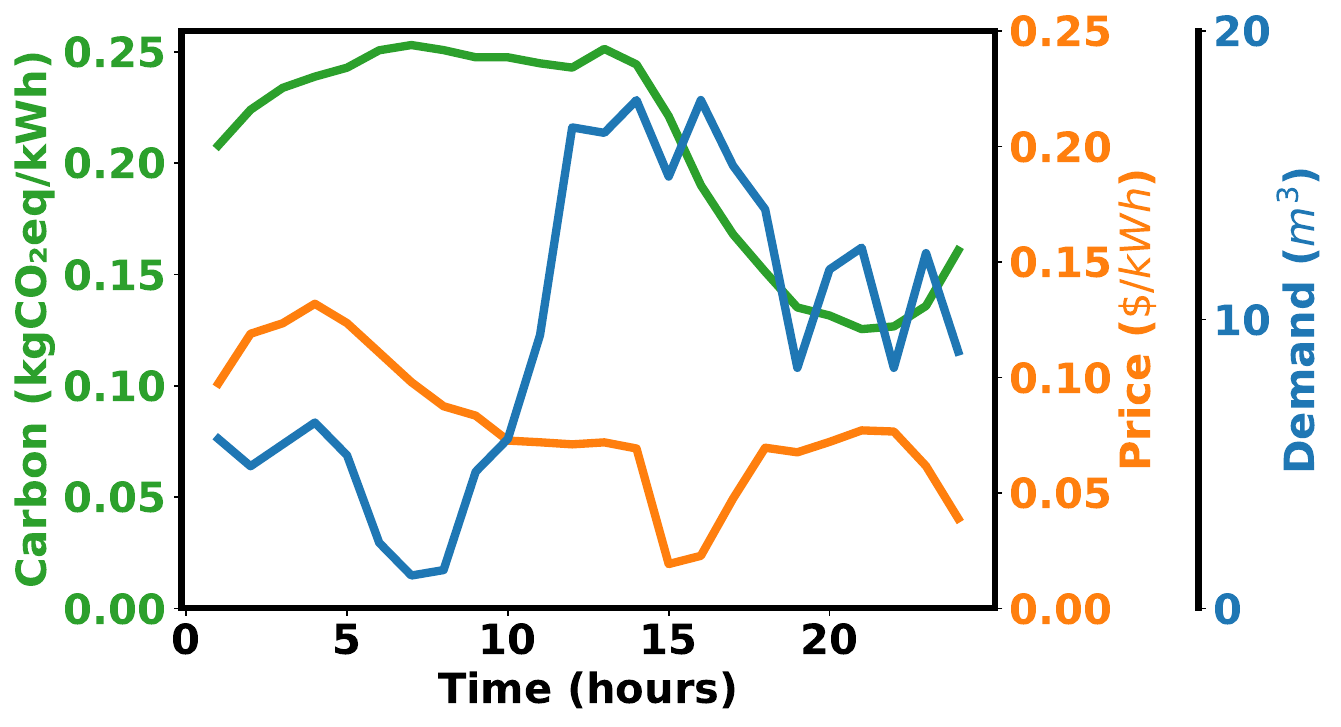}\label{fig:context_seq}
	}
	\subfigure[Action sequence]{
	\includegraphics[width=0.29\textwidth]{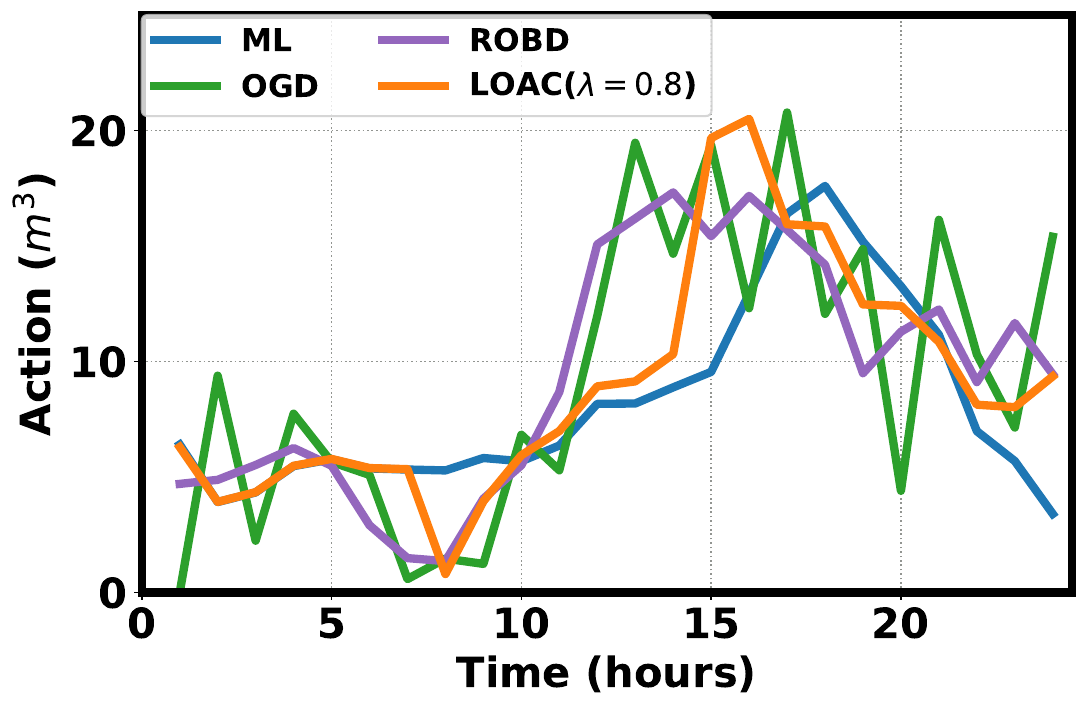}\label{fig:action_seq}
	}
	\subfigure[Water level sequence]{
	\includegraphics[width=0.29\textwidth]{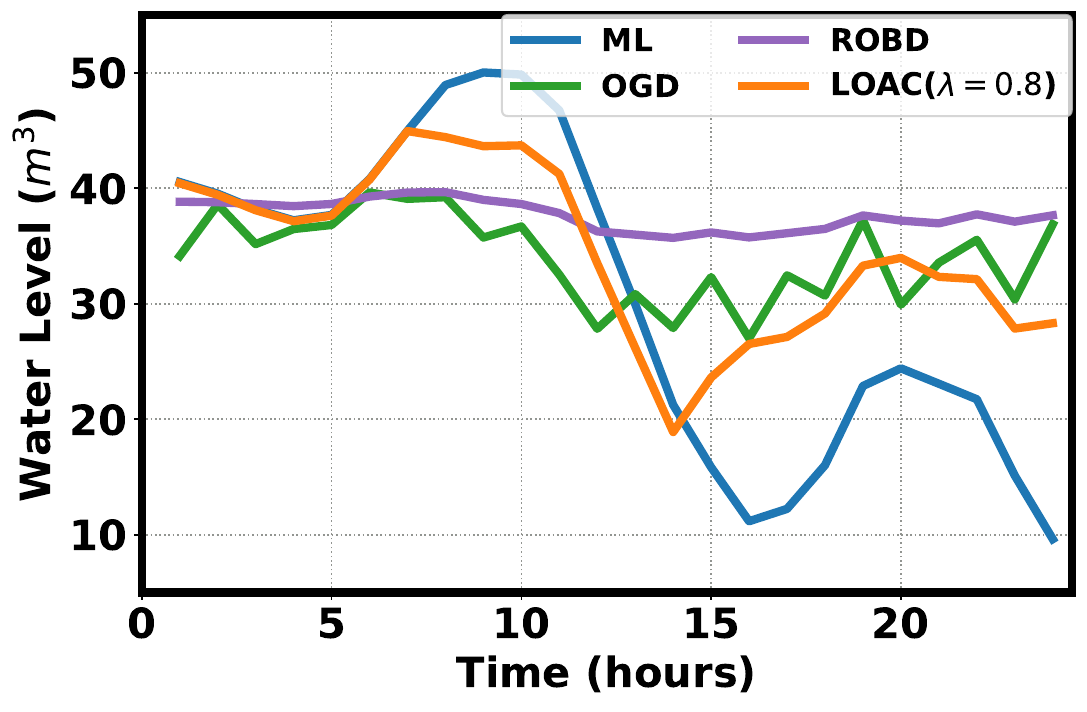}\label{fig:state_seq}
	}
\vspace{-0.4cm}	
\caption{\small{A sequence snapshot of 24 hours. Contexts include carbon intensity, energy price and demand. } }\label{fig:snap}
\end{figure*}
In Figure \ref{fig:snap}, we give a snapshot of a problem instance with 24 hours to get better intuitions on the control process. Figure \ref{fig:context_seq} shows the traces of carbon intensity, energy price, and water demand of the instance. From Figure \ref{fig:action_seq}, we can observe that ML chooses to delay the water supply when the carbon intensity or energy price is high. ML tends to schedule a large water supply when the carbon intensity or energy price is relatively low. This shows the effectiveness of ML policy in utilizing the water tank to save the energy costs by buffering the demand.   However, from Figure \ref{fig:state_seq}, we can find that the water level of ML can be very low at some hour. In this instance, the water level by ML can reduce to 10 $m^3$ which is much lower than the nominal safe water level $\bar{x}=40 m^3$. This results in a high safety risk since the water is not enough when there is an emergency in the building. Comparably, the control priors \ogd and \robd take much more conservative action shown in Figure \ref{fig:action_seq} and maintain the nominal water level very well shown in Figure \ref{fig:state_seq}.  However, they are limited in predicting and exploiting the time-varying energy price and carbon intensity, thus ineffective in saving energy costs and reducing carbon emissions. Different from them, the proposed algorithm \ouralg ($\lambda=0.8$) achieves a good trade-off between safety and costs. It can maintain a water level not far from nominal water level (orange curve in Figure \ref{fig:state_seq}), so the safety risk of \ouralg is low. At the same time, \ouralg regulates the water supply aware of the time-varying carbon intensity and energy price (orange curve in Figure \ref{fig:action_seq}), so it is also effective in saving energy costs and reducing carbon emissions.

\section{Proof of Proposition~\ref{thm:necessary_competitiveness}}
\begin{proof}
We prove by providing a contradictory example.  In this example, the dynamic function is linear, i.e. $f_h(x,u)=\sigma_x x_h +\sigma_u u_h +w_h$, and the control prior has a competitive ratio of $\eta_{\pi^{\dagger}}$ (i.e. $\frac{R_H^{\dagger}}{R_H^*}\leq \eta_{\pi^{\dagger}}$). We prove that at least for this example, $\lambda-$ competitiveness is not guaranteed by \lin.

If \lin guarantees $\lambda-$ competitiveness, since the competitive ratio of $\pi^{\dagger}$ is $\eta_{\pi^{\dagger}}$, we must have
\begin{equation}\label{eqn:competitiveness}
R_H^{\lin}\leq (1+\lambda) R_H^{\pi^{\dagger}}\leq (1+\lambda)\eta_{\pi^{\dagger}}R_H^{\pi^{*}}.
\end{equation}
Since the risk function $r_h$ is $\alpha-$strongly convex and the dynamic function is linear, the total risk $R_H$ is also strongly convex with parameter $\alpha$. By the smoothness of the cost function, we have $\bigtriangledown_{u^*} R_H^{\pi^*} = 0 $, and so 
\begin{equation}\label{eqn:strongconvexity}
R_H^{\lin}\geq R_H^{\pi^*} + \frac{\alpha}{2}\|\rho\tilde{u}+(1-\rho)u^{\dagger}-u^*\|^2,
\end{equation}
where the inequality holds by $\alpha-$strongly convexity of $R_H(u)$. Substituting \eqref{eqn:strongconvexity} into \eqref{eqn:competitiveness}, we have
\begin{equation}\label{eqn:necessaryproof1}
\frac{\alpha}{2}\|\rho\tilde{u}+(1-\rho)u^{\dagger}-u^*\|^2 \leq \left((1+\lambda)\eta_{\pi^{\dagger}}-1\right)R_H^{\pi^{*}},
\end{equation}
and by moving items and the triangle inequality, we have
\begin{equation}\label{eqn:necessaryproof2}
\|\rho(\tilde{u}-u^*)\|-\|(1-\rho)(u^{\dagger}-u^*)\| \leq \sqrt{\frac{2}{\alpha}\left((1+\lambda)\eta_{\pi^{\dagger}}-1\right)R_H^{\pi^{*}}}.
\end{equation}
Applying the $\alpha-$strongly convex of $R_H(u)$ and $\bigtriangledown_{u^*} R_H^{\pi^*} = 0 $ again, we have
\begin{equation}\label{eqn:strongconvexity_2}
R_H^{\pi^{\dagger}}\geq R_H^{\pi^*} + \frac{\alpha}{2}\|u^{\dagger}-u^*\|^2.
\end{equation}
Substituting \eqref{eqn:strongconvexity_2} into \eqref{eqn:necessaryproof2}, we have
\begin{equation}\label{eqn:necessaryproof3}
\begin{split}
\|\tilde{u}-u^*\|&\leq \frac{1-\rho}{\rho}\sqrt{\frac{2}{\alpha}(R_H^{\pi^{\dagger}}-R_H^{\pi^*})} + \frac{1}{\rho}\sqrt{\frac{2}{\alpha}\left((1+\lambda)\eta_{\pi^{\dagger}}-1\right)R_H^{\pi^{*}}}\\
&\leq \sqrt{\frac{2}{\alpha}}\left( \frac{1-\rho}{\rho}\sqrt{\eta_{\pi^{\dagger}}-1}+\frac{1}{\rho}\sqrt{(1+\lambda)\eta_{\pi^{\dagger}}-1}\right)\sqrt{R_H^{\pi^{*}}}
\end{split}
\end{equation}

If \lin guarantees the $\lambda-$ competitiveness, then the ML advice must satisfy 
\begin{equation}\label{eqn:competitiveness_condition}
\frac{\|\tilde{u}-u^*\|^2}{R_H^*}\leq \frac{2}{\alpha}\left( \frac{1-\rho}{\rho}\sqrt{\eta_{\pi^{\dagger}}-1}+\frac{1}{\rho}\sqrt{(1+\lambda)\eta_{\pi^{\dagger}}-1}\right)^2.
\end{equation}
Given $\rho\in (0,1]$ and finite $\eta_{\pi^{\dagger}}$, the right-hand-side is a finite value. Thus, when $\rho\neq 0$,  for arbitrary ML advice with unbounded $\frac{\|\tilde{u}-u^*\|^2}{R_H^*}$, $\lambda-$ competitiveness is not guaranteed.
\end{proof}

\section{Proof of Proposition~\ref{thm:potentaildesign_concrete}}
\label{sec:robustnessproof}

\begin{lemma}[\cite{convex_course}]\label{eqn:tri-smoothness}
For any convex and $\beta-$ cost function $r$ with respect to its input $
(x,u)$, it holds for a parameter $\lambda>0$ that,
\begin{equation}
\begin{split}
&r(x,u)-(1+\lambda)r(x^{\dagger},u^{\dagger})\leq (1+\frac{1}{\lambda})\left(\frac{\beta}{2}\|x-x^{\dagger}\|^2+\frac{\beta}{2}\|u-u^{\dagger}\|^2\right)
\end{split}
\end{equation}
\end{lemma}

\textbf{Proof of Proposition~\ref{thm:potentaildesign_concrete}}
\begin{proof} 
Note that $\phi_h$ is non-negative. Thus, if the safe action set $\mathcal{U}_{\lambda,h}$ in \eqref{eqn:robust_constraint} is non-empty for each $h\in[H]$, then we can always guarantee the competitiveness in Eqn. \eqref{eqn:constraint} by selecting an action in $\mathcal{U}_{\lambda,h}$ in Algorithm \ref{alg:expert_robust_Q}. 
Then we prove the non-empty of safe action set $\mathcal{U}_{\lambda,h}$ by induction.

First of all,  $\mathcal{U}_{\lambda,0}$ is not empty because $u_0^{\dagger}$ is always in $\mathcal{U}_{\lambda,0}$. Then assuming $\mathcal{U}_{\lambda,h-1}$ is not empty, we prove $\mathcal{U}_{\lambda,h}$ is not empty and at least contain an action $u_h^{\dagger}$ as follows. 
Since $\mathcal{U}_{\lambda,h-1}$ is not empty, we have $u_{h-1}\in \mathcal{U}_{\lambda,h-1}$  by Algorithm \ref{alg:expert_robust_Q}, and it holds that
\begin{equation}
R_{h-1}+\phi_{h-1}(u_{h-1})\leq (1+\lambda) R^{\dagger}_{h-1}.
\end{equation}
Thus if $u_h=u_h^{\dagger}$, we have 
\begin{equation}\label{eqn:robustnessproof1}
\begin{split}
&R_{h-1}+r_t(x_h,u_h^{\dagger})+\phi_{h}(u^{\dagger}_{h})- (1+\lambda)\left( R^{\dagger}_{h-1}+r_{h}(x_h^{\dagger},u^{\dagger}_{h})\right)\\
=&R_{h-1}- (1+\lambda) R^{\dagger}_{h-1}+\left(r_h(x_h,u_h^{\dagger})-(1+\lambda)r_{h}(x_h^{\dagger},u^{\dagger}_{h})\right)+\phi_h(u_h^{\dagger})\\
\leq & -\phi_{h-1}(u_{h-1})+\phi_h(u^{\dagger}_h)+\left(r_h(x_h,u_h^{\dagger})-(1+\lambda)r_{h}(x_h^{\dagger},u^{\dagger}_{h})\right)\\
\leq & -\phi_{h-1}(u_{h-1})+\phi_h(u^{\dagger}_h)+(1+\frac{1}{\lambda})\frac{\beta}{2}\|x_h-x_h^{\dagger}\|^2
\end{split}
\end{equation}
where the second inequality holds by Lemma \ref{eqn:tri-smoothness}. 

Since the reservation cost is chosen as 
\begin{equation}
\phi_h(u)=q_h\|f_h(x_{h},u)-f_h(x_{h}^{\dagger},u_h^{\dagger})\|^2,
\end{equation}
 we have
\begin{equation}\label{eqn:phidifference}
\begin{split}
  &-\phi_{h-1}(u_{h-1})+\phi_h(u^{\dagger}_h)\\
  =& -q_{h-1}\|x_h-x_h^{\dagger}\|^2+q_h\|f_h(x_{h},u_h^{\dagger})-f_h(x_{h}^{\dagger},u_h^{\dagger})\|^2\\
  \leq & (-q_{h-1}+q_{h}\sigma_x^2)\|x_h-x_h^{\dagger}\|^2\\
  \leq & -(1+\frac{1}{\lambda})\frac{\beta}{2}\|x_h-x_h^{\dagger}\|^2,
  \end{split}
\end{equation}
where the first inequality comes from the Lipschitz continuity of dynamic $f_h$, and the second inequality holds by the choice of $q_h=C_1(1+\frac{1}{\lambda})\frac{\beta}{2}\sum_{h'=0}^{H-h-1}(C_2\sigma_x^{2})^{h'}$
for some constant $C_1\geq 1$ and $C_2\geq 1$ such that $q_h\sigma_x^2=C_1(1+\frac{1}{\lambda})\frac{\beta}{2}\sum_{h'=1}^{H-h}C_2^{h'-1}\sigma_x^{2h'}\leq C_1(1+\frac{1}{\lambda})\frac{\beta}{2}\sum_{h'=1}^{H-h}C_2^{h'}\sigma_x^{2h'}=q_{h-1}-C_1(1+\frac{1}{\lambda})\frac{\beta}{2}\leq q_{h-1}-(1+\frac{1}{\lambda})\frac{\beta}{2}$.

Substituting \eqref{eqn:phidifference} into \eqref{eqn:robustnessproof1}, it holds for $u_h=u_h^{\dagger}$ that
$R_{h-1}+r_t(x_h,u_h^{\dagger})+\phi_{h}(u^{\dagger}_{h})\leq (1+\lambda)\left( R^{\dagger}_{h-1}+r_{h}(x_h^{\dagger},u^{\dagger}_{h})\right)$. Therefore, $u_h^{\dagger}$ is in the safe action set $\mathcal{U}_{\lambda,h}$ and so $\mathcal{U}_{\lambda,h}$ is not empty.

Therefore by the discussion at the beginning of this proof, the Proposition is proved.
\end{proof}

\section{Proof of Theorem \ref{thm:optimal}}
We denote the policy \ouralg on the basis of the ML policy $\tilde{\pi}$ and the action set $\mathcal{U}_{\lambda,h}$ as 
\begin{equation}\label{eqn:projection_func}
\pi_{\lambda}(s_h)=m(\tilde{\pi}(s_h), \mathcal{U}_{\lambda,h}),
\end{equation}
where $s_h$ is the ML input at round $h$, and $m$ is the projection function in \eqref{eqn:projection} or the linear function in \eqref{eqn:linear_combination}.
By directly applying the ML policy $\tilde{\pi}$ without projection or linear operations, we get the action sequence $\{\tilde{u}_h', h\in[H]\}$ and the state sequence $\{\tilde{x}_h, h\in[H]\}$, and the corresponding ML inputs (which include $\tilde{x}_h$) are denoted as $\tilde{s}_h$.

\begin{lemma}\label{thm:consistency_sequence}
Given two constants $\lambda_1>0$ and $\lambda_0\in(0,\lambda)$, if the potential function is designed as $\phi_h(u)=q_h\|f_h(x_{h},u_h)-f_t(x_{h}^{\dagger},u_h^{\dagger})\|^2$ with $q_h\geq 0$ satisfying
$
2\sigma_x^2q_h\leq q_{h-1}-(1+\frac{1}{\lambda_0})\frac{\beta}{2}$
for $h\in[H-1]$, $q_H=0$, then $u_h$ is in the competitive action set \eqref{eqn:robust_constraint} if 
\[
\|u_h-u_h^{\dagger}\|^2\leq Gr_h^{\dagger},
\]
where $r_h^{\dagger}$ is the risk of the prior at time $h$, and $G=\frac{(\lambda-\lambda_0)}{(1+\frac{1}{\lambda_0})\frac{\beta}{2}+\frac{C_2}{C_2-1}q_h\sigma_{u}^2}$.
\end{lemma}
\begin{proof}
Note that at time $h-1$, the competitiveness constraint holds as
\begin{equation}
R_{h-1}+\phi_{h-1}(u_{h-1})\leq (1+\lambda)\left( R^{\dagger}_{h-1}\right),
\end{equation}
and the sufficient condition for $u_h\in\mathcal{U}_{\lambda,h}$ is
\begin{equation}\label{eqn:sufficientcondition_1}
\begin{split}
& r_h(x_h,u_h)+\phi_h(u_h)-\phi_{h-1}(u_{h-1})\leq  (1+\lambda)r_h(x_h^{\dagger},u_h^{\dagger}).
\end{split}
\end{equation}
Given $\lambda_0\in (0,\lambda)$, subtracting $(1+\lambda_0)r_h(x_h^{\dagger},u_h^{\dagger})$ by both sides and by Lemma \ref{eqn:tri-smoothness}, we get the sufficient condition that \eqref{eqn:sufficientcondition_1} holds if
\begin{equation}\label{eqn:sufficientcondition_2}
\begin{split}
\left((1+\frac{1}{\lambda_0})\frac{\beta}{2}-q_{h-1}\right)&\|x_h-x_h^{\dagger}\|^2+(1+\frac{1}{\lambda_0})\frac{\beta}{2}\|u_h-u_h^{\dagger}\|^2\\
&+q_h\|x_{h+1}-x_{h+1}^{\dagger}\|^2\leq (\lambda-\lambda_0)r_h^{\dagger}.
\end{split}
\end{equation}
By the Lipschitz continuity of $f$ and the smoothness of squared norm,  we have by Lemma \ref{eqn:tri-smoothness}
\[
\begin{split}
\|x_{h+1}&-x_{h+1}^{\dagger}\|^2=\|f_h(x_{h},u_h)-f(x_{h}^{\dagger},u_h^{\dagger})\|^2\\
&\leq C_2\sigma_x^2\|x_{h}-x_{h}^{\dagger}\|^2+\frac{C_2}{C_2-1}\sigma_u^2\|u_{h}-u_{h}^{\dagger}\|^2.
\end{split}
\] 
Thus, we further get the sufficient condition of \eqref{eqn:sufficientcondition_2} as 
\begin{equation}
\begin{split}
&\left((1+\frac{1}{\lambda_0})\frac{\beta}{2}+C_2q_h\sigma_{x}^2-q_{h-1}\right)\|x_h-x_h^{\dagger}\|^2+
 \\
 &\quad \left((1+\frac{1}{\lambda_0})\frac{\beta}{2}+\frac{C_2}{C_2-1}q_h\sigma_{u}^2 \right)\|u_h-u_h^{\dagger}\|^2\leq (\lambda-\lambda_0)r_h^{\dagger}.
\end{split}
\end{equation}
By the condition that $q_{h-1}\geq (1+\frac{1}{\lambda_0})\frac{\beta}{2}+C_2q_h\sigma_{x}^2$, we get the following:
\begin{equation}
\begin{split}
\left((1+\frac{1}{\lambda_0})\frac{\beta}{2}+\frac{C_2}{C_2-1}q_h\sigma_{u}^2\right)\|u_h-u_h^{\dagger}\|^2\leq (\lambda-\lambda_0)r_h^{\dagger},
\end{split}
\end{equation}
which implies the sufficient condition in this lemma.
\end{proof}

\begin{lemma}\label{lma:q_difference}
With a reservation function satisfying the condition in Lemma \ref{thm:consistency_sequence},  $\xi_h(s_h)=\|\pi_{\lambda}(s_h)-\tilde{\pi}(s_h)\|=\|m(\tilde{\pi}(s_h),\mathcal{U}_{\lambda,h})-\tilde{\pi}(s_h)\|$ with $m$ being the projection function in \eqref{eqn:projection_func} is bounded by
\[
\xi_h(s_h)\leq L_{\pi}\|x_h-\tilde{x}_h\|+\left[\delta_h-(\sqrt{1+\lambda}-1)^2Gr_h^{\dagger}\right]^+,
\]
where $G=\frac{2}{L_c\left(1+\frac{C_2}{C_2-1}\sigma_u^2(1-(C_2\sigma_x^2)^{H-h})/(1-C_2\sigma_x^2)\right)}$ and $\delta_h=\|\tilde{\pi}(\tilde{s}_h)-\pi^{\dagger}(s_h^{\dagger})\|$.
\end{lemma}
\begin{proof}
We choose $q_h=(1+\frac{1}{\lambda_0})\frac{\beta}{2}\sum_{i=0}^{H-h-1}\left(C_2\sigma_x^2\right)^{i}$ given any $\lambda_0\in(0,\lambda)$.  The choice of $q_h$ satisfies the requirement for $q_h$ in Lemma \ref{thm:consistency_sequence}  and the sufficient condition becomes 
\begin{equation}
\begin{split}
\|u_h-u_h^{\dagger}\|^2\leq& \frac{\lambda-\lambda_0}{(1+\frac{1}{\lambda_0})\frac{\beta}{2}}\frac{r_h^{\dagger}}{1+\frac{C_2}{C_2-1}\sigma_u^2\sum_{i=0}^{H-h-1}(C_2\sigma_x^2)^{i}}\\
=& \frac{\lambda-\lambda_0}{(1+\frac{1}{\lambda_0})\frac{\beta}{2}}\frac{r_h^{\dagger}}{1+\frac{C_2}{C_2-1}\sigma_u^2(1-(C_2\sigma_x^2)^{H-h})/(1-C_2\sigma_x^2)}.
\end{split}
\end{equation}
By optimally choosing $\lambda_0=\sqrt{1+\lambda}-1$, we have
\[q_h=(1+\frac{1}{\sqrt{1+\lambda}-1})\frac{\beta}{2}\sum_{i=0}^{H-h-1}\left(C_2\sigma_x^2\right)^{i},\] and 
\begin{equation}
\begin{split}
\|u_h-u_h^{\dagger}\|^2\leq
\frac{\frac{2}{\beta}(\sqrt{1+\lambda}-1)^2r_h^{\dagger}}{1+\frac{C_2}{C_2-1}\sigma_u^2(1-(C_2\sigma_x^2)^{H-h})/(1-C_2\sigma_x^2)}.
\end{split}
\end{equation}
Since $\|u_h-u_h^{\dagger}\|\leq A$ where $A$ is the size of the action set, we get the sufficient condition that an action belongs to safe action set \eqref{eqn:robust_constraint} as 
\begin{equation}
\begin{split}
\|u_h-u_h^{\dagger}\|&\leq
\frac{\frac{2}{L_c}(\sqrt{1+\lambda}-1)^2r_h^{\dagger}}{1+\frac{C_2}{C_2-1}\sigma_u^2(1-(C_2\sigma_x^2)^{H-h})/(1-C_2\sigma_x^2)}\\
&=(\sqrt{1+\lambda}-1)^2Gr_h^{\dagger}.
\end{split}
\end{equation}

Let $G'=(\sqrt{1+\lambda}-1)^2G$. We denote the action projected from $\tilde{\pi}(s_h)$ to the norm ball $\mathcal{B}(u_h^{\dagger},  G'r_h^{\dagger})$ as 
$\pi_{\lambda}^{\perp}(s_h)$. We have $\pi_{\lambda}^{\perp}(s_h)=\tilde{\pi}(s_h)$ if $\tilde{\pi}(s_h)\in \mathcal{B}(u_h^{\dagger}, G'r_h^{\dagger})$. And if $\tilde{\pi}(s_h)\notin \mathcal{B}(u_h^{\dagger}, G'r_h^{\dagger})$, we have $\pi_{\lambda}^{\perp}(s_h)=u_h^{\dagger}+G'r_h^{\dagger}\frac{\tilde{\pi}(s_h)-u_h^{\dagger}}{\|\tilde{\pi}(s_h)-u_h^{\dagger}\|}$. Since $\mathcal{U}_{\lambda,h}$ is a close set, the norm ball $\mathcal{B}(u_h^{\dagger}, G'r_h^{\dagger})\subset \mathcal{U}_{\lambda,h}$ thanks to Lemma \ref{thm:consistency_sequence}, we have 
\begin{equation}
\begin{split}
\xi_h(s_h)=&\left\|\pi_{\lambda}(s_h)-\tilde{\pi}(s_h)\right\|\leq \left\|\pi_{\lambda}^{\perp}(s_h)-\tilde{\pi}(s_h)\right\|\\
=&\left[\|\tilde{\pi}(s_h)-u_h^{\dagger}\|-G'r_h^{\dagger}\right]^+.\\
\leq & \|\tilde{\pi}(\tilde{s_h})-\tilde{\pi}(s_h)\|+\left[\|\tilde{\pi}(\tilde{s_h})-u_h^{\dagger}\|-G'r_h^{\dagger}\right]^+\\
\leq &L_{\pi}\|x_h-\tilde{x}_h\|+\left[\delta_h-G'r_h^{\dagger}\right]^+,
\end{split}
\end{equation}
where the first inequality is because $\pi_{\lambda}$ applies the projection or linear operation $m$ on the ML predictions, the second equality holds because if $\tilde{\pi}(s_h)\notin \mathcal{B}(u_h^{\dagger}, G'r_h^{\dagger})$, $\left\|\pi_{\lambda}^{\perp}(s_h)-\tilde{\pi}(s_h)\right\|=\|u_h^{\dagger}-\tilde{\pi}(s_h)-G'r_h^{\dagger}\frac{\tilde{\pi}(s_h)-u_h^{\dagger}}{\|u_h^{\dagger}-\tilde{\pi}(s_h)\|}\|=\|\tilde{\pi}(s_h)-u_h^{\dagger}\|-G'r_h^{\dagger}$ and if $\tilde{\pi}(s_h)\in \mathcal{B}(u_h^{\dagger}, G'r_h^{\dagger})$, $\left\|\pi_{\lambda}^{\perp}(s_h)-\tilde{\pi}(s_h)\right\|=0$, the second inequality holds by the triangle inequality, and the last inequality holds by the Lipschitz continuity of the policy $\pi^*$ and $\|\tilde{s}_h-s_h\|=
\|\tilde{x}_h-x_h\|$ for the same context instance.
\end{proof}

\begin{lemma}\label{state_perturbation}
With a reservation function satisfying the condition in Lemma \ref{thm:consistency_sequence}, the difference of the states with respect to the policy $\pi_{\lambda}$ and the policy $\tilde{\pi}$ is bounded as
\[
\|\tilde{x}_h-x_h\|\leq \sum_{i=0}^{h-1}(\sigma_x+2\sigma_uL_{\pi})^{h-i-1}\sigma_u\left[\eta_i-(\sqrt{1+\lambda}-1)^2Gr_{i}^{\dagger}\right]^+,
\]
where $G=\frac{2}{L_c\left(1+\frac{C_2}{C_2-1}\sigma_u^2(1-(C_2\sigma_x^2)^{H-h})/(1-C_2\sigma_x^2)\right)}$.
\end{lemma}
\begin{proof}
By the state dynamic function and Lipschitz continuity, we have
\begin{equation}
\begin{split}
\|\tilde{x}_h-x_h\|&=\|f(\tilde{x}_{h-1}',\tilde{\pi}(\tilde{s}_{h-1}))-f(x_{h-1},\pi_{\lambda}(s_{h-1}))\|\\
&\leq \sigma_x\|\tilde{x}_{h-1}'-x_{h-1}\|+\sigma_u\|\tilde{\pi}(\tilde{s}_{h-1})-\pi_{\lambda}(s_{h-1})\|\\
&\leq \sigma_x\|\tilde{x}_{h-1}'-x_{h-1}\|+\sigma_u\|\tilde{\pi}(\tilde{s}_{h-1})-\tilde{\pi}(s_{h-1})\|\\
&\qquad +\sigma_u\|\tilde{\pi}(s_{h-1})-\pi_{\lambda}(s_{h-1}))\|\\
&\leq (\sigma_x+\sigma_u L_{\pi})\|\tilde{x}_{h-1}'-x_{h-1}\|+\sigma_u\xi_{h-1}(s_{h-1}),
\end{split}
\end{equation}
where the second inequality holds by the triangle inequality, and the last inequality holds by the Lipschitz continuity of the policy $\tilde{\pi}$ and $\|\tilde{s}_h-s_h\|=
\|\tilde{x}_h-x_h\|$ for the same context instance. 

Applying Lemma \ref{lma:q_difference} for $\xi_{h-1}(s_{h-1})$, we further have
\begin{equation}\label{eqn:state_iterate}
\begin{split}
\|\tilde{x}_h-x_h\|
&\leq (\sigma_x+2\sigma_u L_{\pi})\|\tilde{x}_{h-1}'-x_{h-1}\|\\
&+\sigma_u\left[\eta_{h-1}-(\sqrt{1+\lambda}-1)^2Gr_{h-1}^{\dagger}\right]^+.
\end{split}
\end{equation}
Iteratively applying \eqref{eqn:state_iterate}, we have
\begin{equation}
\begin{split}
\|\tilde{x}_h-x_h\|
\leq  \sum_{i=0}^{h-1}(\sigma_x+2\sigma_uL_{\pi})^{h-i-1}\sigma_u\left[\eta_i-(\sqrt{1+\lambda}-1)^2Gr_{i}^{\dagger}\right]^+.
\end{split}
\end{equation}
\end{proof}

\textbf{Proof of Theorem \ref{thm:optimal}}
\begin{proof}
Now we are ready to bound the difference of expected costs of \ouralg and the pure ML policy $\tilde{\pi}$ which is 
\begin{equation}\label{eqn:averagedifferencebound1}
\begin{split}
&\mathbb{E}_{y_{0:H}}\left[ J_H^{\pi_{\lambda}}(y_{0:H})\right]-\mathbb{E}_{y_{0:H}}\left[J_H^{\tilde{\pi}}(y_{0:H})\right]\\
=&  \mathbb{E}_{y_{0:H}}\left[ \sum_{h=0}^H c_h\left(x_h,m(\tilde{\pi}(s_h),\mathcal{U}_{\lambda,h})\right)-c_h\left(\tilde{x}_h,\tilde{\pi}(\tilde{s}_h)\right)\right].
\end{split}
\end{equation}

We can bound this difference as
\begin{equation}
\begin{split}
&\mathbb{E}_{y_{0:H}}\left[ J_H^{\pi_{\lambda}}(y_{0:H})\right]-\mathbb{E}_{y_{0:H}}\left[J_H^{\tilde{\pi}}(y_{0:H})\right]\\
= & \mathbb{E}_{y_{0:H}}\left[ \sum_{h=0}^H c_h\left(x_h,m(\tilde{\pi}(s_h),\mathcal{U}_{\lambda,h})\right)-c_h\left(x_h,\tilde{\pi}(s_h)\right)\right.\\
&\qquad \left.+c_h\left(x_h,\tilde{\pi}(s_h)\right)-c_h\left(\tilde{x}_h,\tilde{\pi}(\tilde{s}_h)\right)\right]\\
\leq &L_c \mathbb{E}_{y_{0:H}}\left[\sum_{h=0}^H \|m(\tilde{\pi}(s_h),\mathcal{U}_{\lambda,h})-\tilde{\pi}(\tilde{s}_h)\| + (1+L_{\pi})\|x_h-\tilde{x}_h\| \right]\\
\leq  & L_c \mathbb{E}_{y_{0:H}}\left[\sum_{h=0}^H\xi_h(s_h) + (1+2L_{\pi})\|x_h-\tilde{x}_h\| \right],
\end{split}
\end{equation}
where the first inequality holds because the cost functions $c_h$ are $L_c$-Lipschitz continuous, $\tilde{\pi}$ is $L_{\pi}-$Lipschitz and $\tilde{s}_h-s_h=
\tilde{x}_h-x_h$ for the same context instance.
and the second equality is due to the definition of $
\xi_h(s_h)=\|m(\tilde{\pi}(s_h),\mathcal{U}_{\lambda,h})-\tilde{\pi}(s_h)\|$ in Lemma \ref{lma:q_difference} and $\|\tilde{\pi}(\tilde{s}_h)-\tilde{\pi}(s_h)\|\leq L_{\pi}\|\tilde{s}_h-s_h\|=L_{\pi}\|x_h-\tilde{x}_h\|$.

By Lemma \ref{lma:q_difference}, we can further bound the expected cost difference as 
\begin{equation}
\begin{split}
&\mathbb{E}_{y_{0:H}}\left[ J_H^{\pi_{\lambda}}(y_{0:H})\right]-\mathbb{E}_{y_{0:H}}\left[J_H^{\tilde{\pi}}(y_{0:H})\right]\\
\leq & L_c \mathbb{E}_{y_{0:H}}\left[\sum_{h=0}^H\xi_h(s_h) + (1+2L_{\pi})\|x_h-\tilde{x}_h\| \right]\\
\leq & L_c \mathbb{E}_{y_{0:H}}\left[\sum_{h=0}^{H-1} \left[\delta_h-G'r_h^{\dagger}\right]^+ + (1+2L_{\pi})\sum_{h=0}^{H}\|x_h-\tilde{x}_h\| \right],
\end{split}
\end{equation}
where $G'=(\sqrt{1+\lambda}-1)^2G$, $\delta_h=\|\tilde{\pi}(\tilde{s}_h)-\pi^{\dagger}(s_h^{\dagger})\|$, and $\xi_{H}(s_H)=0$ as there is no action at round $H$.

By Lemma \ref{state_perturbation}, the expected cost is bounded as
\begin{equation}\label{eqn:expected_cost_bound}
\begin{split}
&\mathbb{E}_{y_{0:H}}\left[ J_H^{\pi_{\lambda}}(y_{0:H})\right]-\mathbb{E}_{y_{0:H}}\left[J_H^{\tilde{\pi}}(y_{0:H})\right]\\
\leq & L_c \mathbb{E}_{y_{0:H}}\left[\sum_{h=0}^{H-1} \left[\delta_h-G'r_h^{\dagger}\right]^+ +\right.\\
&\left.(1+2L_{\pi})\sum_{h=1}^H\sum_{i=0}^{h-1}(\sigma_x+2\sigma_uL_{\pi})^{h-i-1}\sigma_u\left[\eta_i-G'r_{i}^{\dagger}\right]^+ \right]\\
\leq & L_c \mathbb{E}_{y_{0:H}}\left[\sum_{h=0}^{H-1} \left[\delta_h-G'r_h^{\dagger}\right]^+ + \right.\\
&\left.(1+2L_{\pi})\sigma_u \sum_{h=0}^{H-1} \left[\delta_h-G'r_h^{\dagger}\right]^+\sum_{i=h}^{H-1}(\sigma_x+2\sigma_uL_{\pi})^{h-i-1} \right]\\
\leq & B\mathbb{E}_{y_{0:H}}\left[\sum_{h=0}^{H-1} \left[\delta_h-(\sqrt{1+\lambda}-1)^2Gr_h^{\dagger}\right]^+ \right]
\end{split}
\end{equation}
where $B=L_c\left(1+(1+2L_{\pi})\sigma_u\sum_{i=0}^{H-1}(\sigma_x+2\sigma_uL_{\pi})^{h-i-1}\right)$.

The reservation function in Lemma \ref{lma:q_difference} meet the requirements in Proposition \ref{thm:potentaildesign_concrete} by choosing some proper constants $\rho$, $C_1$ and $C_2$. 
\end{proof}

\section{Proof of Theorem \ref{thm:sublinear_main}}
\begin{proof}
Since the policy $\pi_{\lambda}^{(n)}$ is one from the constrained policy set $\Pi_{\lambda}$, we apply the statistical generalization theorem in \cite{statistical_learning_bousquet2004introduction} and get with probability at least $1-\delta, \delta\in(0,1)$,
\begin{equation}
\left|\mathbb{E}\left[J_H^{\pi_{\lambda}^{(n)}}\right]- \frac{1}{n}\sum_{t=1}^nJ_H^{\pi_{\lambda}^{(n)}}(y_{0:H}^{(t)})\right|\leq 4HP\sqrt{\frac{2}{n}\ln\frac{4N(\epsilon, \Pi_{\lambda}, \hat{L}_1^n)}{\delta}},
\end{equation}
where $N(\epsilon, \Pi_{\lambda}, \hat{L}_1^n)$ is the $\epsilon-$covering number of the competitive policy space $\Pi_{\lambda}$ with $L_1-$norm as the distance measure: the distance of two functions $\pi$ and $\pi'$ is $\|\pi-\pi'\|_{\hat{L}_1^n}=\frac{1}{n}\sum_{t=1}^n\|\pi(s^{(t)})-\pi'(s^{(t)})\|_1$.

By Eqn.~\eqref{eqn:training}, we have $\frac{1}{n}\sum_{t=1}^nJ_H^{\pi_{\lambda}^{(n)}}(y_{0:H}^{(t)})\leq \frac{1}{n}\sum_{t=1}^nJ_H^{\pi_{\lambda}^{*}}(y_{0:H}^{(t)})$. Thus, we have
\begin{equation}
\begin{split}
\mathbb{E}\left[J_H^{\pi_{\lambda}^{(n)}}\right]&\leq \frac{1}{n}\sum_{t=1}^nJ_H^{\pi_{\lambda}^{*}}(y_{0:H}^{(t)})+ 4HP\sqrt{\frac{2}{n}\ln\frac{4N(\epsilon, \Pi_{\lambda}, \hat{L}_1^n)}{\delta}}\\
&\leq \mathbb{E}\left[J_H^{\pi_{\lambda}^{*}}\right]+8HP\sqrt{\frac{2}{n}\ln\frac{4N(\epsilon, \Pi_{\lambda}, \hat{L}_1^n)}{\delta}},
\end{split}
\end{equation}
where the last inequality holds be applying the generalization theorem in \cite{statistical_learning_bousquet2004introduction}.
By Eqn.\eqref{eqn:expected_cost_bound}, we have
\begin{equation}
\begin{split}
\mathbb{E}\left[J_H^{\pi_{\lambda}^{(n)}}\right]
&\leq \mathbb{E}\left[J_H^{\pi^{*}}\right]+B\mathbb{E}\left[\sum_{h=0}^{H-1} \left[\delta_h-(\sqrt{1+\lambda}-1)^2Gr_h^{\dagger}\right]^+ \right]\\
&+\mathcal{O}\left(\sqrt{\frac{1}{n}\ln\frac{N(\epsilon, \Pi_{\lambda}, \hat{L}_1^n)}{\delta}}\right),
\end{split}
\end{equation}
where $\mathcal{O}$ notation indicates the increasing with episode length $H$ and maximum loss value $P$
\end{proof}

\end{document}